\documentclass[prd,nofootinbib,floats,aps,twocolumn,tightenlines,superscriptaddress]
{revtex4-2}

\usepackage[utf8]{inputenc}

\usepackage{float}
\usepackage{soul}
\usepackage{cancel}
\usepackage{subfigure}
\usepackage{amssymb,amsmath,amsfonts,amsthm,epsfig,epstopdf,array,mathrsfs}
\usepackage{braket}
\usepackage{graphicx}
\usepackage{dcolumn}
\usepackage{ marvosym }

\usepackage{subfigure}

\usepackage{bm}
\usepackage{hyperref}
\usepackage{cleveref}
\usepackage{tikz}
\usepackage{comment}
\usepackage[normalem]{ulem}

\def\bea{\begin{eqnarray}}
\def\eea{\end{eqnarray}}
\def\be{\begin{equation}}
\def\ee{\end{equation}}

\theoremstyle{definition}

\newtheorem{defn}{Definition}
\newtheorem{prop}{Proposition}

\newtheorem{exmp}{Example}
\def\bea{\begin{eqnarray}}
\def\eea{\end{eqnarray}}
\def\be{\begin{equation}}
\def\ee{\end{equation}}

\newsavebox\mybox

\definecolor{mrainbow1}{rgb}{0.86, 0.13, 0.13}
\definecolor{mrainbow2}{rgb}{0.89,0.60,0.22}
\definecolor{mrainbow3}{rgb}{0.67, 0.74, 0.32}
\definecolor{mrainbow4}{rgb}{0.39, 0.67, 0.60}
\definecolor{mrainbow5}{rgb}{0.25, 0.39, 0.81}
\definecolor{mrainbow6}{rgb}{0.47, 0.11, 0.53}
\definecolor{b1}{rgb}{0.27,0.49,0.80}
\definecolor{g2}{rgb}{0.51,0.73,0.44}
\definecolor{y3}{rgb}{0.86,0.67,0.24}
\definecolor{r4}{rgb}{0.86, 0.13, 0.13}
\definecolor{l5}{rgb}{0.51, 0.03,0.77}

\begin{document}

\title{Entanglement and correlations  between local observables in de Sitter spacetime}
\author{Patricia Ribes-Metidieri}
\email{patricia.ribesmetidieri@york.ac.uk}
\affiliation{Department of Mathematics, University of York, Heslington, York YO10 5DD, UK }
\affiliation{Institute for Mathematics, Astrophysics and Particle Physics, Radboud University, 6525 AJ Nijmegen, The Netherlands}

\author{Ivan Agullo}
\email{agullo@lsu.edu}
\affiliation{Department of Physics and Astronomy, Louisiana State University, Baton Rouge, LA 70803, USA}

\author{B\'eatrice Bonga}
\email{bbonga@science.ru.nl}
\affiliation{Institute for Mathematics, Astrophysics and Particle Physics, Radboud University, 6525 AJ Nijmegen, The Netherlands}
\affiliation{Theoretical Sciences Visiting Program, Okinawa Institute of Science and
Technology Graduate University, Onna, 904-0495, Japan}

\begin{abstract}

Studies of quantum field entanglement in de Sitter space based on the von Neumann  entropy of local patches have concluded that curvature enhances entanglement between regions and their complements. Similar conclusions about entanglement enhancement have been reached in analyses of Fourier modes in the cosmological patch of de Sitter space. We challenge this interpretation by adopting a fully local approach: examining entanglement between pairs of field modes compactly supported within de Sitter's cosmological patch.  Our approach is formulated in terms of the properties of a metric tensor and an associated complex structure induced by the Bunch–Davies vacuum on the classical phase space. We find that increasing curvature increases correlations between local modes but,  somewhat counterintuitively, \emph{decreases} their entanglement. Our methods allow us to characterize how entanglement is spatially distributed, revealing that a cosmological constant, even if tiny, qualitatively alters the vacuum's entanglement structure.  We show our results are compatible with previous entropy-based studies when properly interpreted. Our findings have implications for entanglement between observables generated during cosmic inflation.\\
  
\end{abstract}

\maketitle

\section{Introduction}

Investigations in relativistic quantum information have revealed that even the simplest states in the most basic quantum field theories possess an extraordinarily rich and complex entanglement content (see e.g. \cite{Witten:2018,Hollands:2017dov} and references therein). How this entanglement is distributed across space and how it is shaped by the curvature and symmetries of spacetime remains, however, an open question. This article aims to contribute to these issues.

De Sitter spacetime serves as a particularly suitable theoretical laboratory for studying these foundational questions. It is more than just a theoretical laboratory: the spacetime during the inflationary epoch of our universe, where primordial density perturbations are believed to have originated, is approximately isometric to a portion of de Sitter space. Given both its mathematical appeal and its physical relevance, this article focuses on the cosmological patch of de Sitter spacetime, also known as the Poincaré patch.

The cosmological patch of de Sitter admits seven independent Killing vector fields---the subset of the de Sitter isometries that leave the patch invariant \cite{Hawking:1973uf,abk1}. These symmetries single out a preferred Fock representation of the quantum field and a distinguished vacuum state, known as the Bunch-Davies vacuum \cite{Bunch:1978yq}. Characterizing the entanglement structure of this state offers a  way to probe the relationship between symmetry, curvature, and entanglement in quantum field theory. 

It is well-known that curvature and symmetries dictate the form of the vacuum two-point function. For example, while two-point correlations decay as the inverse square of the distance for massless fields in flat spacetime, they become nearly scale-invariant in the low-mass limit in de Sitter space. From this perspective, it is natural to expect that entanglement might behave similarly. However, as we will discuss in later sections, entanglement presents important subtleties. One such subtlety we find is the following. The near scale-invariance of two-point correlations in de Sitter implies that correlations at super-Hubble distances are significantly stronger than in Minkowski spacetime. These enhanced correlations play a crucial role in cosmology, as they seed the observed correlations in temperature fluctuations in the cosmic microwave background. There is a widespread intuition that the stronger correlations generated during inflation are accompanied by a corresponding increase in entanglement. One of the goals of this article is to refine this intuition. While entanglement always implies the presence of correlations, the converse is not true: physical systems can exhibit correlations without being entangled. We show that the correlations generated during inflation---those that later become accessible to cosmological observers---do not contain more entanglement than those found in Minkowski spacetime. In more direct terms, we prove that increasing the Hubble rate $H$ of the cosmological patch of de Sitter space---which corresponds to increasing the Ricci curvature---enhances the correlations between localized field modes, but at the same time---and somewhat counterintuitively---\emph{reduces} the entanglement between them.

Beyond its implications for cosmology, this result emphasizes the need to distinguish sharply between entanglement and correlation in quantum field theory. The intuitive link between the two often stems from our experience with {\em pure} states, where all correlations between a subsystem and its complement are necessarily due to entanglement. In field theory, however, the situation is more subtle. In particular, the reduced state of \emph{any} field mode that is localized in space—i.e., has compact spatial support—is necessarily mixed (this is a direct consequence of the Reeh-Schlieder theorem \cite{reehschlieder}). For mixed states, there is no straightforward relationship between correlation and entanglement.

The entanglement content of de Sitter space has been extensively explored in the literature \cite{Grishchuk:1990bj,Albrecht:1992kf,Lesgourgues:1996jc,Kiefer:1998pb,Kiefer:1998qe,Kiefer:2008ku,Polarski:1995jg,Martin:2015qta,ack,Micheli:2022tld,Brahma:2023lqm,Brahma:2023uab,Brahma:2023hki,Brahma:2024yor,Bhattacharyya:2024duw,Martin_2022,Calzetta_1995,maldacena_entanglement_2013,Maldacena:2015bha,Martin:2015qta,Nelson_2016,Martin_2018,Grain:2019vnq,Brahma:2020zpk,Brahma:2021mng,Martin:2021qkg,Agullo:2022ttg,Micheli:2022tld,Espinosa-Portales:2022yok,Bhardwaj:2023squ,Micheli:2022tld,Martin:2021xml,Martin:2021qkg,K:2023oon,PhysRevD.78.044023,PhysRevD.80.124031,Nambu_2023,BELFIGLIO20251}, raising the natural question of what new insights this article seeks to provide. Previous studies have approached the problem from diverse perspectives and using a variety of tools, each with its own strengths and limitations---and sometimes producing results that appear in tension with one another. One prominent approach involves quantifying the entanglement between a region and its complement using entanglement entropy \cite{maldacena_entanglement_2013}. This method has led to interesting findings, suggesting that the entire de Sitter space contains more entanglement than Minkowski space---a result we independently confirm using a different methodology. 

Another  approach is entanglement harvesting \cite{Valentini1991,reznik1,reznik2,Pozas-Kerstjens:2015,Pozas2016,Salton:2014jaa,Ng2014,mutualInfoBH,freefall,HarvestingSuperposed,Henderson2019,bandlimitedHarv2020,ampEntBH2020,carol,boris,ericksonWhen,threeHarvesting2022,twist2022,cisco2023harvesting,SchwarzchildHarvestingWellDone}, in which two spatially separated particle detectors become entangled solely through their local interactions with the field---i.e., without ever directly interacting with each other.  Since the detectors remain causally disconnected and are initially unentangled, any entanglement between them must originate from pre-existing entanglement in the field, swapped to the detectors via the local field-detector interaction. Interestingly, it has been shown that two comoving geodesic detectors in de Sitter space harvest {\em less} entanglement than they would in Minkowski spacetime \cite{Steeg_2009}. This has been attributed to the thermal nature of the Bunch-Davies vacuum as perceived by local observers. Such findings suggest that less entanglement is accessible in de Sitter space, seemingly at odds with the implications of entanglement entropy calculations.

In this article, we pursue a different yet complementary approach to the study of quantum field theoretic entanglement, one that carries its own set of advantages and limitations. Our methods build on previous work \cite{Martin:2021qkg,K:2023oon,PhysRevD.80.124031,PhysRevD.78.044023,bianchi_entropy_2019,Martin:2021xml,ubiquitous}. It focuses on finite sets of field modes and analyzes entanglement among them using tools from Gaussian quantum information theory \cite{serafini2017quantum,Weedbrook_2012}. The focus on finitely many modes has the significant advantage of avoiding ultraviolet divergences, as the operator algebras involved are of type I. 
This approach entails a choice of the set of field modes to be analyzed---just as the entanglement harvesting results depend on the specific response function of the detectors, which effectively selects the modes of the field that couple to them. However, since our interest lies in general features of the field theory rather than in the peculiarities of any specific set of modes, we identify properties that are generic across any selection of modes. Only in this way can we make invariant statements about de Sitter spacetime. 

Our tools have a more geometric flavor than some of the previous works, as the main results of this article are derived from the properties of a metric tensor and a complex structure on the classical phase space. One of the goals of this work is to highlight the power and elegance of the geometric framework we employ.   

The primary limitation of our approach stems from its reliance on Gaussian tools, making it applicable to Gaussian states. Nevertheless, this regime is highly relevant for cosmology, as the Bunch-Davies vacuum is a Gaussian state, and nonlinear effects are negligible during inflation---confirmed by the absence of measurable non-Gaussianities in the cosmic microwave background \cite{refId0}. As is standard in quantum field theory, non-Gaussian corrections could in principle be incorporated perturbatively, though such extensions lie beyond the scope of the present work.

The analysis presented here uncovers structural aspects of entanglement in the cosmological patch of de Sitter spacetime that, to our knowledge, have not been previously explored—although some elements of our conclusions were anticipated in \cite{Nambu_2023}, based on the study of a restricted family of non-compactly supported field modes. In particular, by employing the concept of {\em partner modes} \cite{hotta2015partner,Trevison_2019,hackl_minimal_2019,partnerformula}, we are able to characterize the subtle manner in which entanglement is distributed in the Bunch-Davies vacuum, revealing important differences from Minkowski spacetime. This, in turn, clarifies that the apparent tension between results from entanglement entropy and entanglement harvesting protocols is not a contradiction. On the contrary, we demonstrate that these findings are closely connected, and in fact, one is a consequence of the other.

For a summary of some of the results presented here, see~\cite{ABRM_deSitter_short}.
The rest of the paper is organized as follows. In Sec.~\ref{sec:framework}, we summarize the tools that will be used in the rest of the paper. Sec.~\ref{sec:dSqft} applies these tools to the symmetric vacuum in the cosmological patch of de~Sitter spacetime, and discusses some important features of correlation functions together with a few illustrative examples. In Sec.~\ref{sec:vNEntropy}, we analyze the von Neumann entropy of a single field degree of freedom. Sections~\ref{sec:correlations} and~\ref{sec:LN2dof} analyze mutual information and entanglement, respectively,  between pairs of compactly supported modes of the field. Finally, in 
 Sec.~\ref{sec:partners} we focus on the notion of partners: the system that codifies the correlations between a local mode and the rest of the field.  We use this notion to uncover the distribution of entanglement between different local modes  of a scalar field in the Bunch-Davies vacuum in Sec.~\ref{subsec:where_is_entanglement}. Finally, Sec.~\ref{sec:discussion} summarizes our findings and discusses their implications for curved spacetime quantum field theory and observational signatures of entanglement in the early universe. Throughout this paper, we use units in which $\hbar=c=1$ and  adopt the following convention for the  Fourier transform  $\tilde{h}(\vec k) = \int d^3 x \, h(\vec{x}) e^{i \vec k \cdot \vec x}$.

\section{General Framework and tools\label{sec:framework}}

This section provides a summary of the tools and concepts on which the rest of this article is based. Specifically, it introduces methods to define and characterize finite-dimensional subsystems in linear quantum field theories, Gaussian states, and correlations and entanglement between subsystems. For illustrative purposes, we will often use the example of a massless field in Minkowski spacetime in this section. This example will serve as a reference case for comparison with a light field in de Sitter spacetime, which will be discussed in subsequent sections.

\subsection{Smeared operators}

Consider a scalar field minimally coupled to curvature and satisfying the Klein--Gordon equation  
\begin{equation}
    (\Box + m^2) \hat{\phi}(x) = 0\,,
\end{equation}
on a globally hyperbolic spacetime with metric tensor $g_{ab}$. (The tools described here generalized straightforwardly to other linear bosonic field theories.) The differential operator $\Box$ denotes the d'Alembertian associated with $g_{ab}$, and $m$ is the mass of the field.  

The object $\hat{\phi}(x)$ does not define an operator in any reasonable way (see e.g. \cite{Wald:1995yp}). This is obvious, for instance, by noticing that acting on the Fock vacuum produces a state with infinite norm, $\langle 0|\hat{\phi}(x) \hat{\phi}(x) |0\rangle\to \infty$, thus $\hat{\phi}(x) |0\rangle$ does not lie within the Hilbert space.  
Well-defined operators can be constructed by integrating $\hat{\phi}(x)$ against suitably chosen real functions (also known as smearing): 
\begin{equation} \label{covopt}
    \hat{\Phi}(F) = \int  dV  \, F(x) \, \hat{\phi}(x) .
\end{equation}
where $dV$ is the spacetime volume element.
In simple terms, the field at a given point, $\hat{\phi}(x)$, is too singular to define an operator, but this singularity is tamed by ``smearing'' the field using functions $F(x)$ in spacetime.  

For the purposes of this article---particularly for comparing aspects of the cosmological patch of de Sitter and Minkowski spacetimes---it will be more convenient to work in the canonical formalism. The canonical picture requires the introduction of a foliation of spacetime into a one-parameter family of spatial Cauchy hypersurfaces $\Sigma_t$. For a fixed value of $t$, we can define  
\begin{align}
    \hat{\Phi}(\vec{x}) &:=\hat \phi(x)|_{t}\, , \nonumber \\
    \hat{\Pi}(\vec{x}) &:= \left(\sqrt{h} \, n^a \nabla_a \hat \phi(x)\right)|_{t}\, , \nonumber
\end{align}
where $n^a$ is the future-oriented unit normal to $\Sigma_t$, $\nabla_a$ is the covariant derivative associated with $g_{ab}$, and $\sqrt{h}$ is the determinant of the metric induced on $\Sigma_t$ by $g_{ab}$.  

The operators $\hat{\Phi}(\vec{x})$ and $\hat{\Pi}(\vec{x})$ satisfy the familiar equal-time commutation relations  
\begin{equation} \label{cancom}
    [\hat{\Phi}(\vec{x}), \hat{\Pi}(\vec{x}')] = i \, \delta^{(3)}(\vec{x} - \vec{x}')\, .
\end{equation}
As in the covariant formalism, the canonical operators $\hat{\Phi}(\vec{x})$ and $\hat{\Pi}(\vec{x})$ must be interpreted  as operator-valued distributions. 

Operators that are linear in field and momentum  are obtained as linear combinations of  the form  
\begin{equation} \label{smeared}
    \hat{\Phi}(f) -\hat{\Pi}(g) \, ,
\end{equation}
where $f(\vec{x})$ and $g(\vec{x})$ are real functions on $\Sigma_t$, and 
\begin{equation}
    \hat{\Phi}(f) := \int_{\Sigma_t} d^3x \, f(\vec{x}) \, \hat{\Phi}(\vec x) \, .
\end{equation}
Similarly for $\hat{\Pi}(g, t)$. The minus sign in the second term in \eqref{smeared} has been introduced merely for convenience, as will become evident soon. For the linear field theory under consideration, there exists a simple relation between the covariant operators $\hat{\Phi}(F)$ in \eqref{covopt} and the canonical ones of the form \eqref{smeared} (see, e.g., \cite{Wald:1995yp}).  

There is a convenient way of organizing all smeared operators in the canonical formalism. This organization is based on identifying each pair of functions $(g, f)$ used in \eqref{smeared} with elements of the classical phase space $\Gamma$.   The classical phase space of our linear field theory consists of pairs of functions $(g(\vec{x}), f(\vec{x})) =: \bm{\gamma}$ (defined on an abstract 3-manifold, which can be identified with $\Sigma_t$ for any $t$) equipped with the symplectic structure 
\begin{equation}
   \bm{\omega}(\bm{\gamma}_1, \bm{\gamma}_2) = \int_{\Sigma_t} d^3x \, ( f_1 g_2 - g_1 f_2 ) \, .
\end{equation}
Let us define the vector $ \hat{\bm R}:=(\hat{\Phi}(\vec{x}),\hat{\Pi}(\vec{x}))$. 
Then, for each element $\bm{\gamma} \in \Gamma$, we can define a smeared canonical operator at time $t$ as the symplectic product of $\bm{\gamma}$ and $\hat{\bm R}$:  
\bea
    \hat{O}_{\bm{\gamma}} &:=& \bm{\omega}(\bm{\gamma}, \hat{\bm R}) =\int_{\Sigma_t}\! d^3x \! \left[ f(\vec{x}) \hat{\Phi}(\vec{x}) - g(\vec{x}) \hat{\Pi}(\vec{x}) \right]\nonumber \\ &=& \hat{\Phi}(f) -\hat{\Pi}(g)  .
\eea
This organization of linear operators---assigning one such operator to each element of the classical phase space---may seem complicated at first, but it offers more benefits than hassle. In particular, it will allow us to describe quantum Gaussian states in terms of a co-vector and a metric tensor in the classical phase space. 

Using the canonical commutation relation \eqref{cancom}, the commutator algebra of two linear operators $\hat{O}_{\bm{\gamma}}$ and $\hat{O}_{\bm{\gamma}'}$ is simply determined from  the symplectic product of  $\bm{\gamma}$ and $\bm{\gamma}'$ as follows:  
\begin{equation}
    [\hat{O}_{\bm{\gamma}}, \hat{O}_{\bm{\gamma}'}] = - i \, \bm{\omega}(\bm{\bm{\gamma}}, \bm{\bm{\gamma}}') \, .
\end{equation}
Recall that the specification of the classical phase space $\Gamma$ in field theory requires a choice of the class of allowed functions. In Minkowski spacetime, it is common to choose $\Gamma$ as made of functions in Schwartz's space (see Appendix \ref{app:sobolev} for its definition). Smooth functions of compact support are also a typical choice, particularly in curved spacetimes where Schwartz's space does not generalize. We can start with either option, as it will not affect our discussion. This is so because the construction of the quantum Fock space requires an enlargement of $\Gamma$ (more precisely, a Cauchy-completion), and the final enlarged phase space will be the same regardless of whether we start from Schwartz's space or smooth functions of compact support.

\subsection{Gaussian states and covariance metric}

Gaussian states are completely determined by their first and second moments in the following sense. A generic quantum state, Gaussian or not, is uniquely characterized by all its $n$-th moments (the expectation value of $n$ linear operators). For a Gaussian state, the expectation value of the product of $n$ centered operators, $\langle \hat{\overline{O}}_{\bm{\gamma}_1} \cdots \hat{\overline{O}}_{\bm{\gamma}_n} \rangle$, vanishes if $n$ is odd, and is given by the following combination of second moments if $n$ is even:
\begin{equation} \label{Gauss}
\langle \hat{\overline{O}}_{\!\bm{\gamma}_1}\! \!\cdots \hat{\overline{O}}_{\!\bm{\gamma}_n} \rangle \!=\! \frac{1}{n!}\!\! \sum_\pi \langle \!\hat{\overline{O}}_{\!\bm{\gamma}_{\!\pi\!(\!1\!)}} \!\hat{\overline{O}}_{\!\bm{\gamma}_{\!\pi\!(\!2\!)}} \rangle \! \cdots \!\langle \hat{\overline{O}}_{\!\bm{\gamma}_{\!\pi\!(\!n\!-\!1\!)}} \hat{\overline{O}}_{\!\bm{\gamma}_{\!\pi\!(\!n\!)}} \rangle,
\end{equation}
where the sum extends over all permutations $\pi$ of $n$ elements. Centered operators are defined as $\hat{\overline{O}}_{\bm{\gamma}} := \hat{O}_{\bm{\gamma}} - \langle \hat{O}_{\bm{\gamma}} \rangle$.  
Property \eqref{Gauss} can be understood as the defining property of Gaussian states.

It follows that all $n$-th moments of a Gaussian state can be obtained from the first and second moments, $ \langle \hat{O}_{\bm{\gamma}} \rangle$ and $\langle \hat{O}_{\bm{\gamma}} \hat{O}_{\bm{\gamma}'} \rangle$, $\forall \bm{\gamma},\bm{\gamma}'\in\Gamma$. Furthermore, the second moments can be decomposed into their symmetric and antisymmetric parts:
\begin{equation}
\langle \hat{O}_{\bm{\gamma}} \hat{O}_{\bm{\gamma}'} \rangle =  \frac{1}{2}\Big(\langle \{\hat{O}_{\bm{\gamma}}, \hat{O}_{\bm{\gamma}'}\} \rangle + \langle [\hat{O}_{\bm{\gamma}}, \hat{O}_{\bm{\gamma}'}] \rangle\Big )\,,
\end{equation}
where curly brackets denote the anti-commutator. Since the commutator is proportional to the identity operator for all $\bm{\gamma}, \bm{\gamma}'$, its expectation value is state-independent. Therefore, all information about a Gaussian state is encoded in the first moments and the symmetric part of the second moments.

This information can be used to organize Gaussian states as follows. The first moments, $\langle \hat{O}_{\bm{\gamma}} \rangle$, can be viewed as a linear map from the classical phase space to the reals, defining a covector in $\Gamma$:
\begin{equation}
\mu(\bm{\gamma}) := \langle \hat{O}_{\bm{\gamma}} \rangle.
\end{equation}
I.e., $\mu$ is the co-vector in $\Gamma$ whose action on ${\bm \gamma}\in \Gamma$ equals to the expectation value of $ \hat{O}_{\bm{\gamma}}$ in the state under consideration. 

Similarly, the symmetrized second moments define a symmetric bilinear map on $\Gamma$ denoted as $\sigma$:
\begin{equation}
\sigma(\bm{\gamma}, \bm{\gamma}') := \langle \{ \hat{\overline O}_{\bm{\gamma}}, \hat{\overline  O}_{\bm{\gamma}'} \} \rangle\,.
\end{equation}
It is straightforward to show that $\sigma$ defined in this way is positive definite. That is, a quantum state defines a {\em metric tensor}  $\sigma$ in the classical phase space.  We will refer to $\sigma$ as the {\em covariance metric} of the quantum state. Its matrix elements in a basis are commonly referred to as the covariance matrix.

While every quantum state defines a covector and a metric on the classical phase space, it is only for Gaussian states that these two objects completely and uniquely characterize the state. This fact makes it possible to reformulate any quantum calculation involving Gaussian states as geometric operations in the classical phase space, offering significant computational and conceptual advantages, which we exploit throughout this article.

As mentioned above, $\sigma$ is a covariant rank-2 tensor. To uncover certain aspects of Gaussian states more transparently, it is useful to raise one index of $\sigma$ using the inverse of the symplectic structure.\footnote{In field theory, the symplectic structure is generally a weakly degenerate two-form and may not admit an inverse. However, an inverse can be uniquely defined once the phase space $\Gamma$ is Cauchy completed using the inner product defined by $\sigma$.} Let $\Omega$ denote the inverse of the symplectic structure $\omega$ on $\Gamma_{\sigma}$—the Cauchy completion of $\Gamma$ with respect to $\sigma$. By combining $\sigma$ and $\Omega$, one obtains a linear map on $\Gamma_{\sigma}$:
\begin{equation}
J^{\alpha}_{\ \beta} = -\Omega^{\alpha\gamma} \sigma_{\gamma\beta} .
\end{equation}
To simplify the notation, we have used an extended index notation, where Greek indices denote tensor indices in the infinite-dimensional vector space $\Gamma_{\sigma}$, i.e., they include the $\vec{x}$ dependence. More explicitly, the previous equation reads  
\[
J(\vec{x}, \vec{x}') = -\int d^3x''\, \Omega(\vec{x}, \vec{x}'') \sigma(\vec{x}'', \vec{x}').
\]  
It is not difficult to prove that a Gaussian state is pure if and only if $J$ satisfies $J^2 = -\mathbb{I}$. Such linear maps are known as {\em complex structures}. In other words, pure Gaussian states define complex structures in $\Gamma_{\sigma}$, while mixed Gaussian states ---which satisfy $J^2 < -\mathbb{I}$--- define {\em restricted complex structures}. One can check that $\sigma$ and $J$ defined in this way satisfy certain compatibility conditions with the symplectic structure $\omega$---these conditions are related to the mathematical concept of K\"ahler vector spaces (see, e.g., \cite{Ashtekar:1975zn, Hackl:2020ken, partnerformula} for further details).

\subsubsection{Massless scalar field in Minkowski spacetime}

In order to show these tools in practice, let us apply them to a familiar example describing the usual Poincaré invariant vacuum state of a massless, real scalar field propagating in Minkowski spacetime.

For the classical phase space $\Gamma$, we take the Schwartz space (see Appendix~\ref{app:sobolev} for its definition).
 The Minkowski vacuum is a Gaussian state with
\be
\mu_M(\vec{x})=\left( \langle \hat{\Pi}(\vec{x})\rangle, \langle \hat{\Phi}(\vec{x})\rangle \right)=(0,0)\, ,
\label{eq:mu-M}
\ee
and covariance metric
\be
\sigma_M(\vec{x},\vec{x}')\!=\!\begin{pmatrix} \langle \{ \hat{\Pi}(\vec{x}),\hat{\Pi}(\vec{x}')\}\rangle & \langle \{\hat{\Phi}(\vec{x}') ,\hat\Pi(\vec{x})\}\rangle \\
\langle\{\hat{\Phi}(\vec{x}) ,\hat\Pi(\vec{x}')\}\rangle & \langle \{\hat{\Phi}(\vec{x}), \hat \Phi(\vec{x}')\}\rangle
\end{pmatrix},
\label{eq:sigma-M}
\ee where the expectation values in Eqs.~\eqref{eq:mu-M}-\eqref{eq:sigma-M} are evaluated in the vacuum $\ket{0}$. 
The components of $\sigma(\vec{x},\vec{x}')$ are bi-distributions which, for the Minkowski vacuum,  act 
on test functions as follows 
\bea \label{synmcorrMink}
\langle 0|\{\hat{\Phi}[f], \hat\Phi[f']\}|0\rangle&=&\text{Re}(f,f')_{-1/2}\, , 
\nonumber \\ 
\langle 0|\{\hat{\Pi}[g], \hat\Pi[g']\}|0\rangle&=&\text{Re}(g,g')_{1/2}\, ,\nonumber \\
\langle 0|\{\hat{\Phi}[f], \hat \Pi[g']\}|0\rangle&=&0 \, . \eea
where $(f,g)_{s}$ denotes a Sobolev product of order $s$, defined in Fourier space as
\[
(f,g)_{s} \equiv \int_{\mathbb{R}^3} \frac{d^3k}{(2\pi)^3} \, |\vec{k}|^{2s} \tilde{f}(\vec{k}) \, \tilde{g}^*(\vec{k})\, .
\] 
Expressions \eqref{synmcorrMink} can be readily checked using the standard expansion of the field and momentum operators in terms of creation and annihilation operators (see e.g. Appendix B in \cite{ubiquitous}).

Using tools from asymptotic harmonic analysis, it is relatively straightforward to determine the large-separation behavior of the Sobolev products in Eq.~\eqref{synmcorrMink}. 
For instance, for ${\bm \gamma} = (g, f)$ and ${\bm \gamma}' = (g', f')$ supported on non-overlapping spherical regions with center-to-center separation $|\Delta \vec x|$, we find that
\[\mathrm{Re}(f,f')_{-\frac{1}{2}} \sim (|\Delta \vec x|)^{-2},\]
and 
\[\mathrm{Re}(g,g')_{\frac{1}{2}} \sim (|\Delta \vec x|)^{-4}, 
\]
(Appendix~\ref{app:asymptotic_sobolev} contains further details of the derivation of these asymptotic properties). This reproduces the familiar fall-off of field-field and momentum-momentum correlations in the Minkowski vacuum. 
 
The action of $\sigma_M$ on any pair ${\bm\gamma}=(g,f),{\bm\gamma'}=(g',f')\in \Gamma$ can therefore be expressed as
\be \label{sigmaMink} \begin{split}
\sigma_M({\bm\gamma},{\bm\gamma'})&=(g,f)\begin{pmatrix} \mathrm{Re}(\cdot,\cdot)_{1/2} & 0 \\
0 & \mathrm{Re}(\cdot,\cdot)_{-1/2}\end{pmatrix}\begin{pmatrix} g'\\f'\end{pmatrix}\\
&=\mathrm{Re}(g,g')_{1/2}+\mathrm{Re}(f,f')_{-1/2}\, .
\end{split} \ee

As described above, $\sigma_M$ defines a metric in $\Gamma$, which can be used to complete it. The resulting space, $\Gamma_{\sigma}$, consists of all pairs $(g,f)$ of functions in $\mathbb{R}^3$ which have finite Sobolev norm of order $s=1/2$ and $s=-1/2$, respectively (see, e.g. \cite{Much:2018ehc}). The set of functions with finite Sobolev norm of order $s$ defines the so-called Sobolev space of order $s$, denoted as $\dot{H}_s(\mathbb{R}^3)$  (see Appendix \ref{app:sobolev} for a further details about these spaces). Hence, 
\[
\Gamma_{\sigma} = \dot{H}_{1/2}(\mathbb{R}^3) \times \dot{H}_{-1/2}(\mathbb{R}^3).
\]

It is worth mentioning that, while all functions in $\dot{H}_{1/2}(\mathbb{R}^3)$ are continuous (although not necessarily differentiable), there are functions in $\dot H_{-1/2}(\mathbb{R}^3)$ that are discontinuous. For instance, the top-hat function belongs to $\dot H_{-1/2}(\mathbb{R}^3)$, but not to $\dot H_{1/2}(\mathbb{R}^3)$, and hence it can be used to smear the field operator $\hat{\Phi}(\vec{x})$ (but not the momentum $\hat{\Pi}(\vec{x})$). 

The space $\Gamma_{\sigma}$ contains the Schwartz space as a dense subspace as well as all smooth functions of compact support.

Next, using $\sigma_M$, one obtains
\bea \label{JM} J_M(\vec x, \vec x')&=&-\int d^3x''\, \Omega(\vec x,\vec x'') \sigma_M(\vec x'',\vec x')\nonumber \\
&=& \int \frac{d^3k}{(2\pi)^3} e^{i\vec k(\vec x-\vec x')}\,\begin{pmatrix} 0 & -\frac{1}{k}\\k& 0\end{pmatrix}, \nonumber \eea
where we have used that the bi-distribution $\sigma_M(\vec x,\vec x')$ can be written as
\be \sigma_M(\vec x,\vec x')=\int \frac{d^3k}{(2\pi)^3} e^{i\vec k(\vec x''-\vec x')}\,\begin{pmatrix} k & 0\\0& \frac{1}{k}\end{pmatrix} \, . \ee
From this, it is easy to verify that
\[
J_M^2(\vec{x}, \vec{x}')=-\begin{pmatrix} 1 & 0\\0& 1\end{pmatrix}\, \delta^{(3)}(\vec{x}-\vec{x}')\, ,
\]
confirming that the Minkowski vacuum is a pure state.

\subsection{Gaussian subsystems}
The organization of linear quantum observables $\hat{O}_{\gamma}$ in terms of vectors ${\bm \gamma}$ in the classical phase space allows one to define subsystems in the quantum theory via the concept of classical subsystems.

In the classical theory, a subsystem is characterized by a {\em symplectic subspace} of the classical phase space, $\Gamma_A \subset \Gamma_{\sigma}$. This definition is basis-independent, meaning it does not require specifying any coordinates in $\Gamma_{\sigma}$. Using the relation $\gamma \to \hat{O}_{\gamma}$, the symplectic subspace $\Gamma_{A}$ defines a subalgebra of quantum observables ---generated by products of operators $\hat{O}_{\gamma}$ in the standard manner--- with $\gamma \in \Gamma_A$. In the algebraic approach to quantum field theory, this subalgebra defines a quantum subsystem. This establishes the relation between classical subsystems and their quantum counterpart.

If $\Gamma_{A}$ has finite dimension $2N_A$, the associated subalgebra is isomorphic to the algebra generated by the positions and momenta of $N_A$ quantum harmonic oscillators. The associated Weyl algebra is a Type I von Neumann sub-algebra of the quantum field theory \cite{haag_algebraic_1964}. 

If $\hat{\rho}$ is a Gaussian state in the field theory, its restriction to subsystem $A$ defines the reduced state $\hat{\rho}^{\rm red}_A$. This reduced state is also Gaussian ---as its $n$th-point functions, being a subset of the $n$th-point functions of the Gaussian state $\hat{\rho}$, automatically satisfy equations \eqref{Gauss}. 

Following the discussion in the previous subsection, the Gaussian state $\hat{\rho}^{\rm red}_A$ can be uniquely characterized by the pair $(\mu_A, \sigma_A)$, defined as the restriction to $\Gamma_A$ of the pair $(\mu, \sigma)$ that defines the state $\hat{\rho}$:
\be 
\mu_A := \mu|_{\Gamma_A}, \quad \sigma_A := \sigma|_{\Gamma_A}.
\ee
In general, the Gaussian state $(\mu_A, \sigma_A)$ is mixed, even when the parent state $(\mu, \sigma)$ is pure. The calculation of $(\mu_A, \sigma_A)$ from $(\mu, \sigma)$ is straightforward, as it only involves restricting the action of the latter pair to $\Gamma_A$. We illustrate this calculation with the following example.

\subsubsection{Example: Single-mode Subsystem in a Massless Scalar Theory in Minkowski Spacetime}
\label{example-ball-Minkowski}

This is a continuation of the example introduced in the previous subsection, where a massless, real scalar field is prepared in the Minkowski vacuum $|0\rangle$.

The smallest possible subsystem corresponds to choosing a two-dimensional symplectic subspace $\Gamma_A \subset \Gamma_{\sigma}$. Though $\Gamma_A$ exists independently of any choice of basis, it is convenient to define it by explicitly constructing a basis. So let us consider the following two vectors in $\Gamma_{\sigma}$:
{ \begin{equation}\label{eq:ex1vNS}
    {\bm \gamma}^{(1)} (\vec{x}) = \left(\begin{matrix} 0\\ f^{(\delta)} (\vec{x}) \end{matrix}\right), \quad {\bm \gamma}^{(2)} = \left(\begin{matrix} - g^{(\delta)} (\vec{x})\\ 0 \end{matrix}\right),
\end{equation}
with $f^{(\delta)}(\vec{x})$ defined as the (non-negative, spherically symmetric) function:
\bea\label{eq:fdelta_family}
    f^{(\delta)} (\vec{x})\! &=& \!A_{\delta} \left( 1 - \frac{|\vec{x}|^2}{R^2} \right)^{\delta} \!\Theta(R - r), \\ \nonumber g^{(\delta)} (\vec{x}) \!&=&\!R\,f^{(\delta)} (\vec{x}). 
\eea
Here, $\delta$ is a positive real parameter, kept general for flexibility. $\Theta(x)$ is the Heaviside step function, which makes $f^{(\delta)}(\vec{x})$ compactly supported within a ball of radius $R$, and
\begin{equation}\label{eq:Adelta}
    A_{\delta} = { R^{-2}} \pi^{-3/4} \sqrt{\frac{\Gamma(5/2 + 2\delta)}{\Gamma(1+2\delta)}}\,,
\end{equation}
is a normalization constant, making ${\bm \omega}({\bm \gamma}^{(2)},{\bm \gamma}^{(1)})=1$.}
Figure~\ref{fig:W_delta_vs_r} shows the shape of $f^{(\delta)}(\vec{x})$ for different values of $\delta$.

\begin{figure}

    \begin{flushleft}
    \includegraphics[width=0.45\textwidth]{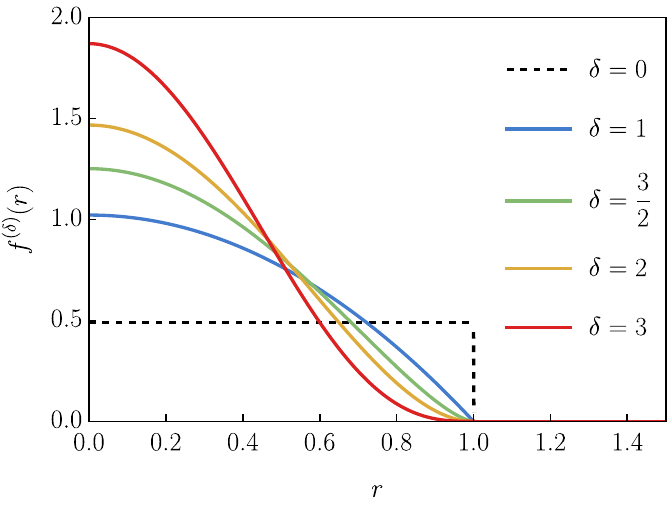}
    \end{flushleft}\caption{\label{fig:W_delta_vs_r} Shape of the smearing functions $f^{(\delta)}(\vec{x})$ for a few values of $\delta$ (with $R=1$). $r$ represents the  radial coordinate.}
\end{figure}

The parameter $\delta$ determines the differentiability class of $f^{(\delta)}(\vec{x})$. For example, for $\delta = 0$, $f^{(\delta)}(\vec{x})$ reduces to the top-hat function, which is discontinuous. In this case, $f^{(0)}(\vec{x})$ belongs to the Sobolev space $\dot{H}_{-\frac{1}{2}}$, but not to $\dot{H}_{\frac{1}{2}}$, so it can be used to smear $\hat{\Phi}(\vec x)$ but not $\hat{\Pi}(\vec x)$. For $\delta=1$, $f^{(1)}(\vec{x})$ is continuous but its first derivative is not. All functions $f^{(\delta)}$ for $\delta > 0$ belong to both $\dot{H}_{\frac{1}{2}}$ and $\dot{H}_{-\frac{1}{2}}$.

The symplectic structure ${\bm \omega}$ restricted to $\Gamma_A$, written in the basis $({\bm \gamma}^{(1)},\bm{\gamma}^{(2)})$, is
\be {\bm \omega}_A=\begin{pmatrix} \omega(\bm\gamma^{(1)},\bm\gamma^{(1)}) & \omega(\bm\gamma^{(1)},\bm\gamma^{(2)}) \\ \omega(\bm\gamma^{(2)},\bm\gamma^{(1)}) & \omega(\bm\gamma^{(2)},\bm\gamma^{(2})
    \end{pmatrix}=
\begin{pmatrix} 0 & -1\\1& 0\end{pmatrix}. \ee
${\bm \omega}_A$ is a symplectic structure in its own right, confirming that the vector space $\Gamma_A$ spanned by $\bm\gamma^{(1)}$ and $\bm\gamma^{(2)}$ is a symplectic subspace of $\Gamma_{\sigma_M}$.

The calculation of $\mu_A$ is straightforward: since $\mu_M = 0$ for the Minkowski vacuum, it automatically follows that $\mu_A$ also vanishes. 

The calculation of $\sigma_A$ reduces to computing the product of the basis vectors $\bm\gamma^{(1)}$ and $\bm\gamma^{(2)}$:
\begin{equation}\label{sigmaexamp1}
    \sigma_A = \begin{pmatrix}
        \sigma(\bm\gamma^{(1)}, \bm\gamma^{(1)}) & \sigma(\bm\gamma^{(1)}, \bm\gamma^{(2)}) \\
        \sigma(\bm\gamma^{(1)}, \bm\gamma^{(2)}) & \sigma(\bm\gamma^{(2)}, \bm\gamma^{(2)})
    \end{pmatrix}\,.
\end{equation}
Using $\sigma_M$ given in Eq.~\eqref{sigmaMink} and the form of the basis vectors, the components of $\sigma_A$ are
\begin{equation}
    \sigma_A = \begin{pmatrix}
        (f^{(\delta)}, f^{(\delta)})_{-1/2} & 0 \\
        0 & (g^{(\delta)}, g^{(\delta)})_{1/2}
    \end{pmatrix}.
\end{equation} 

With this, given any two vectors $\bm\gamma$ and $\bm\gamma'$ in $\Gamma_A$, with components $(a_1, a_2)$ and $(a'_1, a'_2)$, respectively, in the basis $(\bm\gamma^{(1)}, \bm\gamma^{(2)})$ of $\Gamma_A$, we have
\begin{widetext}
    \bea 
    \sigma_A(\bm \gamma, \bm\gamma') &=& 
    \begin{pmatrix} a_1 & a_2 \end{pmatrix}
    \begin{pmatrix}
        (f^{(\delta)}, f^{(\delta)})_{-1/2} & 0 \\
        0 & (g^{(\delta)}, g^{(\delta)})_{1/2}
    \end{pmatrix}
    \begin{pmatrix} a'_1 \\ a'_2 \end{pmatrix} \nonumber \\
    &=& a_1 a_1'\, (f^{(\delta)}, f^{(\delta)})_{-1/2} + 
    a_2 a_2'\, (g^{(\delta)}, g^{(\delta)})_{1/2}.
\eea 
\end{widetext}

The functions $f^{(\delta)}$ ---defined in \eqref{eq:fdelta_family}--- are particularly well-suited for calculating these components analytical, since their Sobolev products can be computed using the Fourier transform of $f^{(\delta)}$:
\bea
    \tilde{f}^{(\delta)} (\vec{k}) &=& A_{\delta} \, \pi^{3/2} 2^{\delta + \frac{3}{2}} R^3 \, \Gamma(\delta + 1) \,\nonumber \\ &\times&  (kR)^{-\delta - \frac{3}{2}} J_{\frac{1}{2} (2\delta + 3)}(kR)\,,
\eea
where $J_{\alpha}(x)$ denotes the Bessel function of the first kind of order $\alpha$ and $\Gamma(x)$ is the Gamma function. The resulting expressions are:
\begin{equation}\label{eq:SN1o2B}
      (g^{(\delta)}, g^{(\delta)})_{\frac{1}{2}} = \frac{4 \Gamma (2 \delta ) \Gamma (\delta +1)^2 \Gamma \left(2 \delta +\frac{5}{2}\right)}{\sqrt{\pi } \Gamma \left(\delta +\frac{1}{2}\right)^2 \Gamma (2 \delta +1) \Gamma (2 \delta +2)}\,,
\end{equation}
\begin{equation}\label{eq:SNm1o2B}
      (f^{(\delta)}, f^{(\delta)})_{-{\frac{1}{2}}} = \frac{2 \Gamma (\delta +1) \Gamma \left(2 \delta +\frac{5}{2}\right) \Gamma \left(\delta +1\right)}{\sqrt{\pi } \Gamma \left(\delta +\frac{1}{2}\right) \Gamma \left(\delta +\frac{3}{2}\right) \Gamma \left(2 \delta +3\right)} \,.
\end{equation}

Substituting these expressions back into \eqref{sigmaexamp1} yields the covariance matrix of the reduced state for subsystem~$A$. Since $\mu_A = 0$, this covariance matrix fully characterizes the reduced state $\hat{\rho}_A^{\mathrm{red}}$. I.e., all physical predictions about subsystem $A$ can be derived from $\sigma_A$. 

\subsection{Correlations and entanglement between subsystems}

Throughout this article, we will focus on quantifying the correlations and entanglement between {\em finite-dimensional} subsystems of a linear bosonic field theory. Given two finite-dimensional subsystems, $A$ and $B$, we can apply standard tools in quantum mechanics to quantify entropies, correlations, and entanglement. 

We will also focus on basis-independent measures of correlations and entanglement, that depend solely on the subsystems under consideration, rather than on any specific choice of basis within each subsystem. In the classical theory, a change of canonical (Darboux) basis in the phase space $\Gamma$ corresponds to a linear canonical transformation, i.e., a linear transformation which leaves the symplectic structure ${\bm \omega}$ unchanged (these linear transformations are commonly referred to as symplectic transformations). The group formed by all such transformations, the symplectic group, is infinite-dimensional in field theory. 

Symplectic transformations restricted to a finite-dimensional subsystem $\Gamma_A$ and correspond to the identity in the symplectic-orthogonal complement $\Gamma_{\bar{A}}$ of $\Gamma_A$, are called {\em system-local} symplectic transformations, and form the  finite-dimensional subgroup ${\rm Sp}(2N_A,\mathbb{R})$.  We will focus on measures of correlations and entanglement that are invariant under such system-local symplectic transformations. 

Entanglement entropy, mutual information, and logarithmic negativity are all invariant under system-local symplectic transformations on each of the two subsystems involved. To make this invariance explicit, we will express these measures in terms of system-local symplectic invariants. It is also recall (see e.g. \cite{serafini2017quantum}) that entropies, correlations, and entanglement in a Gaussian state are fully determined by the covariance metric $\sigma$ and do not depend on the first-moment covector $\mu$. Thus, we will will focus attention on system-local symplectic invariants of a covariant metric $\sigma$.

Consider two single-mode subsystems $A$ and $B$ of the field theory under consideration. The description below straightforwardly generalizes to $N$-dimensional subsystems, with $N \in \mathbb{N}$, but in this article we will primarily restrict to single-mode subsystems, i.e., $N=1$. Let $\sigma_{AB}$ denote the covariance metric of the combined system and  $\sigma_{A}$ and $\sigma_{B}$ those of the individual subsystem. The rank-two covariant tensor $C$ defined as 
\begin{equation}
C^{\text{corr}} = \sigma_{AB} - \sigma_{A} \oplus \sigma_{B}
\end{equation}
encodes all information about correlations between the two subsystems --- $C^{\text{corr}}=0$ if and only if the reduced state $\hat{\rho}_{AB}^{\text{red}}$ is a product state $\hat{\rho}_{AB}^{\text{red}} = \hat{\rho}_{A}^{\text{red}} \otimes \hat{\rho}_{B}^{\text{red}}$. In matrix form, $C^{\text{corr}}$ and $\sigma_{AB}$ have the following structure: 
\begin{equation}\label{coormatrix}
C^{\text{corr}} = \begin{pmatrix} 0 & C \\ C^{T} & 0 \end{pmatrix}, \quad \sigma_{AB} = \begin{pmatrix} \sigma_{A} & C \\ C^{T} & \sigma_{B} \end{pmatrix} .
\end{equation}
Here, $\sigma_A$, $\sigma_B$, and $C$ are $2 \times 2$ matrices. The components of these matrices depend on the choice of basis. On the contrary, the determinants $\det \sigma_{AB}$, $\det \sigma_A$, $\det \sigma_B$, and $\det C$ are invariant under system-local symplectic transformations (this can be easily seen using that all symplectic transformations have unit determinant). These determinants encode the invariant information in the reduced state of the system $(A,B)$ that we are interested in. 

The following six combinations of these determinants will be particularly useful:
\begin{align} \label{eq:symplecticeigvalsall}
\nu_I &\equiv \sqrt{\det \, \bm{\sigma}_I}, \quad I=A,B, \nonumber\\   
\nu_{\pm}^2 &\equiv \frac{\Delta \pm \sqrt{\Delta^2 - 4 \, \det \, \bm{\sigma}_{AB}}}{2}, \\
\tilde{\nu}_{\pm}^2 &\equiv \frac{\tilde{\Delta} \pm \sqrt{\tilde{\Delta}^2 - 4 \, \det \, \bm{\sigma}_{AB}}}{2}, \nonumber
\end{align}
where\footnote{Note, in passing, that $\nu_A$ and $\nu_B$ are the absolute values of the eigenvalues of $J_A$ and $J_B$, respectively. The quantities $\nu_\pm$ are the absolute values of the eigenvalues of $J_{AB}$, and $\tilde{\nu}_\pm$ are derived from the eigenvalues of the \emph{partial transposed} of  $J_{AB}$, defined below.}
\begin{align}
\Delta &:= \det \, \bm{\sigma}_A + \det \, \bm{\sigma}_B + 2 \, \det \, \bm{C}, \\
\tilde{\Delta} &:= \det \, \bm{\sigma}_A + \det \, \bm{\sigma}_B - 2 \, \det \, \bm{C}.
\end{align}

In this article, we will focus on computing (i) the von Neumann entropy of each subsystem, (ii) the mutual information between two subsystems, and (iii) their logarithmic negativity. For the reader's convenience, we briefly recall how these quantities can be computed using the symplectic invariants listed above.

\subsubsection{von Neumann entropy}
The von Neumann entropy of a single-mode subsystem $A$, when the total system is prepared in the Gaussian state $\hat \rho$, is given by
\begin{equation} \label{eq:S} \begin{split}
    S(\nu_A) =& \left[\left( \frac{\nu_A+1}{2}\right) \log_2\left( \frac{\nu_A+1}{2}\right) \right. \\ 
    & \left. - \left( \frac{\nu_A-1}{2}\right) \log_2\left( \frac{\nu_A-1}{2}\right) \right]\,,
\end{split}
\end{equation}
where $\nu_A$ was defined in \eqref{eq:symplecticeigvalsall}. 
Recall that von Neumann entropy  quantifies the mixedness of $\hat \rho^{\rm red}_A$, with $S(\nu_A)=0$ indicating a pure state. 

Furthermore, when the state of the field $\hat \rho$ is pure, $S(\nu_A)$ quantifies the entanglement between $A$ and its complement $\bar{A}$ (containing the remaining degrees of freedom in the field other than $A$). It is important to remember that $S(\nu_A)$ quantifies the entanglement between $A$ and its complement only when the state $\hat \rho$ is pure; for mixed $\hat \rho$, the entropy $S_A$ can be zero even in the absence of entanglement. 

It is also important to note that, for any Hadamard state $\hat{\rho}$, we have $S(\nu_A) \neq 0$ for any subsystem $A$ compactly localized in space. This fact follows from the Reeh--Schlieder theorem~\cite{reehschlieder}. Consequently, all compactly supported subsystems are entangled with other field modes.

\subsubsection{Correlations}

Mutual information serves as an invariant measure of the correlations between two subsystems. For two single modes, it can be computed in terms of the von Neumann entropy of each subsystem as:
\begin{equation} \label{eq:def_MI}
\mathcal{I}(A,B) = S(\nu_A) + S(\nu_B) - S(\nu_+) - S(\nu_-)\,. 
\end{equation}
Mutual information captures all correlations between the subsystems, both classical and quantum. When the state $\hat \rho_{AB}$ describing the combined system is pure, all correlations are genuinely quantum, and non-zero mutual information implies the existence of entanglement. However, this is not true if $\hat \rho_{AB}$ is mixed. 

In field theory, if $A$ and $B$ correspond to subsystems made of field modes compactly localized in space, $\hat \rho_{AB}$ is {\em always} mixed. Thus, in this context, it is essential not to identify $\mathcal{I}(A,B)$ with entanglement. 

\subsubsection{Entanglement between subsystems}\label{entsubs}

As already mentioned, when the reduced state $\hat \rho_{AB}$ is mixed, the von Neumann entropy of one of the subsystems does not quantify entanglement between $A$ and $B$. In this case, it is necessary to use entanglement measures applicable to general states $\hat \rho_{AB}$. Logarithmic negativity (LN) provides such an entanglement measure, particularly useful when dealing with quantum Gaussian states given the efficiency with which it can be computed. 

The LN is defined via the positivity of the partial transpose (PPT) criterion for separability{~\cite{peres96,Horodecki:1996nc,Simon_2000}}.\footnote{The PPT criterion states that, if the partial transpose of $\hat \rho_{AB}$ ---i.e., the result of transposing the operator $\hat \rho_{AB}$ only with respect to one of the subsystems, either $A$ or $B$--- fails to be positive semi-definite, then $\hat \rho_{AB}$ is entangled. This holds regardless of the basis used to perform the partial transposition~\cite{serafini2017quantum}.} For general states $\hat \rho_{AB}$ ---pure or mixed, Gaussian or not--- the LN is defined from the \emph{partial transposition} of $\hat \rho_{AB}$ as:
\begin{equation} \label{LNnorm1}
\mathrm{LN}(\hat{\rho}) = \log_2 \| \hat{\rho}_{AB}^{T_A} \|_{1} \, ,
\end{equation}
where $T_A$ denotes the transposition only on subsystem $A$, and $\| \cdot \|_1$ denotes the trace norm, defined as $\| \hat O \|_1:={\rm Tr}\sqrt{\hat O^\dagger \hat O}$ and equal to the sum of the absolute value of the eigenvalues of $\hat O$, when diagonalizable. Although transposition is a basis-dependent operation, it is straightforward to show that $\| \hat{\rho}^{T_A} \|_{1}$ does not depend on the basis used in the transposition. Thus, LN is invariant under system-local symplectic transformations. 

When $\hat \rho_{AB}$ is a Gaussian state and both $A$ and $B$ are single-mode systems, Eq.~\eqref{LNnorm1} reduces to \cite{serafini2017quantum}
\begin{equation} \label{LN}
\mathrm{LN}(\hat{\rho}_{AB}) = \max\{0, -\log_{2} \tilde{\nu}_-\}\,,
\end{equation}
where $\tilde{\nu}_-$ was defined in Eq.~\eqref{eq:symplecticeigvalsall}. Thus, $\mathrm{LN}(\hat{\rho}_{AB})$ can be computed in a remarkably simple manner. 

For a two-mode system $(A,B)$ prepared in a Gaussian state, it has been proven in \cite{Simon_2000} that LN is non-zero if and only if $A$ is entangled with $B$. Therefore, the condition $\tilde{\nu}_- < 1$ is both necessary and sufficient for quantum entanglement in such systems. Furthermore, lower values of $\tilde{\nu}_-$ ---corresponding to higher values of LN--- indicate a greater degree of entanglement. \\

The tools presented so far in this section, summarized in Table~\ref{tab:summary}, allow the computation of correlations and entanglement between any pair of modes within a field theory and enable the study of how these quantities vary with distance and, in the case of de Sitter space, with curvature. 

\begin{table*}
    \centering
    \begin{tabular}{l|c}
   \hline \hline
                  Measure         &Used in this article to characterize\\
                  \hline
       von Neumann Entropy  & Entanglement in pure states \\
       
      Mutual information   &Total (classical and quantum) correlations in pure and mixed states\\
    Logarithmic Negativity  &  Entanglement in pure and mixed states \\
    \hline\hline 
    \end{tabular}
    \caption{Summary of some correlation measures and how they are used in this work.}
    \label{tab:summary}
\end{table*}

It is worth highlighting that working with finite-dimensional subsystems ensures that all quantities defined above are free from ultraviolet divergences.

\section{The cosmological patch of de Sitter spacetime\label{sec:dSqft}}

In this section, we provide a concise overview of the quantum theory of a linear scalar field  in the cosmological  patch of de Sitter spacetime. This theory is of great interest for early universe cosmology, as it accurately describes the scalar curvature perturbations and (individual polarization modes of) tensor perturbations in the cosmic inflationary phase. Keeping this application in mind, we will focus attention on a light field, whose mass is small compared to the Hubble radius $m^2/H^2\ll 1$. 

Consider a four-dimensional manifold with $\mathbb{R}^4$ topology equipped with the conformally-flat line element 
\begin{equation}
    ds^2 = a^2(\eta) \left( -d\eta^2 + d\mathbf{x}^2 \right),
    \label{eq:dS_metric}
\end{equation}
where $\mathbf{x}\in \mathbb{R}^3$ and $\eta$ ranges from $-\infty$ to $0$. The function $a(\eta)$ ---the so-called scale factor---  is chosen to be $a(\eta) = - \frac{1}{H \eta}$. This spacetime belongs to the 
Friedmann-Lema\^itre-Robertson-Walker family of spatially-flat cosmologies, it has constant Ricci curvature throughout the spacetime, and it has a spacelike singularity when $\eta\to -\infty$ (the big bang). 

This spacetime is isomorphic to ``half'' of four-dimensional de Sitter spacetime. For this reason, it is commonly referred to as the \emph{cosmological} or \emph{Poincaré patch} of de Sitter space (abbreviated as PdS hereafter).

Although a generic FLRW spacetime possesses six isometries---corresponding to spatial translations and rotations---the PdS patch admits an additional Killing vector field. This extra symmetry can be understood as being ``inherited'' from the full de Sitter spacetime, which is maximally symmetric. 

The de Sitter group in four spacetime dimensions has ten independent Killing vector fields. Locally, all of them are isometries of PdS. However, because the Poincaré patch covers only a portion of de Sitter space, not all of these transformations correspond to global isometries of PdS. Only the subgroup that leaves the Poincaré patch invariant corresponds to global isometries of PdS (see, e.g., \cite[Sec.~IV~C]{abk1}). This subgroup forms a seven-dimensional Lie group, generated by three spatial translations, three rotations, and one additional isometry defined by the Killing vector field
\begin{equation}
K^{\mu} = -H\, \eta\, \partial^{\mu}_{\eta} - H\, x\, \partial^{\mu}_{x} - H\, y\, \partial^{\mu}_{y} - H\, z\, \partial^{\mu}_{z}\,.
\end{equation}

It is well known that invariance under spatial rotations and translations alone is insufficient to single out a preferred vacuum state in the quantum field theory under consideration. However, the inclusion of the additional isometry in PdS selects a unique state that is both invariant under the seven global isometries of PdS and satisfies the Hadamard condition. This  distinguished state is known as the \textit{Bunch-Davies vacuum} \cite{Chernikov:1968zm,Tagirov:1972vv,Bunch:1978yq}. Since $K^{\mu}$ is not time-like in the entire Poincaré patch, the Bunch-Davies vacuum is not the ground states of any Hamiltonian and quanta over it do not  have a natural interpretation in terms of particles.

For a scalar field obeying the Klein-Gordon equation,
\begin{equation}
    (\Box + m^2) \hat{\phi}(\eta, \mathbf{x}) = 0\,,
\end{equation}
the Bunch-Davies vacuum can be described mathematically as follows. Adopting a Fock representation, one can express the operator-valued distribution $\hat{\phi}(\eta, \vec{x})$ as
\begin{equation}\label{BDrep}
    \hat{\phi}(\eta, \vec x) \!= \!\!\int\!\!\! \frac{d^3 k}{(2\pi)^3} e^{i \vec k \cdot \vec x}\!\left( e^{\mathrm{BD}}_{k}(\eta)   \hat{A}_{\mathbf{k}} + e^{\mathrm{BD}*}_{k}(\eta) \hat{A}_{-\mathbf{k}}^\dagger \right),
\end{equation} 
where the functions $e^{\mathrm{BD}}_{k}(\eta) e^{i \vec k \cdot \vec x}$ are mode solutions of the Klein-Gordon equation, with
\begin{equation}
    e^{\mathrm{BD}}_{k}(\eta) = \sqrt{\frac{-\pi \eta}{4 a(\eta)^2}} H^{(1)}_{\nu}(-k\eta),
\end{equation}
and $H^{(1)}_{\nu}$ the Hankel function of the first kind of order $\nu = \sqrt{\frac{9}{4} - \frac{m^2}{H^2}}$. These complex solutions are referred to as the \textit{Bunch-Davies mode functions}, and the Fock vacuum annihilated by all operators $\hat{A}_{\mathbf{k}}$ is the Bunch-Davies vacuum.

The Bunch-Davies vacuum is a Gaussian state with vanishing first moments, $\mu_{\mathrm{BD}} = 0$. Its covariance metric, $\sigma_{\mathrm{BD}}$, can be computed analogously to the Minkowski spacetime example discussed in the previous section. Specifically, the action of $\sigma_{BD}(\vec{x}, \vec{x}')$ on two classical phase space elements, $\bm \gamma(\vec{x}) = (g(\vec{x}), f(\vec{x}))$ and $\bm \gamma'(\vec{x}) = (g'(\vec{x}), f'(\vec{x}))$, is given by
\begin{widetext}
    \bea\label{sigmap}
    \sigma_{\mathrm{BD}}(\bm \gamma, \bm \gamma') &=& \int d^3xd^3x'\begin{pmatrix} g(\vec x) & f(\vec x)\end{pmatrix}\begin{pmatrix} \langle 0|\{ \hat{\Pi}(\vec{x}),\hat{\Pi}(\vec{x}')\} |0\rangle & \langle 0|\{\hat{\Phi}(\vec{x}') ,\hat\Pi(\vec{x})\}|0\rangle \\\langle 0|\{\hat{\Phi}(\vec{x}) ,\hat \Pi(\vec{x}')\}|0\rangle & \langle 0|\{\hat{\Phi}(\vec{x}), \hat \Phi(\vec{x}')\}|0\rangle
\end{pmatrix}\begin{pmatrix} g'(\vec x')\\f'(\vec x')\end{pmatrix}  \\ \nonumber  &=&    \langle 0 | \{ \hat{\Pi}[g], \hat{\Pi}[g'] \} | 0 \rangle 
    + \langle 0 | \{ \hat{\Phi}[f], \hat{\Pi}[g'] \} | 0 \rangle 
+ \langle 0 | \{ \hat{\Phi}[f'], \hat{\Pi}[g] \} | 0 \rangle 
    + \langle 0 | \{ \hat{\Phi}[f], \hat{\Phi}[f'] \} | 0 \rangle.
\eea
\end{widetext}

Similar to Minkowski spacetime, the classical phase space can be taken, to begin with, as consisting of functions in Schwartz space, and the Cauchy-complete it using $\sigma_{\mathrm{BD}}$.

Each term in $\sigma_{\mathrm{BD}}(\bm \gamma, \bm \gamma')$ can be computed using the field representation \eqref{BDrep}. When the functions comprising $\bm \gamma$ and $\bm \gamma'$ are compactly supported in a small region of space and the typical separation between their supports is small ---both distances compared to the Hubble radius--- the product $\sigma_{\mathrm{BD}}(\bm\gamma, \bm\gamma')$ is well-approximated by $\sigma_M(\bm\gamma, \bm\gamma')$ (this  is explicitly illustrated below using several examples). To observe sizable differences between the Bunch-Davies vacuum and the Minkowski vacuum, one must consider either  $\bm\gamma(\vec{x})$ and $\bm \gamma'(\vec{x})$ supported in ``super-Hubble'' regions --- i.e., regions with $R > H^{-1}$ --- or functions whose supports are separated by super-Hubble distances.\footnote{ When working with non-spherically symmetric functions, we define the ``radius of support'' of a function $f(\vec x)$ as 
\[
R = \left( \frac{3}{4\pi} \int_A d^3x\, \sqrt{h} \, f(\vec{x}) \right)^{1/3}.
\]
The pre-factor $(3/4\pi)^{1/3}$ is somewhat irrelevant, since we will use $R$ merely as a reference scale. When the discussion involves two functions, as for instance when talking about a phase space element ${\bm \gamma}$, $R$ will refer to the largest of the scales $R$ defined from each function.}

In inflationary cosmology, interest is  focused on  field modes whose support satisfies $R \gg H^{-1}$ at the end of inflation, since these are the modes that become observationally accessible to us. (More precisely, the smallest primordial wavelength resolvable in the cosmic microwave background is of the order of $10^4$ Mpc today. At the end of inflation, such a mode had a physical wavelength exceeding $e^{50}$ times the Hubble radius in typical inflationary models \cite{Liddle_2003}.) Motivated by this, we will henceforth restrict our attention to the regime $R \gg H^{-1}$, although many of the results presented here can be extended to other regimes, as discussed below. Since we are interested in the regime $RH \gg 1$, we will use $RH$ as our control parameter. \\

\begin{prop}\label{prop:BDcorrelations}
In the regime of interest, namely $RH \gg 1$ and $m/H \ll 1$, the symmetrized expectation values appearing in Eq.~\eqref{sigmap} can be expanded in powers of $RH$ and ${\mu}^2 \equiv \frac{3}{2} - \sqrt{\frac{9}{4} - \frac{m^2}{H^2}} \ll 1$. To leading order, the expressions are:
\begin{widetext}
\bea \label{phiphi} \braket{\{\hat{\Phi}[f],\hat{\Phi}[f']\}} &=& \mathrm{Re}(f,f')_{-\frac{1}{2}} + \frac{2^{2-2\mu^2}\,\pi\, (RH)^{2-2\mu^2}}{\cos^2(\pi \mu^2)\, \Gamma \left(- \frac{1}{2} + \mu^2\right)^2 } \, { R^{-2+2\mu^2}}\,   \mathrm{Re}(f,f')_{-\frac{3}{2}+\mu^2} +\mathcal{N}_{\Phi\Phi}, \\ 
\label{pipi}
    \braket{\{\hat{\Pi}[g],\hat{\Pi}[g']\}} &=& \mathrm{Re}(g, g')_{\frac{1}{2}}+ \mathcal{N}_{\Pi\Pi}\,, \\
\label{phipi}
    \braket{\{\hat{\Phi}[f],\hat{\Pi}[g']\}} &=& \frac{2^{1-\mu^2}{R^{-1+\mu^2}}\, \sqrt{\pi} \, (RH)^{1-\mu^2}}{\cos(\pi \mu^2)\, \Gamma \left(-\frac{1}{2} + \mu^2\right)}\,    \mathrm{Re}(f, g')_{-\frac{1- \mu^2}{2} }+\mathcal{N}_{\Phi\Pi}\,.
\eea

\end{widetext}
\end{prop}
Without loss of generality, we have chosen the scale factor $a(\eta_0) =1$ at the instant we are calculating these correlations functions. 
The terms denoted by $\mathcal{N}_{\Phi\Phi}$, $\mathcal{N}_{\Pi\Pi}$, and $\mathcal{N}_{\Phi\Pi}$ are of order $\mathcal{O}(\mu^2)$ when $\mu \ll 1$, and are therefore subleading. Their explicit form, together with a proof of this statement, is given in Appendix~\ref{app:proofprop1}. The terms $\mathcal{N}_{\Phi\Phi}$, $\mathcal{N}_{\Pi\Pi}$, and $\mathcal{N}_{\Phi\Pi}$ will not play a relevant role in the remainder of this article—although we explicitly compute them in several examples below to illustrate that they are indeed subleading (see, e.g. Fig.~\ref{fig:2pfballwithNterms}).

Expressions \eqref{phiphi},  \eqref{pipi}   and~\eqref{phipi} can be directly compared with the corresponding ones in Eq.~\eqref{synmcorrMink} for the Minkowski vacuum. The main differences arise in two aspects: (i) the field-momentum correlation \eqref{phipi} is different from zero for the Bunch-Davies vacuum, unlike for the Minkowski vacuum, and (ii) the field-field correlations \eqref{phiphi} in the Bunch-Davies vacuum include an extra term proportional to the Sobolev product of order $-\frac{3}{2} + \mu^2$. This term is absent in the Minkowski vacuum and is the origin of the most relevant difference between both states.

In the flat-space limit $H \to 0$ and $m\to 0$ (in this order), the terms $\mathcal{N}_{\Phi\Phi}$, $\mathcal{N}_{\Pi\Pi}$, and $\mathcal{N}_{\Phi\Pi}$ vanish, and 
\eqref{phiphi}-\eqref{phipi} reduce to the smeared two-point functions of Minkowski spacetime, as given in Eq.~\eqref{synmcorrMink}.

The term proportional to the Sobolev product of order $-\frac{3}{2} + \mu^2$ in \eqref{phiphi} is responsible for most of the distinctive features of the Bunch-Davies vacuum.  For instance, for $f$ and $f'$  supported within non-overlapping spherical regions, 
when $H \neq 0$ and the separation between the supports of $f$ and $f'$  is large compared to the Hubble radius, i.e., $|\Delta \vec x| \gg H^{-1}$, it is simple to show that 
\be \label{almostscinv}
\mathrm{Re}(f,f')_{-\frac{3}{2}+ \mu^2} \sim |\Delta \vec x|^{-2\mu^2}\,,
\ee
when $\mu\ll 1$ (see Appendix~\ref{app:asymptotic_sobolev} and the proof of Prop.~\ref{prop:asymptotics_partner} below for details of the derivation). Thus, this Sobolev product is responsible for the characteristic, nearly scale-invariant field-field correlations of the Bunch-Davies vacuum. Furthermore, this term is infrared divergent in the limit $m \to 0$, accounting for the well-known infrared divergence of the Bunch-Davies state. In contrast, the momentum-momentum and field-momentum correlations remain finite when $m \to 0$ and do not exhibit approximate scale invariance.

We use Eqs.~\eqref{phiphi}-\eqref{phipi} to write the  complex structure of the Bunch--Davies vacuum in the limit $\mu \ll 1$ and $RH\gg1$ as
\begin{widetext}

{  \begin{equation}\label{JdS}
     J_{BD}(\vec x,\vec x')=\int \frac{d^3k}{(2\pi)^3} e^{-i \vec{k}\cdot (\vec{x}-\vec{x}')} \left( \begin{matrix}
        K\, (RH)^{1- \mu ^2} \, k^{-1+ \mu ^2} & - \frac{1}{k}- K^2\,(RH)^{2-2\mu^2}\, k^{-3+2\mu^2} \\ 
     k & -K\, (RH)^{1- \mu ^2} \, k^{-1+ \mu ^2}
    \end{matrix}\right) + \mathcal{O}(\mu^2) \,.
\end{equation}  
where $K =  \frac{2^{1-\mu^2} \sqrt{\pi}R^{-1+\mu^2}}{\cos(\pi \mu^2)\Gamma\left(-\frac{1}{2} + \mu^2\right) }$.
Is it easy to check that $J_{BD}^2=-\mathbb{I} \ \delta(\vec x -\vec x')+\mathcal{O}(\mu^2)$,  and $J_{BD}$ reduces to $J_M$ in \eqref{JM} in the limit $H\to 0$ and $\mu\to 0$ (in this order).} 
\end{widetext}

{\bf Remark: \label{remark_IR}} As mentioned before, the field-field correlation \eqref{phiphi} suffers from an infrared divergence in the massless limit. One way to avoid this
divergence while keeping $m=0$ is by restricting the test functions \( f(\vec{x}) \) used to smear the field operator. In particular, if one restricts to smearing functions with zero spatial average, i.e., \( f \in \Gamma_{\sigma_{BD}} \) such that
\be \label{zeroaverage} 
\int_{\mathbb{R}^3} d^3x\, f(\vec{x}) = 0,
\ee
all correlation functions are infrared finite. This restriction effectively reduces the kind of degrees of freedom of the theory. 

Smearing functions satisfying \eqref{zeroaverage} restricts us to a rather special family of degrees of freedom. This family is special because it does not display several of the distinctive features of the Bunch-Davies vacuum. In particular, the long-distance behavior of 
\(\mathrm{Re}(f,f')_{-\frac{3}{2}+ \mu^2}\) shown in \eqref{almostscinv} no longer holds. Consequently, correlations for these field modes do not exhibit an almost scale-invariant behavior at large separations (see Appendices~\ref{app:asymptotic_sobolev} and \ref{app:proofs_special_functions} for further details on these special functions).
 
Since functions in this family are somewhat fine-tuned and do not display the distinctive features of the Bunch-Davies vacuum, most of the proofs presented in this article need to be adapted accordingly. To improve readability, the proofs specific to this family have been relegated to { Appendix~\ref{app:proofs_special_functions}}.

It is also worth noting that condition \eqref{zeroaverage} is morally analogous to introducing an infrared cutoff \( k_c \) on the modulus of the wave numbers in Fourier space. However, \eqref{zeroaverage} is a somewhat milder constraint, as it corresponds to removing only the zero mode \( \vec{k} = 0 \), rather than truncating all low-momentum modes $|\vec k|<k_c$.\footnote{ A subset of these special modes has been recently examined in \cite{Nambu_2023} within the framework of a massless scalar field in de Sitter spacetime.} This remark serves to contextualize the rather drastic limitation on accessible degrees of freedom imposed by introducing an infrared cutoff.

\subsection{Examples}
\label{subsec:vNEntropy_examples}
To illustrate the calculations discussed so far in this section, we now present three concrete examples for which the correlators in Eqs.~\eqref{phiphi}--\eqref{phipi} can be computed analytically. We will repeatedly use these examples throughout this article. 

\begin{exmp}
\label{ex:singlemodeball}
{\bf Single-mode subsystem supported in a ball}\\

This example is the same as that presented in Sec.~\ref{example-ball-Minkowski}, but now using the Bunch-Davies vacuum for a light scalar field, instead of the Minkowski vacuum for a massless field.

The calculations required to obtain the self-correlators \eqref{phiphi}--\eqref{phipi} reduce to evaluating Sobolev products of the functions $f^{(\delta)}$ and  $g^{(\delta)}=R \,f^{(\delta)}$ used to define the single-mode subsystem (these functions were defined in Eq.~\eqref{eq:ex1vNS}). As in the Minkowski case, these norms can be computed analytically:
\begin{widetext}
\begin{equation}\label{eq:SNm3o2AB}
{ (f^{(\delta)},f^{(\delta)})_{-\frac{3}{2} + \mu^2}} 
= \frac{2 \Gamma (\delta +1) \Gamma \left(2 \delta +\frac{5}{2}\right) \Gamma \left(\mu ^2\right) \Gamma \left(-\mu ^2+\delta +2\right)}{
\sqrt{\pi } \Gamma \left(\delta +\frac{1}{2}\right) \Gamma \left(-\mu ^2+\delta +\frac{5}{2}\right) \Gamma \left(-\mu ^2+2 \delta +4\right)} {  R^{2-2\mu^2}}\,,
\end{equation}
and 
\begin{equation} \label{eq:SNm1o2pmu2AB}
{ (f^{(\delta)},g^{(\delta)})_{-\frac{1}{2} + \frac{\mu^2}{2}}} 
= \frac{2 \Gamma (\delta +1) \Gamma \left(2 \delta +\frac{5}{2}\right) \Gamma \left(\frac{\mu ^2}{2}+1\right) \Gamma \left(-\frac{\mu ^2}{2}+\delta +1\right)}{
\sqrt{\pi } \Gamma \left(\delta +\frac{1}{2}\right) \Gamma \left(-\frac{\mu ^2}{2}+\delta +\frac{3}{2}\right) \Gamma \left(-\frac{\mu ^2}{2}+2 \delta +3\right)}{  R^{1-\mu^2}}\,.
\end{equation}
\end{widetext}
The norms  $||f^{(\delta)}||^2_{-1/2}$ and $||g^{(\delta)}||^2_{1/2}$ were reported in \eqref{eq:SN1o2B} and \eqref{eq:SNm1o2B}.

Note that $||f^{(\delta)}||^2_{-3/2+\mu^2}$ is proportional to $\Gamma(\mu^2)$, which diverges in the massless limit $\mu^2 \to 0$ producing the infra-red divergence of the Bunch-Davies vacuum.

The self-correlators \eqref{phiphi}--\eqref{phipi} are plotted in Fig.~\ref{fig:2pfballwithNterms} for different values of $RH$ and different masses $\mu$.

To verify that the terms $\mathcal{N}_{\Phi\Phi},\mathcal{N}_{\Phi\Pi}$ and $\mathcal{N}_{\Pi\Pi}$ in Eqs.~\eqref{phiphi}--\eqref{phipi} are indeed subleading in the regime $RH \gg 1$ and $\mu \ll 1$, we have evaluated them numerically (their explicit expressions are provided in Appendix~\ref{app:proofprop1}). 
\begin{figure*}[hbtp]
    \centering
\includegraphics[width=0.35\textwidth]{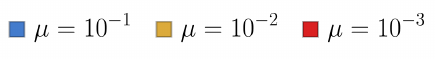}\\
    \subfigure[]{\includegraphics[width=0.45\textwidth]{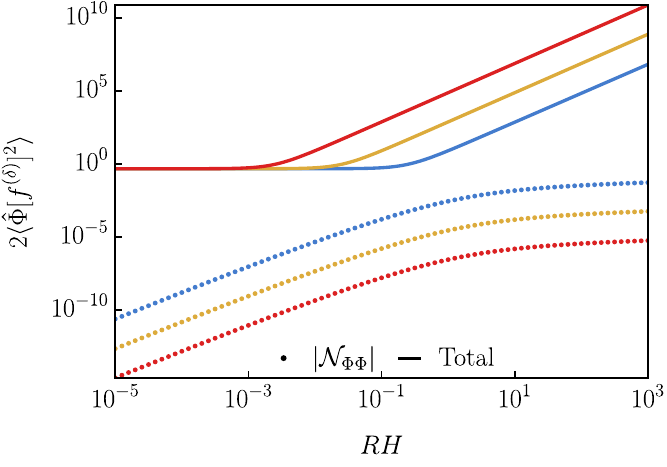}}~\subfigure[]{\includegraphics[width=0.45\textwidth]{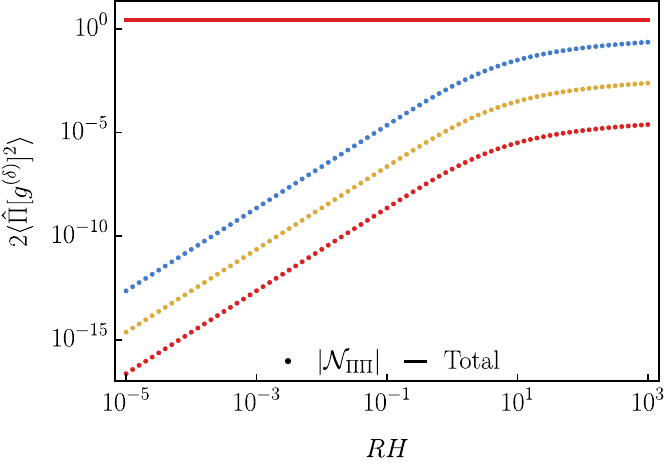}} \\
    \subfigure[]{\includegraphics[width=0.45\textwidth]{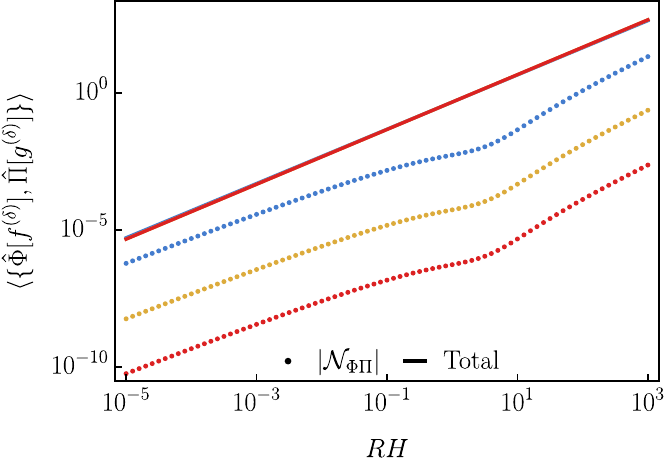}} 
   \caption{(a) Field–field, (b) momentum–momentum, and (c) field–momentum smeared self-correlations in the Bunch-Davies vacuum for the single mode defined in Eq.~\eqref{eq:ex1vNS}. 
    Note the growth of the field–field correlations as $\mu$ decreases—this growth signals an infrared divergence in the limit $\mu \to 0$. The momentum-momentum and field-momentum correlations are only mildly dependent on $\mu$ so that these lines appear on top of each other.   
    The dashed lines represent the terms $\mathcal{N}_{\Phi\Phi}$, $\mathcal{N}_{\Pi\Pi}$, and $\mathcal{N}_{\Phi\Pi}$ in Eqs.~\eqref{phiphi}--\eqref{phipi}. These figures confirm that these terms are subleading in the regime of interest ($RH \gg 1$ and $\mu \ll 1$), and they will therefore be neglected from now on. 
    }
    \label{fig:2pfballwithNterms}
\end{figure*}
\end{exmp}

Note that in Fig.~\ref{fig:2pfballwithNterms} the momentum–momentum self-correlations do not change appreciably with $RH$. This is a consequence of the small mass used in this example, $m/H \ll 1$. In this limit, the momentum–momentum self-correlations are indistinguishable from their values in flat spacetime—the effects of curvature manifest primarily in the field–field and field–momentum correlations. In flat spacetime and for a massless field, the vacuum correlations are invariant under a rescaling of $R$.

\begin{exmp}
\label{ex:singlemodeshell}
{\bf Single-mode subsystem supported in a spherical shell}\\

As a second example, we consider a single-mode subsystem supported in a spherical shell of finite width. This mode is defined by the symplectic subspace ${\rm span}[\gamma_{S}^{(1)},\gamma_{S}^{(2)}]$ of the classical phase space, where the phase space elements $\gamma_{S}^{(1)}$ and $\gamma_{S}^{(2)}$ are given by
{ \begin{equation}\label{eq:ex2vNS}
     \bm\gamma_{S}^{(1)} (\vec x) = \begin{pmatrix}
    0\\ 
   {\color{orange}} f_{S} (|\vec{x}|)
    \end{pmatrix}, \quad 
  \bm  \gamma_{S}^{(2)} (\vec x) = \begin{pmatrix}
    - g_{S} (|\vec{x}|)\\ 
    0
    \end{pmatrix},
\end{equation}
with the smearing functions $f_S(\vec x),g_S(\vec x)$ defined as
\bea\label{eq:fS}
 f_{S}(\vec{x}) &=& A_{S} \, \Big(|\vec{x}|-(R_S - d)\Big) \Big((R_S+d)-|\vec{x}| \Big) \nonumber \\ \nonumber 
 && \times \Theta\Big(|\vec{x}|-(R_S-d)\Big) \Theta\Big((R_S+d) - |\vec{x}|\Big)\,, \\ \nonumber  g_{S}(\vec{x}) &=& R_S \, f_{S}(\vec{x}).
\eea

The real numbers $R_S \pm d$ correspond to the inner and outer radii of the shell, and 
\begin{equation}
   A_{S} =  { R_S^{-1/2}\,} \frac{\sqrt{\frac{105}{\pi }}}{8 \sqrt{{d}^7+7 {d}^5 {R}_S^2}}
\end{equation}
is a normalization constant ensuring that ${\bm \omega}(\bm\gamma_{S}^{(2)},\bm\gamma_{S}^{(1)}) = 1$.} 
 The form of the smearing function $f_{S}(\vec{x})$ is illustrated in Fig.~\ref{fig:shell_geometric}.

The Sobolev norms relevant for computing the self-correlators \eqref{phiphi}--\eqref{phipi} can also be obtained analytically in this case, although the resulting expressions are lengthy  and not particularly illuminating (we do not reproduce them here). The  self-correlation  are qualitatively similar to the ones showed in  previous example (see Fig.~\ref{fig:2pfballwithNterms}).

\savebox{\mybox}{\includegraphics[width=0.5\textwidth]{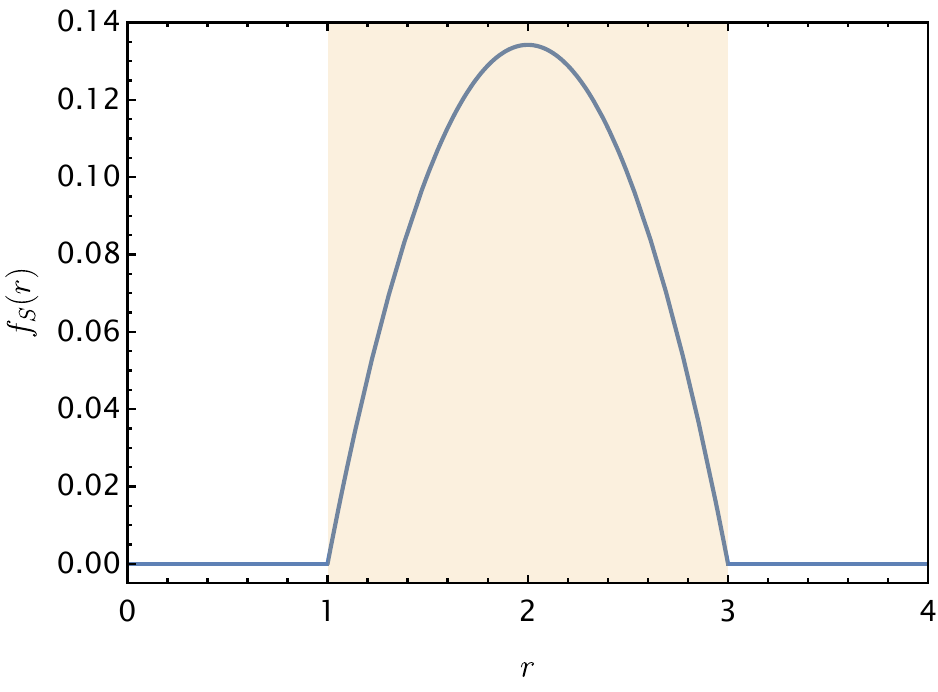}}
\begin{figure*}[hbtp]
    \centering
     \subfigure[]{
        \begin{tikzpicture}
           \node at (0,0) {\resizebox{!}{1.1 \ht\mybox}{
        \includegraphics[width=\textwidth]{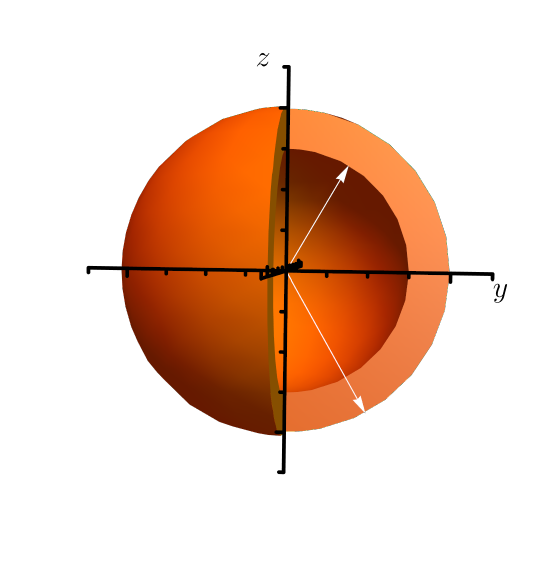}}};
           \node[rotate=60] at (0.75,0.75) {\textcolor{white}{$R_S - d$}};
           \node[rotate=-60] at (0.8,-0.6) {\textcolor{white}{$R_S + d$}};
        \end{tikzpicture}
     }~\subfigure[]{\includegraphics[width=0.5\textwidth]{03a_fs_vs_r.pdf}}
    \caption{(a) Geometric representation of the support of the mode in Eq.~\eqref{eq:ex2vNS}. (b) The smearing function $f_S(r)$ as a function of the radial distance $r=|\vec{x}|$. The support of the shell is shaded in orange.}  
    \label{fig:shell_geometric}
\end{figure*}

\end{exmp}

\begin{exmp}\label{ex:twoballs}{\bf Two modes supported in non-overlapping balls.}\\

In this example, we consider two single-mode subsystems, $A$ and $B$, each identical to the single-mode defined in Example~\ref{ex:singlemodeball}, but supported in disjoint spheres of radius $R$. The centers of the two supports are separated by a distance $|\Delta \vec x|$ such that $|\Delta \vec x|/R > 2$, so they do not overlap.

The self-correlators of each mode are identical to those computed in Example~\ref{ex:singlemodeball}.
The Sobolev products appearing in the cross-correlators can also be computed analytically in this case, with the results:
\begin{widetext}
{
\begin{equation}\label{eq:SIPs1o2BABB}
\mathrm{Re}(g^{(\delta)}_A,g^{(\delta)}_B)_{\frac{1}{2}} = -\frac{2^{-2 \delta }\,\Gamma (\delta +1)\, \Gamma \left(2 \delta +\frac{5}{2}\right)}{ \Gamma \left(\delta +\frac{1}{2}\right)\, \Gamma \left(\delta +\frac{5}{2}\right)^2 } \left(\frac{|\Delta \vec x|}{R}\right)^{-4}  
\, {}_3F_2\left(\tfrac{3}{2}, 2, \delta +2;\, \delta +\tfrac{5}{2}, 2 \delta +4;\, \tfrac{4 R^2}{|\Delta \vec x|^2}\right),
\end{equation}
\begin{equation}\label{eq:SIPsm3o2BABB}
\begin{split}
\mathrm{Re}(f^{(\delta)}_A,f^{(\delta)}_B)_{-\frac{3}{2} + \mu^2} 
&= -\frac{2^{-2 \delta -1} \, \Gamma (\delta +1) \, \Gamma \left(2 \delta +\frac{5}{2}\right) \, \cos \left(\pi  \mu ^2\right) \, \Gamma \left(2 \mu ^2 -1\right)}{\Gamma \left(\delta +\tfrac{1}{2}\right) \, \Gamma \left(\delta +\tfrac{5}{2}\right)^2} 
\left(\frac{|\Delta \vec x|}{R}\right)^{-2 \mu ^2} \\ 
&\quad \times \, {}_3F_2\left(\delta +2, \mu ^2, \mu ^2 - \tfrac{1}{2}; \, \delta + \tfrac{5}{2}, 2 \delta + 4; \, \tfrac{4 R^2}{|\Delta \vec x|^2}\right)  R^{2-2\mu^2},
\end{split}
\end{equation}
and 
\begin{equation}\label{eq:SIPsm1o2BABB}
\begin{split}
\mathrm{Re}(f^{(\delta)}_A,g^{(\delta)}_B)_{-\frac{1}{2} + \frac{\mu^2}{2}} 
&= \frac{2^{-2 \delta -1} \, \Gamma (\delta +1) \, \Gamma \left(2 \delta +\frac{5}{2}\right) \, \cos \left(\frac{\pi  \mu ^2}{2}\right) \, \Gamma \left(\mu ^2 +1\right)}{ \Gamma \left(\delta +\tfrac{1}{2}\right) \, \Gamma \left(\delta +\tfrac{5}{2}\right)^2} 
\left(\frac{|\Delta \vec x|}{R}\right)^{-2 - \mu ^2} \\
&\quad \times {}_3F_2\left(\delta +2, \tfrac{\mu ^2}{2} + \tfrac{1}{2}, \tfrac{\mu ^2}{2} + 1; \, \delta + \tfrac{5}{2}, 2 \delta + 4; \, \tfrac{4 R^2}{|\Delta \vec x|^2}\right)  R^{1-\mu^2}.
\end{split}
\end{equation}
}

\end{widetext}

Since ${}_3F_2(a,b,c;\,d,e;\,x) \to 1$ as $x \to 0$, the behavior of these Sobolev inner products in the regime $|\Delta \vec x|/R \gg 1$ is dominated by the prefactor. In particular,
\[
\mathrm{Re}(f^{(\delta)}_A,f^{(\delta)}_B)_{-\frac{3}{2} + \mu^2}\sim \left(\frac{|\Delta \vec x|}{R}\right)^{-2 \mu^2},
\]
for  $\mu^2 \ll 1$,  confirming that field-field correlations are nearly scale-invariant in the small-mass limit. On the other hand, \eqref{eq:SIPsm3o2BABB} diverges in the massless limit $\mu \to 0$.

Figure~\ref{fig:correlations_vs_Dx} shows the correlations between the two modes, as a function of the separation between their centers (in units of the Hubble radius).   
\begin{figure*}[tb]
    \centering
    \subfigure[]{\includegraphics[width=0.5\textwidth]{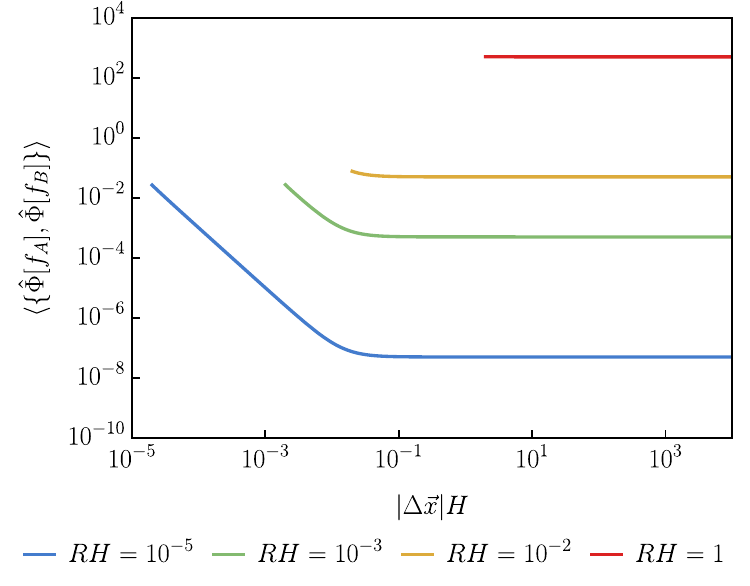}}~\subfigure[]{\includegraphics[width=0.5\textwidth]{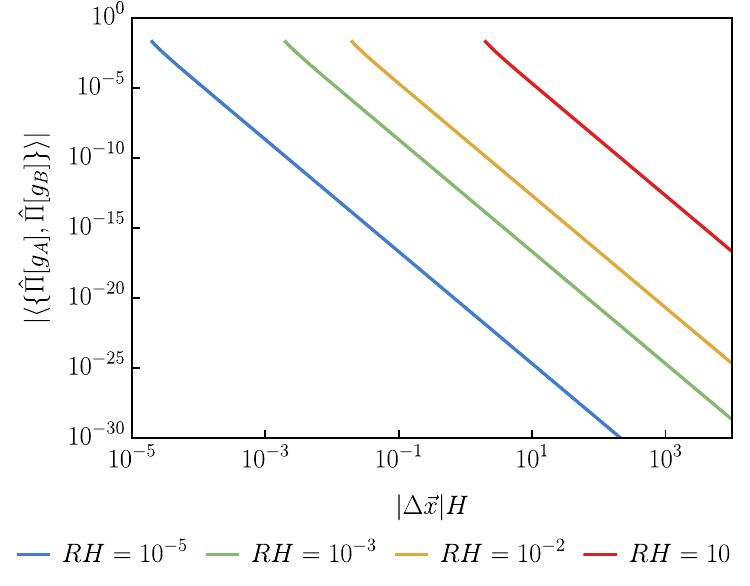}}
    \subfigure[]{\includegraphics[width=0.5\textwidth]{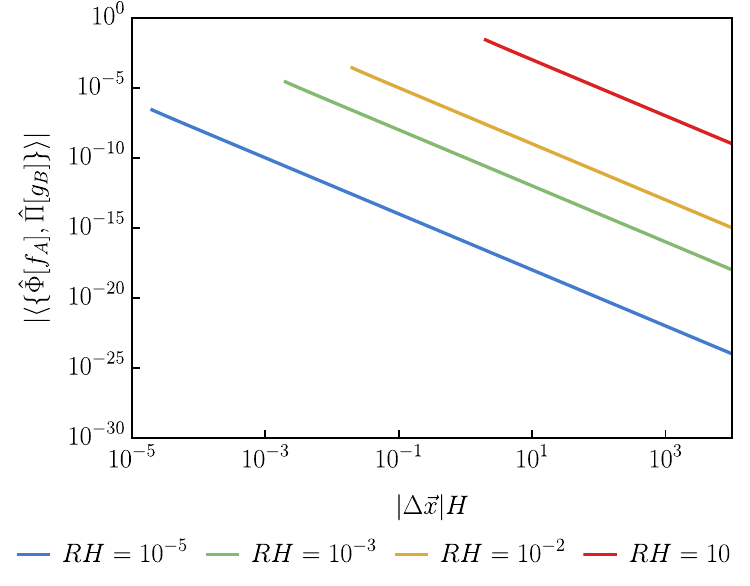}}
    \caption{(a) Field–field, (b) momentum–momentum, and (c) field–momentum correlations between two identical modes supported in non-overlapping balls, shown as functions of their separation $|\Delta \vec{x}|$, for $\mu^2 = 10^{-4}$ and various values of $RH$. The modes are defined as in Example~\ref{ex:singlemodeball} with $\delta = 2$. The non-overlapping condition $|\Delta \vec{x}| > 2R$ explains why the lines corresponding to different values of $R$ are plotted over different ranges of $|\Delta \vec{x}|$. Panel (a) illustrates the almost scale-invariant behavior of the field–field correlations when $|\Delta \vec{x}| H \gtrsim 1$. In contrast, the momentum-momentum and field-momentum correlations are not almost scale-invariant}
    \label{fig:correlations_vs_Dx}
\end{figure*}

\end{exmp}

\section{Von Neumann entropy\label{sec:vNEntropy}}

The von Neumann entropy of a subsystem composed of finitely many degrees of freedom in the Bunch--Davies vacuum is ultraviolet finite, and also infrared finite, provided $m\neq 0$. Furthermore, since the Bunch--Davies state is pure, the von Neumann entropy quantifies the entanglement between the subsystem and its complement ---namely, the remaining degrees of freedom in the field theory. (This interpretation is further reinforced by the concept of the {\em partner mode}, discussed in Sec.~\ref{sec:partners}.)

In this section, we focus on a single-mode subsystem that is locally supported in space, and study the way  its von Neumann entropy changes with $H$. We show that the entropy grows with $H$ for any such mode.

Let \( A \) be a single-mode subsystem whose classical state space is \( \Gamma_A = \mathrm{span}[\gamma_A^{(1)}, \gamma_A^{(2)}] \), where \( \gamma_A^{(1)}, \gamma_A^{(2)} \in \Gamma_{\sigma_{BD}} \) have compact support, and satisfy \( \boldsymbol{\omega}(\gamma_A^{(2)}, \gamma_A^{(1)}) = 1 \), so that \( \Gamma_A \) is a symplectic subspace of the classical phase space.

\begin{prop}\label{prop:vNgeneral}
    The von Neumann entropy of any\footnote{We restrict in this section to modes defined from functions not belonging to the special family defined in \eqref{zeroaverage}. These are discussed in Appendix \ref{app:proofs_special_functions}.} single-mode subsystem \(A\), compactly supported in a super-Hubble region ($RH \gg 1$) of a light scalar field theory ($m/H \ll 1$) in the Bunch--Davies vacuum, grows monotonically with $H$. 

\end{prop}
\begin{proof}
The expression \eqref{eq:S} for the von Neumann entropy is a monotonically increasing function of its argument. Thus, it suffices to show that the symplectic invariant $\nu_A^2$ grows with $H$ in the regime of interest.  

To prove this, we will show that $\nu_A^2 - (\nu_{A}^{\mathrm{Mink}})^2 \geq 0$ and increases with $H$, where $\nu_{A}^{\mathrm{Mink}}$ denotes the symplectic eigenvalue of the same mode but for a massless field in the Minkowski vacuum. ($\nu_{A}^{\mathrm{Mink}}$ is independent of $H$ and is used here merely for convenience in isolating the $H$-dependent terms in $\nu_A^2$.)

 Let us denote by $(\hat{O}_A^{(1)},\hat{O}_A^{(2)})$ the operator associated with the basis vectors $\gamma_A^{(1)}$ and $\gamma_A^{(2)}$. The symplectic eigenvalue $\nu_A$ of a single mode subsystem, given in Eq.~\eqref{eq:symplecticeigvalsall}, takes the form:
 \begin{equation}\label{eq:nuI_O1O2}
        \nu_A^2 = 4\braket{(\hat{O}_A^{(1)})^2} \braket{(\hat{O}_A^{(2)})^2} - \braket{\{\hat{O}_A^{(1)},\hat{O}_A^{(2)}\}}^2 \,.
    \end{equation}
    The expectation values in this expression reduce to a combination of the second moments of the smeared field and momentum operators:
\begin{equation}\label{eq:OIiOIj}\begin{split}
            &\braket{\{\hat{O}_A^{(i)},\hat{O}_A^{(j)}\}} = \\ &\braket{\{\hat{\Phi}[f_A^{(i)}],\hat{\Phi}[f_A^{(j)}]\}} + \braket{\{\hat{\Pi}[g_A^{(i)}],\hat{\Pi}[f_A^{(j)}]\}} \\ 
            &- \braket{\{\hat{\Phi}[f_A^{(i)}],\hat{\Pi}[g_A^{(j)}]\}} - \braket{\{\hat{\Pi}[g_A^{(i)}],\hat{\Phi}[f_A^{(j)}]\}}\,.
            \end{split}
        \end{equation}
        
    The second moments in the Bunch--Davies vacuum are written in  Eqs.~\eqref{phiphi}-~\eqref{phipi}, and for the Minkowski vacuum in \eqref{synmcorrMink}. Using these expressions, the leading contributions to $\nu_A^2 - (\nu_{A}^{\mathrm{Mink}})^2$ can be expressed as a polynomial in $RH$. In the regime $RH\gg1$, 
    the value of $\nu_A^2 - (\nu_{A}^{\mathrm{Mink}})^2$ is determined by the coefficient of the leading power in $RH$.  
    This coefficient depends on the choice of $\gamma_A^{(1)}$ and $\gamma_A^{(2)}$.

    If $\gamma_A^{(1)}=(g_A^{(1)},f_A^{(1)})$ and $\gamma_A^{(2)}=(g_A^{(2)},f_A^{(2)})$ we  must  differentiate  three different cases: 
    \begin{enumerate}
        \item  $f_A^{(1)} \neq 0$,  $f_A^{(2)} \neq 0$ with $f_A^{(1)} \neq f_A^{(2)}$;
        \item $f_A^{(1)} =0$, $f_A^{(2)} \neq 0$ (or vice-versa).
        \item  $f_A:= f_A^{(1)} = f_A^{(2)}$ and $f_A \neq 0$. 
    \end{enumerate}
    
\noindent {\bf Case 1:} Using Eqs.~\eqref{phiphi}-~\eqref{phipi}, we find         $\nu_A^2 -(\nu_{A}^{\mathrm{Mink}})^2= \mathfrak{a}\, (RH)^{4-4\mu^2}(1 + \mathcal{O}(\mu^2)) +\mathcal{O}[(RH)^{3-3\mu^2}]\,,$ 
    with 
    \begin{align*}
        \mathfrak{a} &= ||f_A^{(1)}||_{-\frac{3}{2}+\mu^2}^2\, ||f_A^{(2)}||_{-\frac{3}{2}+\mu^2}^2 \\
        &\qquad - \mathrm{Re}(f_A^{(1)}|f_A^{(2)} )_{-\frac{3}{2}+\mu^2}^2\,.
    \end{align*}
    Since homogeneous Sobolev spaces of order $|s| < 3/2$ are Hilbert spaces with an inner product given in Eq.~\eqref{eq:innerprod}~\cite{bahouri_fourier_2011}, the Cauchy-Schwarz inequality implies that $\mathfrak{a} \geq 0$.  Saturation of this inequality, i.e., $\mathfrak{a} =0$, can only occur if $f_A^{(1)} = f_A^{(2)}$, which corresponds to Case 3, and will be analyzed below. Thus, $ \nu_A^2 -(\nu_{A}^{\mathrm{Mink}})^2 > 0 $ for $RH\gg1$, $mH\ll1$ whenever $0\neq f^{(1)}_A \neq f^{(2)}_A \neq 0$. Furthermore, since $\mathfrak{a} > 0$,  $\nu_A^2-(\nu_{A}^{\mathrm{Mink}})^2$
—and hence $\nu_A$—grows with $H$, which automatically implies that the entropy of the mode $A$ increases as $H$ increases. \\
 
\noindent {\bf Case 2:} Using Eqs.~\eqref{phiphi}-~\eqref{phipi}, one obtains
$\nu_A^2 -(\nu_{A}^{\mathrm{Mink}})^2= \mathfrak{b}\,  (RH)^{2-2\mu^2}(1 + \mathcal{O}(\mu^2)) + \mathcal{O}[(RH)^{1-\mu^2}]$, with $$ \mathfrak{b} = ||f_A^{(2)}||^2_{-\frac{3}{2}+\mu^2} ||g_A^{(1)}||_{\frac{1}{2}}^2 - \mathrm{Re} (f_A^{(2)}|g_A^{(1)})_{-\frac{1-\mu^2}{2}}^2\,. $$
 
Next, we observe that 
\begin{widetext}
    \be \nonumber  \begin{split}
    \mathrm{Re}(f_A^{(2)},g_A^{(1)})_{s}^2 &\leq \left|\int \frac{d^3\,k}{(2\pi)^3} |\vec k|^{2s} \tilde{f}_A^{(2)}(\vec k) (\tilde{g}_A^{(1)})^*(\vec{k}) \right|^2 \\
    &\leq \left( \int \frac{d^3\,k}{(2\pi)^3} |\vec k|^{2s} |\tilde{f}_A^{(2)}(\vec{k})||\tilde{g}_A^{(1)}(\vec{k})|\right)^2 \\
    & =  \left( \int \frac{d^3\,k}{(2\pi)^3} \left(|\vec{k}|^{2s+2s'}|\tilde{f}_A^{(2)}(\vec{k})|^2 \right)^{1/2} \left(|\vec{k}|^{2s-2s'}|\tilde{g}_A^{(1)}(\vec{k})|^2 \right)^{1/2}\right)^2 \\ 
    &\leq  \int \frac{d^3\,k}{(2\pi)^3} |\vec{k}|^{2s+2s'}|\tilde{f}_A^{(2)}(\vec{k})|^2  \int \frac{d^3\,k'}{(2\pi)^3} |\vec{k}'|^{2s-2s'}|\tilde{g}_A^{(1)}(\vec{k}')|^2 = || f_A^{(2)}||_{s+s' }^2||g_A^{(1)}||_{s-s'}^2\,,
\end{split}
 \ee 
\end{widetext}
where we have used H\"older's inequality  in the last line (see Appendix~\ref{app:sobolev}, with $p=q=2$ and $r=1$). This inequality holds for arbitrary $s$ and $s'$. For $s=-1/2+\mu^2{/2}$ and $s'=-1+\mu{^2}/2$, this inequality implies $\mathfrak{b}\geq0$.

\noindent {\bf Case 3:} In this case, $ \nu_A^2 -(\nu_{A}^{\mathrm{Mink}})^2 = \mathfrak{c}\, (RH)^{2-2\mu^2} (1 + \mathcal{O}(\mu^2)) + \mathcal{O}[(RH)^{1-\mu^2}]\,,$ where 
    $$\mathfrak{c}\! =\! ||f_A||^2_{\!-\frac{3}{2}\!+\mu^2} ||g_A^{(1)}\!-g_A^{(2)}||_{\frac{1}{2}}^2\! -\! \mathrm{Re} (f_A,g_A^{(1)}\!-g_A^{(2)})_{\!-\frac{1\!-\mu^2}{2}}^2\,.$$
 Following the same argument as in Case 2 above, we conclude  $\mathfrak{c} \geq 0$. \\

In Cases~2 and~3, the inequalities can be saturated; that is, one can choose functions such that $\mathfrak{b}=0$ and $\mathfrak{c}=0$. In these situations, it becomes necessary to analyze the first non-vanishing contributions to $\nu_A^2 - (\nu_A^{\mathrm{Mink}})^2$. We find that the conditions $\mathfrak{b}=0$ and $\mathfrak{c}=0$ can only occur for a subset of smearing functions $f(\vec{x})$ belonging to the special family defined in Eq.~\eqref{zeroaverage}. A detailed discussion of this case is provided in Appendix~\ref{app:proofs_special_functions}.\\   
\end{proof}

Although Prop.~\ref{prop:vNgeneral} focuses on single modes supported in super-Hubble regions, we conjecture that it remains true for arbitrary values of $RH$. Although we do not have a completely general proof, in 
Appendix~\ref{app:vNspecial} we prove this conjecture for single-mode subsystems of the form 
\(\gamma_{A}^{(1)} = (0, f(\vec{x}))\) and \(\gamma_{A}^{(2)} = (g(\vec{x}), 0)\) ---corresponding to “pure-field’’ and “pure-momentum’’ smeared operators.

\begin{exmp}
\label{ex:singlemodeball_cont}
{\bf Single-mode subsystem supported in a ball (cont.)}\\

This is a continuation of Example~\ref{ex:singlemodeball}.  
To compute the von Neumann entropy of this mode using Eq.~\eqref{eq:S}, we need the symplectic eigenvalue of the reduced state \( \hat{\rho}^{\rm red}_A \), which, as discussed earlier, takes the form
\begin{equation}\label{eq:nuAex1vNS}
    \nu_{A}^2 \!=\! 4\!\braket{\hat{\Phi}(f^{(\delta)})^2}\!\braket{\hat{\Pi}(g^{(\delta)})^2} \!- \!\braket{\{\hat{\Phi}(f^{(\delta)}), \hat{\Pi}(g^{(\delta)})\}}^{\!2}\!\!.
\end{equation}

An analytic expression for \( \nu_A \) is obtained by inserting the second moments from Eqs.~\eqref{phiphi}--\eqref{phipi} into Eq.~\eqref{eq:nuAex1vNS}, and by using the Sobolev norms defined in Eqs.~\eqref{eq:SNm3o2AB}, \eqref{eq:SNm1o2B}, and \eqref{eq:SN1o2B}.

Figure~\ref{fig:SvNball} shows the von Neumann entropy \( S(\nu_A) \) of the single-mode subsystem for $\delta=1$, as a function of the radius \( RH \), and for different values of the mass. The value of \( S(\nu_A) \) for the same mode but in the case of a massless field in Minkowski spacetime is also shown for comparison.

The figure demonstrates that \( S(\nu_A) \) in the Bunch--Davies vacuum grows with $H$ and is always greater than or equal to its value in the Minkowski vacuum, with equality reached in the limit \( RH \ll 1 \). This is expected, as curvature effects become negligible in sufficiently small regions. When \( RH \gg 1 \), the entropy \( S(\nu_A) \) grows logarithmically with \( RH \). Furthermore, \( S(\nu_A) \) increases as $\mu$ decreases, and diverges in the limit $\mu\to 0$.

\begin{figure}
    \centering
    \includegraphics[width= \linewidth]{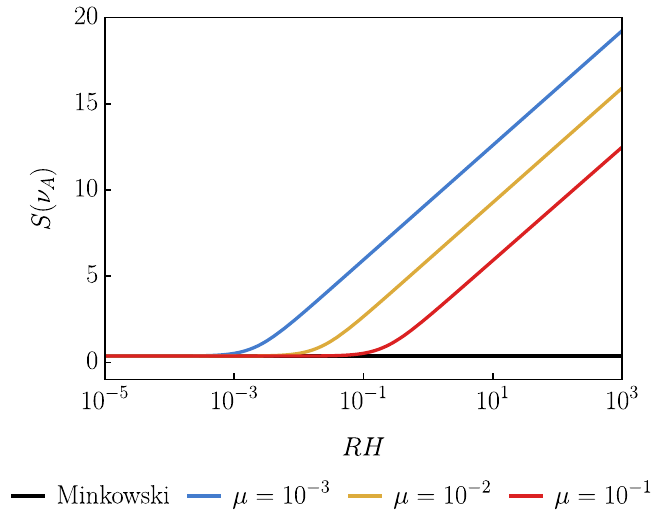}
    \caption{Von Neumann entropy of a single-mode subsystem supported in a ball of radius \( R \), with the mode defined in Eq.~\eqref{eq:ex1vNS}. The entropy is shown for different values of the mass parameter \( \mu \).   This figure shows that $S(\nu_A)$ grows with $H$, except when the region supporting the mode is small compared to the Hubble radius ($RH \ll 1$). In this regime, the entropy becomes independent of $H$ and converges to its Minkowski value. Additionally, the value of \( S(\nu_A) \) increases with decreasing mass, as expected due to the infra-red divergence in the massless limit.}  
    \label{fig:SvNball}
\end{figure}

\end{exmp}

\section{Correlations}
\label{sec:correlations}

As discussed in Sec.~\ref{sec:framework}, mutual information provides a convenient measure of correlations between subsystems. Unlike field-field or momentum-momentum  correlations between two subsystems, mutual information is local-sympletic invariant, meaning it does not rely on the choice of any basis or operators in each sub-system. As already emphasized before, the mutual information of two compactly supported subsystems is not a measure of entanglement, but rather a measure of the total correlations, including classical ones. The reason is that local modes always have mixed reduced states, which can exhibit non-quantum correlations.

For pedagogical purposes, we first illustrate the behavior of mutual information in the Bunch-Davies vacuum for two examples, which can be worked out analytically, after which we will prove the generality of the features displayed in these examples. 

\subsection{Examples\label{subsec:correlations_examples}}

\begin{exmp}
\label{ex:ballballcorrelations}
{\bf Two modes supported in non-overlapping balls (cont.)}\\

This example is a continuation of Example~\ref{ex:singlemodeshell}. The expression for the mutual information \( \mathcal{I}_{AB} \) was given in Eq.~\eqref{eq:def_MI}; it requires the computation of the symplectic eigenvalues \( \nu_A \), \( \nu_B \), and \( \nu_\pm \). Since the two modes are identical, we have \( \nu_A = \nu_B \), and \( \nu_A \) was already computed in Example~\ref{ex:singlemodeball}. 

The eigenvalues \( \nu_\pm \) can be computed using Eq.~\eqref{eq:symplecticeigvalsall}. For this, we need \( \det \sigma_{AB} \) and \( \det C \), where \( \sigma_{AB} \) is the covariance matrix of the joint state, and \( C \) the correlation matrix, defined in Eq.~\eqref{coormatrix}. These determinants can be written analytically using the expressions for the Sobolev inner products in Eqs.~\eqref{eq:SIPs1o2BABB}, \eqref{eq:SIPsm3o2BABB}, and \eqref{eq:SIPsm1o2BABB}, which were presented in Example~\ref{ex:singlemodeshell}.

The expressions for \( \det \sigma_{AB} \) and \( \det C \) are lengthy and not particularly illuminating,  they nevertheless allow us to compute \( \nu_\pm \), and hence the mutual information analytically. The result is plotted in Fig.~\ref{fig:MI_vs_DxH} as a function of $RH$, the mass of the field,  and the physical separation \( |\Delta \vec x| \) between the centers of the two regions. The mutual information for the same pair of modes in the Minkowski vacuum is also shown for comparison.

This figure conveys three key messages:
\begin{enumerate}
\item In the limit \(RH \ll 1\) and \(|\Delta \vec x| H \ll 1\) (i.e., both support and mode–mode separation are sub-Hubble), the value of the mutual information in de Sitter approaches the correponding value in Minkowski spacetime. In this regime, the mutual information falls off as \(|\Delta \vec x|^{-4}\).
\item For \( |\Delta \vec x|H > 1 \), the mutual information scales as \( \mathcal{I}_{AB} \propto (|\Delta \vec x|)^{-4\mu^2} \). Since \( \mu \ll 1 \), this implies that \( \mathcal{I}_{AB} \) is nearly independent of the separation (almost scale-invariant).
\item \(\mathcal{I}_{AB}\) increases with \(RH\), implying that the two modes become more correlated as $H$ grows. This growth, however, saturates at large \(RH\). 
\end{enumerate}

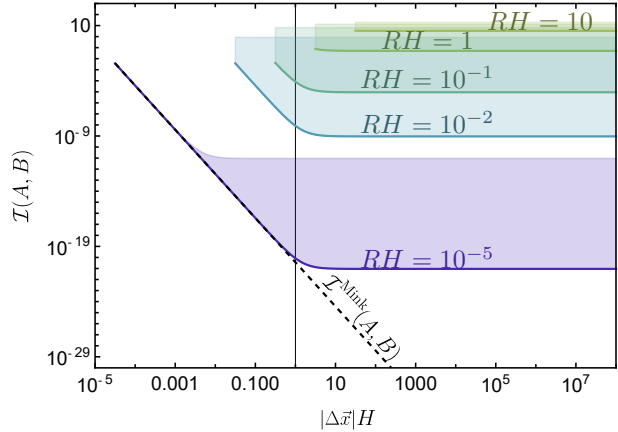
\begin{figure}
    \centering
    \begin{tikzpicture}
        \node (MI) at  (current page.center) {\includegraphics[width=\linewidth]{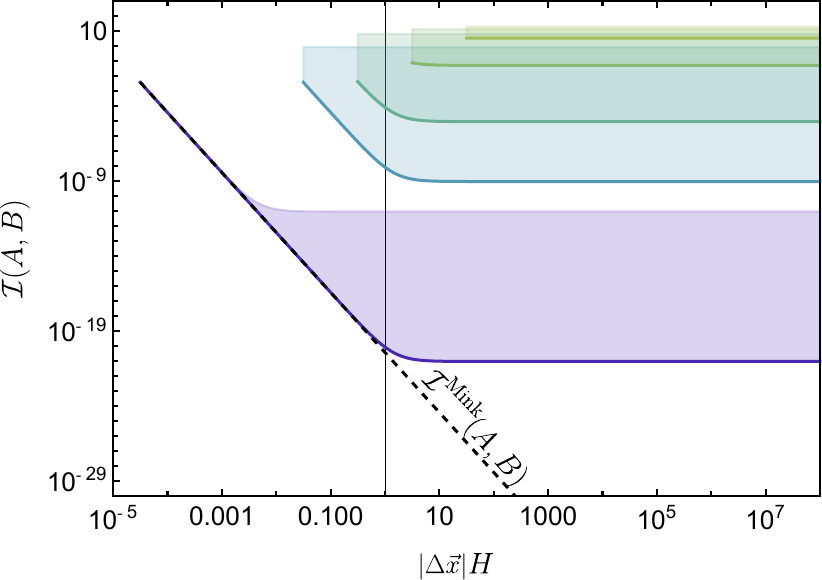}};
        \node[text={rgb:red,74;green,39;blue,178}] at ([xshift=1.5cm,yshift=-0.55cm]MI.center) {$RH = 10^{-5}$};
        \node[text={rgb:red,83;green,155;blue,181}] at ([xshift=1.5cm,yshift=1.35cm]MI.center) {$RH = 10^{-2}$};
        \node[text={rgb:red,104;green,175;blue,147}] at ([xshift=1.5cm,yshift=1.95cm]MI.center) {$RH = 10^{-1}$};
        \node[text={rgb:red,131;green,186;blue,112}] at ([xshift=1.5cm,yshift=2.5cm]MI.center) {$RH = 1$};
         \node[text={rgb:red,163;green,190;blue,86}] at ([xshift=3.cm,yshift=2.75cm]MI.center) {$RH = 10$};   
    \end{tikzpicture}
\caption{Mutual information between two non-overlapping single-mode subsystems supported in balls of radius $R$, plotted as a function of their separation $|\Delta \vec x|$. The modes are the same as those described in Example~\ref{ex:singlemodeshell}, with $\delta = 1$. This figure is obtained  analytically (neglecting the terms $\mathcal{N}_{\Phi\Phi}$, $\mathcal{N}_{\Phi\Pi}$ and $\mathcal{N}_{\Pi\Pi}$ in \eqref{phiphi}-\eqref{phipi}).  All curves are plotted only for $|\Delta \vec x| > 2R$, ensuring that the two modes do not overlap. The shaded region indicates how the mutual information varies with the field mass, showing results for $\mu^2$ ranging from $10^{-10}$ (bottom) to $10^{-15}$ (top). For reference, the gray vertical line corresponds to $|\Delta \vec x|H=1$.}
    \label{fig:MI_vs_DxH}
\end{figure}
\end{exmp}

\begin{exmp}
\label{ex:ballshellcorrelations}{\bf Two modes, one supported in a ball and the other in a spherical shell around it.}\\

In this example, we consider a mode \(A\) supported in a ball and a second mode \(B\) supported on a concentric spherical shell. For \(A\), we take the same mode as in Example~\ref{ex:singlemodeball}, and for \(B\) the same mode as in Example~\ref{ex:singlemodeshell}.  The geometric configuration where the two modes are supported is illustrated in Fig.~\ref{fig:geometric_ball-shell}.
\begin{figure}
    \centering
   \begin{tikzpicture}
       \node (BS) at ([xshift=0.25\textwidth]current page.center) {\includegraphics[width=0.35\textwidth]{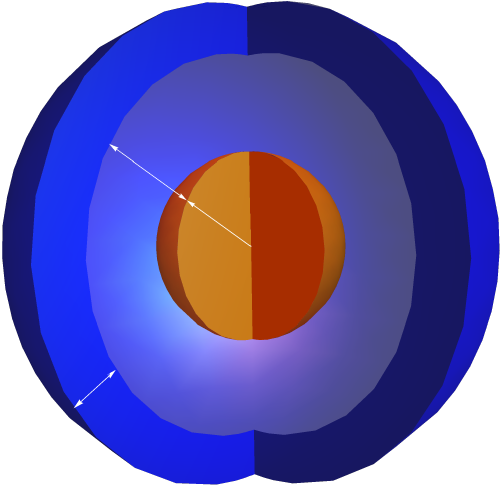}}; 
       \node[rotate=-35] at ([xshift=-1.1cm, yshift=1.1cm]BS.center) {\textcolor{white}{\textbf{$\Delta x$}}}; 
       \node[rotate=-35] at ([yshift=0.5cm,xshift=-0.3cm]BS.center) {\color{white}$R$}; 
       \node[rotate=35] at ([xshift=-2.15cm,yshift=-1.42cm]BS.center) {\color{white}$2d$}; 
   \end{tikzpicture}
    \caption{Geometric configuration of the modes in Example~\ref{ex:ballshellcorrelations}.  The shell has inner and outer radii $R_S \pm d$, where $2d$ is its width and  $R_S = R + \Delta x$. Here, $\Delta x$ denotes the radial distance between the edge of the ball (of radius $R$), and the inner edge of the shell.}
    \label{fig:geometric_ball-shell}
\end{figure}
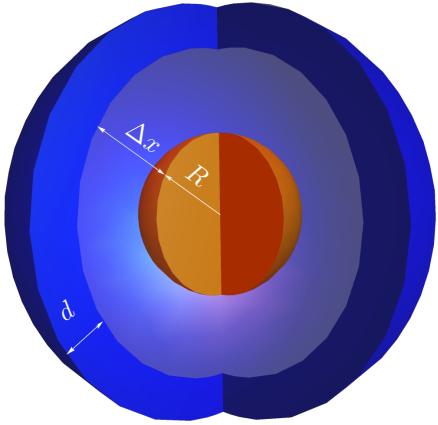

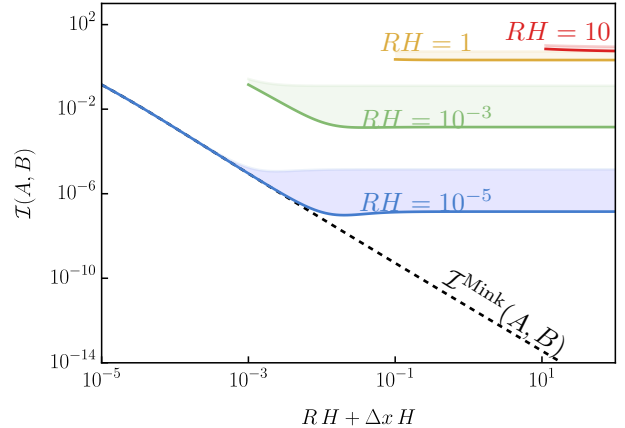
\begin{figure}[t!]
    \centering
    \begin{tikzpicture}
        \node at (0,0) { \includegraphics[width=\linewidth]{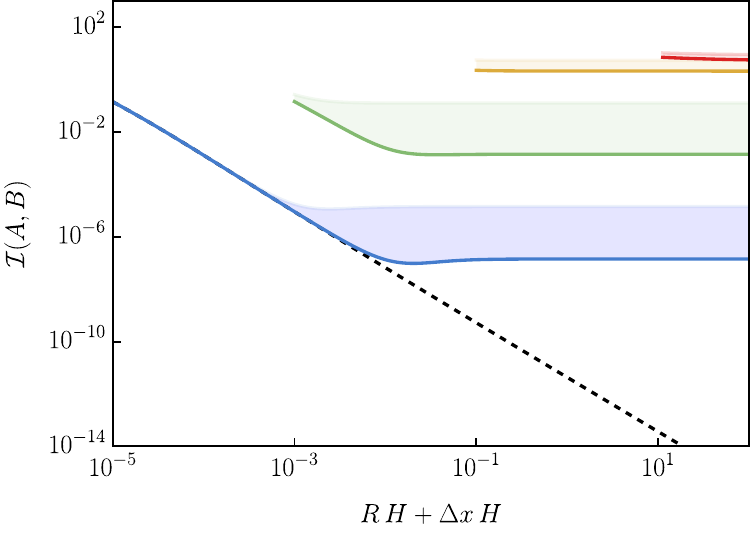}};
        \node[rotate=-32] at (2.55,-1.25) {$\mathcal{I}^{\mathrm{Mink}}(A,B)$};
        \node at (1.5,0.3) {\color{b1}$RH=10^{-5}$}; 
        \node at (1.5, 1.55) {\color{g2}$RH =10^{-3} $ };
        \node at (1.5, 2.45) {\color{y3}$RH=1$}; 
        \node at (3.25, 2.65) {\color{r4} $RH = 10$}; 
    \end{tikzpicture}
   
    \caption{ Mutual information of a single-mode subsystem supported in a ball of radius $R$ and another single-mode supported in a concentric shell of radial width  $2d$ and inner radius $\Delta x$, as a function of $\Delta xH$. The mode supported in the ball is defined as in Example~\ref{ex:singlemodeball},  with $\delta =2$, and the mode in the shell is defined as in Example~\ref{ex:singlemodeshell}.  The shaded region shows the variation of the mutual information when $\mu^2$ changes from  $\mu^2 = 10^{-4}$ (bottom) to  $\mu^2 = 10^{-6}$ (top). }
    \label{fig:MI_ball&shell}
\end{figure}

Figure~\ref{fig:MI_ball&shell} shows the mutual information between these two modes, with the field prepared in the Bunch--Davies vacuum, as a function of the radial separation \(\Delta x\) between the ball and the shell. The calculation can be performed analytically---although the resulting expressions are lengthy.  

The mutual information exhibits essentially the same qualitative behavior as in the previous example, with one notable difference: when both the support of the two modes and their separation are sub-Hubble modes—the regime where $\mathcal{I}(A,B)$ agrees with that found in Minkowski spacetime—one finds \(\mathcal{I}_{AB} \propto (\Delta  x)^{-2}\), instead of \(\propto (\Delta  x)^{-4}\) as in the previous case. The difference arises from the specific geometric configuration considered here: the shell mode \(B\) completely surrounds the ball mode \(A\), thereby enhancing the correlations between them and slowing down the rate of fall-off.

  \end{exmp}

\subsection{General result\label{subsec:correlations_general}}
 Let $(\hat{O}_A^{(1)}, \hat{O}_A^{(2)})$ and $(\hat{O}_B^{(1)}, \hat{O}_B^{(2)})$ be two space-like separated, single-mode subsystems defined from phase space elements 
 \begin{equation}
     \gamma_{I}^{(i)} = (g_{I}^{(i)}(\vec{x}),f_{I}^{(i)}(\vec{x}))\,, ~ I=A,B\,  \text{ and } i = 1,2\,. 
 \end{equation}
In this section, we analyze some asymptotic properties of the mutual information for these non-overlapping  modes.

\begin{prop} {\bf Large-separation behavior of $\mathcal{I}(A,B)$.} 
    Consider a real, light scalar field theory in the Bunch-Davies vacuum. The mutual information of two typical---i.e., non-special---single-mode subsystems compactly supported in regions distant from each other (specifically, $|\Delta \vec x|H\gg1$ and { $|\Delta \vec x|/R\gg 1$})  grows monotonically with $H$.
   
\end{prop}

\begin{proof}

We focus on pairs of modes for which a meaningful notion of distance can be defined---namely, when the support of these modes can be separated continuously without deforming or intersecting them (this is not possible, for instance, for the ball-shell configuration discussed above). When the shape of the modes precludes defining a distance between them, the fall-off behavior of the mutual information must be analyzed separately.

Recall that  
\begin{equation}
  \mathcal{I}(A,B) = S(\nu_A) + S(\nu_B) - S(\nu_+) - S(\nu_-)\,,
\end{equation}
where only \( \nu_\pm \) depend on the separation between the two modes.
Using Simon’s normal form of the covariance matrix, one finds
\begin{widetext}
\begin{equation}
\nu_{\pm}^2 = 
\frac{1}{2} \left(
\nu_A^2 + \nu_B^2 + 2 c_- c_+ 
\pm 
\sqrt{(\nu_A^2 - \nu_B^2)^2 
+ 4 c_+ c_- (\nu_A^2 + \nu_B^2)
+ 4 \nu_A \nu_B (c_+^2 + c_-^2)}
\right)\,,
\end{equation}
\end{widetext}
where $c_{\pm}$ denote the components of the cross-correlations matrix $C$ when written in Simon's normal form (see \ref{app:simon_form}). 

The combinations \( c_+ c_- = \det C \) and \( c_+^2 + c_-^2 \) can be expressed in terms of the correlators in 
Eqs.~(\ref{phiphi}--\ref{phipi}), whose large-separation behavior, $|\Delta \vec x|H\gg1$, can be obtained using  standard tools of asymptotic harmonic analysis {(see Appendix~\ref{app:asymptotic_sobolev} and, e.g.~\cite{bender1999advanced} for more details).}

{ The key observation is that the} large-separation behavior of \( c_- c_+ \) and \( c_+^2 + c_-^2 \), and consequently of the 
symplectic eigenvalues \( \nu_\pm \), is dominated by 
\( \mathrm{Re}(f_{A}^{(i)},f_B^{(j)})_{-\frac{3}{2} + \mu^2} { \sim (|\Delta \vec x|)^{-2\mu^2}}\).  This implies that
\begin{equation}\label{Iassymp}
\mathcal{I}(A,B) \sim { \alpha(H)}\, (|\Delta \vec x|)^{-4\mu^2}\,,
\end{equation}
{  where $\alpha$ is a monotonically growing function of  $H$; its form lengthy, and not particularly illuminating.} 
This  shows that \( \mathcal{I}(A,B) \) becomes nearly scale-invariant and grows monotonically 
with \( H \) at large separations.

{  In the limit $H \to 0$ and {  $m\to 0$ (in this order)}, the term including \( \mathrm{Re}(f_{A}^{(i)},f_B^{(j)})_{-\frac{3}{2} + \mu^2} \) vanishes, and instead $(f_{A}^{(i)},f_B^{(j)})_{-\frac{1}{2}} $ becomes dominant. One finds  $(f_{A}^{(i)},f_B^{(j)})_{-\frac{1}{2}} \sim |\Delta \vec x|^{-2}$ (see Appendix~\ref{app:asymptotic_sobolev}), recovering the well-known result $\mathcal{I}^{\mathrm{Mink}}(A,B) \sim (\Delta x_{AB})^{-4}$ in Minkowski spacetime~\cite{martin_real-space_2021,Martin:2021qkg,Shiba:2012np}.}

In this proof we have assumed that the functions \( f_{I}^{(i)} \) do not belong to the special class of functions defined in 
Eq.~\eqref{zeroaverage}; Appendix~\ref{app:proofs_special_functions} discusses these special cases in detail.
\end{proof}

The previous proposition describes the large-separation behavior of $\mathcal{I}(A,B)$, and applies irrespective of the size \( R \) of the supports of the two modes relative to $H$. 
It is also interesting to explore a different limit, namely when the supports of the two modes are 
super-Hubble, \( RH\gg 1 \), irrespective of the distance between them. 

\begin{prop}
    In the limit $RH\gg 1$, $\mathcal{I}(A,B)$ becomes independent of $R$ and depends logarithmically on the separation between the regions. 
\end{prop}

\begin{proof}

We showed in Proposition~2 that, in the regime $RH \gg1$, $\nu_I^{2} \gg 1$, $I=A,B$. Consequently, the mutual information can be approximated as: 
\begin{equation}\label{IRHgg1}
     \mathcal{I}(A,B) \underset{RH\gg 1}{\sim} \frac{1}{2} \log_2\left(\frac{\nu_A^2 \nu_B^2}{\nu_+^2 \nu_-^2}\right).
\end{equation} 

Using Simon's normal form (see Appendix \ref{app:simon_form}), \eqref{IRHgg1} can be re-written as
\begin{equation}
   \mathcal{I}(A,B) \!\!\underset{RH\gg 1}{\sim}\!\!\!  -\frac{1}{2} \log_2\!\left(\!1 +\frac{c_-^2c_+^2}{\nu^2_A\nu^2_B} - \frac{c_-^2 + c_+^2}{\nu_A \nu_B} \!\right). 
\end{equation}
In the regime $RH \gg 1$, $c_\pm$ can be expanded as polynomials of $RH$.  The leading order dependence in $RH$ cancels out in the argument of the logarithm, making $\mathcal{I}(A,B)\sim \mathcal{O}((RH)^0)$, i.e., independent of $R H$, in the limit $RH \gg 1$.

The dependence on the separation $|\Delta \vec x|$  in the limit $RH \gg 1$ is 
\begin{equation}\label{Ilog}
     \mathcal{I}(A,B) {\sim} -\log_2(1-\beta\, |\Delta \vec x|^{-4\mu^2})\, , \end{equation}
where $1>\beta |\Delta \vec x|^{-4\mu^2}>0$---the concrete value depends on the choice of smearing functions. 
This logarithmic dependence on the distance was first reported in~\cite{Martin:2021qkg} using a  specific family of single-mode subsystems.

Notice that for large separations $|\Delta \vec x|H\gg 1$, the expansion of the logarithm in \eqref{Ilog} produces $ \mathcal{I}(A,B) {\sim}  |\Delta \vec x|^{-4\mu^2}$, consistently recovering the almost scale invariance at super-Hubble separations. 
\end{proof}

\section{Entanglement between two local degrees of freedom}\label{sec:LN2dof}

In the previous section we showed that de Sitter curvature enhances correlations between pairs of compactly supported, non-overlapping modes. In this section, we turn to quantum correlations, focusing specifically on entanglement. In contrast to correlations, we find that de Sitter curvature actually {\em decreases} the entanglement between compactly supported modes.  

As in the previous section, we begin with two simple yet illustrative examples—the same ones considered in the previous section—in which the entanglement can be computed analytically, and then proceed to provide a general proof.

\subsection{Examples\label{subsec:LN2dof_examples}}

As explained in Sec.~\ref{entsubs}, in the Bunch--Davies vacuum the entanglement between any two modes is nonzero if and only if \(\tilde{\nu}_- < 1\), and it decreases monotonically with \(\tilde{\nu}_-\) (see Eq.~\eqref{eq:symplecticeigvalsall} for the definition of \(\tilde{\nu}_-\)).\\

\begin{exmp} \label{ex:ballballLN}{\bf Two modes supported in non-overlapping balls (cont.)}\\

After a tedious calculation we conclude that, in  the small mass limit $\mu \ll 1$, $\tilde{\nu}_{-}$ can be written as
\begin{equation} \label{eq:nuPTdS}
\tilde \nu_{-} = \sqrt{(\tilde \nu_-^{\mathrm{Mink}})^2 + (RH)^2 \tilde{\mathcal{F}}_{-}^{(\delta)}} \,,
\end{equation}
where
\begin{widetext}
\begin{equation} \label{eq:ftpm_example}
\begin{split}
R^{2-2\mu^2}\tilde{\mathcal{F}}_{-}^{(\delta)} =& \left(\|{g}_A^{(\delta)}\|_{\frac{1}{2}}^2 +\mathrm{Re}({g}^{(\delta)}_A, {g}^{(\delta)}_B )_{\frac{1}{2} } \right)\left(\|f_A^{(\delta)}\|_{-\frac{3}{2} + \mu^2}^2 -  \mathrm{Re}(f^{(\delta)}_A, f^{(\delta)}_B )_{-\frac{3}{2} +\mu^2} \right) \\
& - \left(\ { \mathrm{Re}\left[(f^{(\delta)}_A, g^{(\delta)}_A)_{-\frac{1}{2} + \frac{\mu^2}{2}}\right]^2 } - {\mathrm{Re}\left[(f^{(\delta)}_A, g^{(\delta)}_B )_{-\frac{1}{2} +\frac{\mu^2}{2}} \right]^2}\right)\\ 
& -   2 \, \frac{\left(\|{g}_A^{(\delta)}\|^2_{\frac{1}{2}}\,\mathrm{Re}(f^{(\delta)}_A, {g}^{(\delta)}_B )_{-\frac{1}{2} + \frac{\mu^2}{2}} - {\mathrm{Re}(f_A^{(\delta)},g_A^{(\delta)})_{-\frac{1}{2} + \frac{\mu^2}{2}}}\,\mathrm{Re}({g}^{(\delta)}_A, {g}^{(\delta)}_B )_{\frac{1}{2} }\right)}{\mathrm{Re}(f^{(\delta)}_A, f^{(\delta)}_B )_{-\frac{3}{2}  + \mu^2} \|{g}_A^{(\delta)}\|^2_{\frac{1}{2}} - \|f_A^{(\delta)}\|_{-\frac{3}{2} + \mu^2}^2\,\mathrm{Re}({g}^{(\delta)}_A,{g}^{(\delta)}_B )_{\frac{1}{2} } } \\
& \quad \times \left( \|f_A^{(\delta)}\|^2_{-\frac{3}{2} + \mu^2}\,\mathrm{Re}(f^{(\delta)}_A, {g}^{(\delta)}_B )_{-\frac{1}{2} +\frac{\mu^2}{2}} -  \mathrm{Re}(f^{(\delta)}_A, f^{(\delta)}_B )_{-\frac{3}{2} +\mu^2}  {\mathrm{Re}(f_A^{(\delta)}, g_A^{(\delta)})_{-\frac{1}{2} +\frac{ \mu^2}{2}}} \right) \\
& + \mathcal{O}(\mu^2)\,.
\end{split}
\end{equation}
\end{widetext}
In Eq.~\eqref{eq:nuPTdS}, $\tilde \nu_-^{\mathrm{Mink}}$ is the value of $\tilde \nu_-$ for a massless field in the Minkowski vacuum. Thus, this equation conveniently separates the contribution to $\tilde \nu_-$ that is genuinely due to spacetime curvature. 

Substituting into this expression the Sobolev norms and inner products from Eqs.~\eqref{eq:SN1o2B},~\eqref{eq:SNm1o2B},~\eqref{eq:SNm3o2AB},~\eqref{eq:SNm1o2pmu2AB}, and \eqref{eq:SIPs1o2BABB}–\eqref{eq:SIPsm1o2BABB}, we obtain a closed expression for $\tilde \nu_-$ in the Bunch--Davies vacuum.

For the example of two modes supported in non-overlapping spherical regions, we find that \(\tilde{\nu}_-^{\mathrm{Mink}}\) is larger than one for all separations \(|\Delta \vec x| > 2R\) and all values of \(\delta\). This implies that the two modes are \emph{not} entangled in the Minkowski vacuum~\cite{ubiquitous}. Nevertheless, the modes are still correlated (as shown in the previous section), indicating that these correlations do not correspond to entanglement. Note also that the entropy of each individual mode is nonzero as we showed in Example~\ref{ex:singlemodeball_cont}, which means that each mode must be entangled with \textit{some} other modes.  

In the Bunch--Davies vacuum, whether the modes are entangled depends on the sign of \(\tilde{\mathcal{F}}_-^{\delta}\) given in \eqref{eq:ftpm_example}. For entanglement to occur, this quantity must be negative, ensuring that \(\tilde{\nu}_- < 1\). However, we find that \(\tilde{\mathcal{F}}_-^{\delta}\) is strictly positive in the regime \(\mu \ll 1\), for all \(\delta\) and for all separations \(|\Delta \vec x| > 2R\). This behavior is illustrated in Fig.~\ref{fig:modification_nuPT_deSitter}.

We conclude that curvature effects do not generate entanglement between the specific two modes considered in this example. On the contrary, curvature pushes the system further away from being entangled, since \(\tilde{\mathcal{F}}_-^{\delta}\) is positive, and $\tilde \nu_-$ grows with \(H\)—as follows from \eqref{eq:nuPTdS}, together with the fact that \(\tilde{\mathcal{F}}_-^{\delta}\) is independent of \(H\). Note that curvature pushes the system away from being entangled even though it simultaneously enhances the correlations between modes.

\begin{figure}
    \centering
    \includegraphics[width=0.5\textwidth]{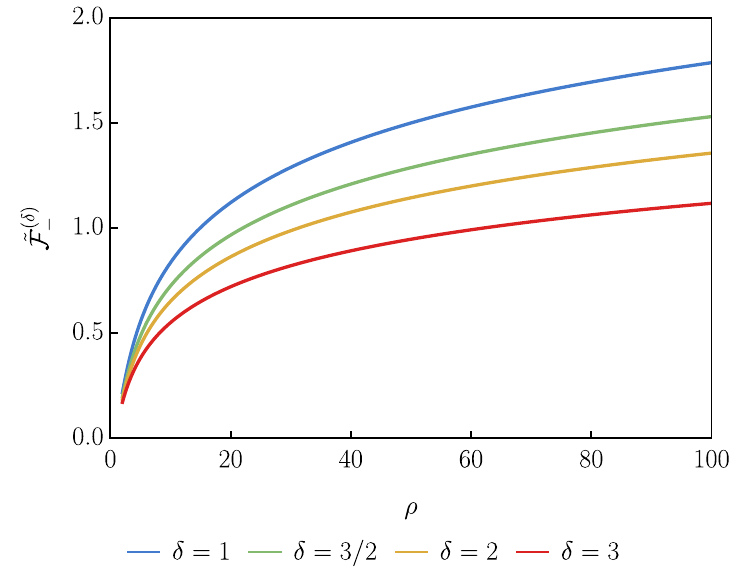}
    \caption{$\tilde{\mathcal{F}}^{(\delta)}_-$ as a function of $\rho=|\Delta \vec x|/R$, where 
   $|\Delta \vec x|$ denotes separation between the centers of the two modes ($RH =1$ and $\mu= 10^{-20}$).}
    \label{fig:modification_nuPT_deSitter}
\end{figure}
\end{exmp}

\begin{exmp}\label{ex:ballshellLN}{\bf Two modes supported in a ball and a surrounding spherical shell, respectively (cont.).}\\

This is a continuation of Example~\ref{ex:ballshellcorrelations}. Equations~\eqref{eq:nuPTdS} and~\eqref{eq:ftpm_example} also apply here, with the Sobolev norms and inner products in Eq.~\eqref{eq:ftpm_example} replaced by those given in Examples~\ref{ex:singlemodeball},~\ref{ex:singlemodeshell},   and~\ref{ex:ballshellcorrelations}.

This case is interesting because $\tilde \nu_-^{\mathrm{Mink}} < 1$ for a certain range of $\Delta x$ (see Fig.~\ref{fig:LN_BallShell}), meaning that the two modes are entangled in the Minkowski vacuum. In this ball-shell configuration, the two modes are ``closer together'' compared to the ball-ball configuration, since the shell completely surrounds the ball, which intuitively explains why there is more entanglement.

We find that the function $\tilde{\mathcal{F}}^{(\delta)}_-$ in Eq.~\eqref{eq:nuPTdS} is always positive in this configuration (and $H$-independent). It then follows from \eqref{eq:nuPTdS} that entanglement decreases monotonically as $H$ increases. This behavior is opposite to that of correlations, which increase with $H$. This example highlights a subtle and non-intuitive interplay between curvature, correlations, and entanglement.

Figure~\ref{fig:LN_BallShell} shows the logarithmic negativity (LN) of the system as a function of the radial separation between the ball and the shell, for various values of $H$. A reference curve corresponding to the Minkowski vacuum result is also included.

This figure displays the following behaviors:
\begin{itemize}
    \item[(i)] LN decays exponentially with the radial distance $\Delta x$ between the shell and ball modes, vanishing beyond a certain threshold distance;
    \item[(ii)] LN approaches the Minkowski result in the limit $RH \to 0$;
    \item[(iii)] LN decreases with increasing $H$, and vanishes entirely beyond a certain threshold value of $RH$.
\end{itemize}

\begin{figure}
    \centering
    \begin{tikzpicture}
       \node (C) at ([xshift=-0.25\textwidth]current page.center){ \includegraphics[width=0.495\textwidth]{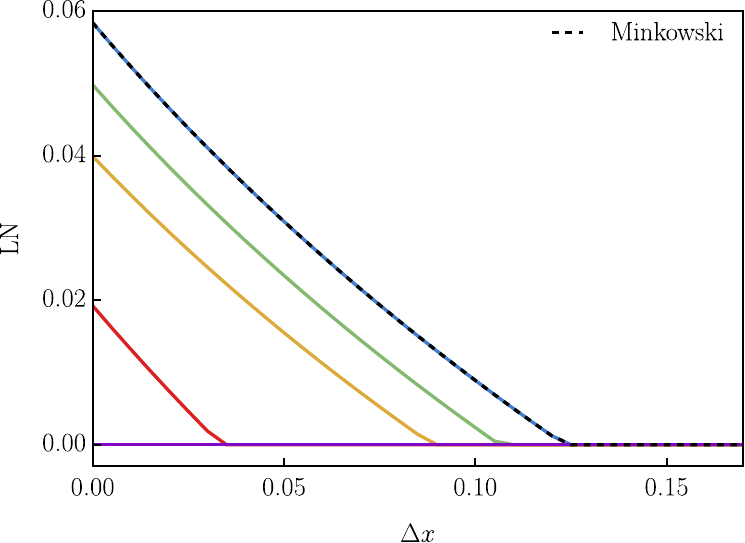}};
       \node at ([xshift=1.5cm,yshift=1.2cm]C.center) {\color{white}$\Delta x$}; 
       \node[fill=white,rotate=-45,inner sep=0pt] at ([xshift=-1.5cm,yshift=1.25cm]C.center) {\color{b1} $RH = 10^{-5}$ };
       \node[fill=white,rotate=-37,inner sep=0pt] at ([xshift=.cm,yshift=-0.75cm]C.center) {\color{g2} $RH = 10^{-2}$ };
       \node[fill=white,rotate=-44,inner sep=0pt] at ([xshift=-1.75cm,yshift=0.05cm]C.center) {\color{y3} $RH = 10^{-1}$ };
       \node[fill=white,rotate=-45,inner sep=0pt] at ([xshift=-2.35cm,yshift=-1.4cm]C.center) {\color{r4} $RH = 1$ };
       \node[fill=white,inner sep=1pt] at ([xshift=3cm,yshift=-2cm]C.center) {\color{l5} $RH = 10$ };
 \end{tikzpicture}
    \caption{Logarithmic negativity between two single modes compactly supported, respectively,  within a sphere of radius $R$ and a surrounding spherical shell with inner and outer radii $R_S \pm d$, with $R_S = R + \Delta x$ and $d = 0.3R$.  We use $\mu = 10^{-2}$ in this plot. The plot shows that LN decreases with both $H$ and $|\Delta  x|$ and vanish beyond a certain threshold values.}
    \label{fig:LN_BallShell}
\end{figure}
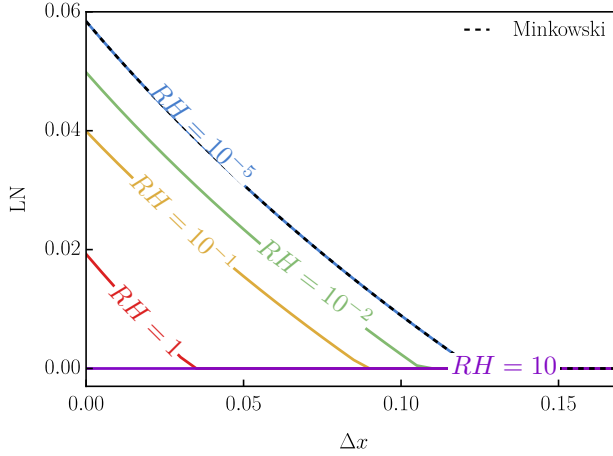

\end{exmp}

\subsection{General results\label{sec:generalLN}} 

\begin{prop}\label{prop:LN_general}
    Consider a real scalar field of mass $m \ll H$ in the Bunch--Davies vacuum, and let $A$ and $B$ be \textit{any} two non-overlapping single-mode subsystems compactly supported in super-Hubble regions, i.e., $RH \gg 1$. Then the entanglement between $A$ and $B$ increases monotonically with $H$.

\end{prop}
\begin{proof}
As explained above, this  Proposition is equivalent to saying  $\tilde \nu_-$ grows monotonically with $H$. Let $(\gamma_A^{(1)},\gamma_A^{(2)})$ and $(\gamma_B^{(1)},\gamma_B^{(2)})$ be
    bases of the symplectic subspaces $\Gamma_A$ and $\Gamma_B$ defining the two modes under consideration.

    If $\gamma_I^{(i)}=(g_I^{(i)},f_I^{(i)})$, $I=A,B$, $i=1,2$, the proof must  differentiate  three different cases: 
    \begin{enumerate}
        \item $f^{(i)}_I \neq 0$ and $f^{(1)}_I \neq f^{(2)}_I$, $I=A,B$.
        \item Either $f^{(1)}_{I} = 0$, or $f^{(2)}_I =0$ for both $I =A$ and $I=B$. 
        \item $f_I^{(1)}=f_I^{(2)}$, $I=A,B$, 
    \end{enumerate}
    
    {\bf Case 1:} $\tilde{\nu}_-$  can be written in terms of  symplectic invariants $\tilde{\Delta}$ and $\det \sigma_{AB}$ as in Eq.~\eqref{eq:symplecticeigvalsall}. In the regime $RH \gg 1$, the leading order contributions to these symplectic invariants can be obtained as follows. The components of the covariance metric can be written as 
\begin{equation}
    {\bm\sigma}_{AB}\! -\! ({\bm\sigma}_{{AB}})_{\mathrm{Mink}} \!=\!(RH)^{{ 2-2\mu^2}}  {\bm F} + \mathcal{O}[(RH)^{  { 1-\mu^2}}]\,,
\end{equation}
where \be 
{\bm F} : = \left(\begin{matrix}
        {\bm F_A} & {\bm F_C}\\ 
        {\bm F_C}^{\mathrm{T}}&  {\bm F_B}
    \end{matrix}\right)
\ee
is a $4\times 4$ matrix made of the following blocks:
\begin{widetext}
    \begin{equation}\label{eq:FI_general_case1}
    {\bm F_I} = { R^{-2+2\mu^2}} \left(\begin{matrix}
        ||f_I^{(1)}||_{-\frac{3}{2}+\mu^2}^2 & \mathrm{Re}\left[(f_I^{(1)},f_{I}^{(2)})_{-\frac{3}{2}+\mu^2}\right] \\ 
        \mathrm{Re}\left[(f_I^{(1)},f_I^{(2)})_{-\frac{3}{2}+\mu^2}\right] & ||f_I^{(2)}||_{-\frac{3}{2}+\mu^2}^2
    \end{matrix}\right) (1+\mathcal{O}(\mu^2))\,,
\end{equation}
with $I = A,B$, and 
\begin{equation}
     {\bm F_C} = { R^{-2+2\mu^2}} \left(\begin{matrix} 
    \mathrm{Re}\left[(f_A^{(1)},f_{B}^{(1)})_{-\frac{3}{2}+\mu^2}\right] & \mathrm{Re}\left[(f_A^{(1)},f_{B}^{(2)})_{-\frac{3}{2}+\mu^2}\right]\\ 
    \mathrm{Re}\left[(f_A^{(2)},f_{B}^{(1)})_{-\frac{3}{2}+\mu^2}\right] & \mathrm{Re}\left[(f_A^{(2)},f_{B}^{(2)})_{-\frac{3}{2}+\mu^2}\right]
    \end{matrix}\right) (1+\mathcal{O}(\mu^2))\,.
\end{equation}
\end{widetext}
(The functions $g_I^{(i)}$ contribute to \eqref{eq:FI_general_case1}  at next-to-leading order.) On the other hand, 
\begin{equation}\begin{split}
    \tilde{\Delta}\! -\! \tilde{\Delta}_{\mathrm{Mink}} \!=&(RH)^{4-4\mu^2}(\det {\bm F_A} + \det {\bm F_B} \\
    &- 2 \det {\bm F_{C}}) + \mathcal{O}[(RH)^{3-3\mu^2}]\,. 
    \end{split}
\end{equation}
It follows from the Cauchy-Schwarz inequality that the matrices ${\bm F_A}$ and ${\bm F_B}$ are positive definite and, consequently, $\det {\bm F_A}>0$ and $\det {\bm F_B}>0$. 

If $\det{\bm F_C}$ were equal or larger than zero, then the Gaussian separability lemma (see Appendix~\ref{app:sobolev}) would guarantee that the two modes are not entangled, and the Proposition is trivially satisfied.  Therefore, we restrict attention to configurations for which $\det {\bm F_{C}} <0$. Using the expressions above we find
\be \label{eq:nutpm2_minus_Mink}\begin{split}
    \tilde \nu_-^2-(\tilde \nu^{\rm Mink}_-)^2&=(RH)^{4-4{\mu}^2} \tilde{\mathcal{F}}_-\, (1 + \mathcal{O}({\mu}^2)) \\ 
    &+ \mathcal{O}[(RH)^3]\,,
\end{split}
\ee
when $RH\gg 1$, where 
\begin{widetext}
    \be \label{eq:mathcalFtm_general} \tilde{\mathcal{F}}_- = \frac{1}{2}\left(\det {\bm F_A} + \det {\bm F_B} - 2 \det {\bm F_{C}} - \sqrt{(\det {\bm F_A} + \det {\bm F_B} - 2 \det {\bm F_{C}})^2 - 4 \det {\bm F} }\right) \,.\ee 
\end{widetext}
It follows from the positivity of ${\bm F_I}$ with $I=A,B$ and $\det {\bm F_{C}}<0$ that the first three terms in \eqref{eq:mathcalFtm_general} are all positive.

Since the block-matrix ${\bm F}$ is positive semi-definite (as can be shown using Simon's normal form and the requirement that the symplectic eigenvalues of the two-mode system be real), one can use Fischer's inequality (see Appendix~\ref{app:sobolev}) to find the upper bound $\det {\bm F} \leq \det {\bm F_A}\det {\bm F_B}$, with the inequality being saturated if and only if ${\bm F_C } =0$. Using this upper bound, we conclude that the argument inside the square root is positive and the square root is smaller than the sum of the first three terms in Eq.~\eqref{eq:mathcalFtm_general}. Thus, $\tilde{\mathcal{F}}_- >0$ and $H$-independent, so  $\tilde \nu_-^2-(\tilde \nu^{\rm Mink}_-)^2$  is positive and grows with $H$ in the regime $RH\gg 1$. \\

{\bf Cases 2 and 3:} As shown in the proof of Prop.~\ref{prop:vNgeneral}, Cases 2 and 3 yield the same leading-order behavior in the regime $RH \gg 1$ for the symplectic eigenvalue $\nu_I$ of any single-mode subsystem. The same argument applies here:  From the symplectic invariants in Eq.~\eqref{eq:symplecticeigvalsall}, one sees that Cases 2 and 3 also produce identical leading-order contributions to the logarithmic negativity. Therefore, we restrict attention to Case~2.  

The proof of the proposition in Case~2 proceeds along the same lines as in Case~1, with the only difference being that the leading-order terms in ${\bm \sigma}_{AB}$ and $\tilde{\Delta}$, in an expansion for small $RH$, scale as $(RH)^{2-2\mu^2}$ (rather than $(RH)^{4-4\mu^2}$). To establish the positive definiteness of the matrices analogous to those in Eq.~\eqref{eq:FI_general_case1}, one uses Prop.~10 in Appendix~\ref{app:sobolev}, instead of relying on the Cauchy--Schwarz inequality.

There exist other cases not covered in the classification, like when only one of the functions $f^{(i)}_I$ with $i=1,2$ and $I=A,B$ vanishes or when the functions used to define the two modes subsystems are identical in shape (although keeping them non-overlapping). One can easily extend this proof using very similar tools.

See Appendix~\ref{app:proofs_special_functions} for details of this proof when the functions $f_I^{(i)}$ belong to the special family defined in \eqref{zeroaverage}.

\end{proof}

As in previous sections, although our formal proof is restricted to modes with super-Hubble support, we conjecture that the result applies to any pair of non-overlapping modes. Appendix~\ref{app:extension_LN_allRH} shows that this conjecture holds for single-mode subsystems of the form 
\(\gamma_{A}^{(1)} = (0, f(\vec{x}))\) and \(\gamma_{A}^{(2)} = (g(\vec{x}), 0)\), corresponding to “pure-field’’ and “pure-momentum’’ smeared operators, respectively.

\section{Partner systems and entanglement distribution\label{sec:partners}}

The results obtained so far may appear to be in tension. On the one hand, we have shown in Sec.~\ref{sec:vNEntropy} that increasing $H$ enhances the von Neumann entropy of localized modes. Since the entropy of a mode quantifies its entanglement with the rest of the field degrees of freedom, a larger $H$ implies that the field theory as a whole contains more entanglement. On the other hand, we have found in the last section that increasing $H$ reduces the entanglement between pairs of localized modes.

In this section, we argue that these two results are compatible with each other. Even more, we argue that the former is in fact a direct consequence of the latter.

The apparent contradiction is resolved by analyzing how correlations and entanglement are spatially distributed in the Bunch--Davies vacuum. In particular, we show that the notion of a \emph{partner mode} \cite{hotta2015partner,Trevison_2019,hackl_minimal_2019,partnerformula} provides a convenient framework to characterize the spatial structure of entanglement in quantum field theory in the cosmological patch of de Sitter spacetime. Partner modes were previously discussed in de Sitter spacetime in \cite{Nambu_2023} for a subset of special modes belonging to the family introduced in a Remark on page \pageref{remark_IR}.  Here, we discuss universal aspects of partner modes in the cosmological patch of de Sitter spacetime.

Consider a single-mode subsystem $A$ compactly supported, and assume that the field  is prepared in a Gaussian state, either pure or mixed. Because $A$ is compactly supported, its  reduced state $\hat \rho^{\text{red}}_A$ is mixed. This in turn implies that $A$ is entangled with the rest of the system.  Interestingly,  there exists a unique single-mode $A_p$, different from and independent of $A$, that encodes all entanglement with $A$. In particular, when the state of the field is pure---e.g the Bunch-Davies vacuum---$A_p$ ``purifies'' $A$, in the sense that  the reduced state of the two-mode system $(A, A_p)$ is pure---and the von Neumann entropy of $A$ quantifies the entanglement between $A$ and $A_p$.

The concept of a partner mode was introduced in~\cite{hotta2015partner} in the context of Hawking radiation, and was further developed in~\cite{Trevison_2019} and \cite{hackl_minimal_2019}. A formal treatment, including a proof of existence and uniqueness as well as generalizations to mixed Gaussian states, is given in~\cite{partnerformula}. In this section, we follow the approach of~\cite{partnerformula} and restrict to pure Gaussian states.

Naturally, the partner mode $A_p$ depends on the specific choice of mode $A$. Our interest, however, lies in the universal features of $A_p$ that hold for any localized mode $A$. In particular, we focus on the asymptotic behavior of the partner mode at large distances from the region where $A$ is supported. We argue that these fall-off properties are universal and encode invariant information about the spatial distribution of correlations and entanglement.

The partner of a given mode $A$ can be determined as follows. Let $A$ be a single-mode subsystem characterized by the classical symplectic subspace $\Gamma_A$. Let $\Pi_A$ be the symplectic-orthogonal projector onto $\Gamma_A$ in the classical phase space.\footnote{Because $\Gamma_A$ is symplectic, the phase space $\Gamma_{\sigma}$ can be decomposed as $\Gamma_{\sigma} = \Gamma_A \oplus \Gamma_{\bar A}$, where $\Gamma_{\bar A}$ is the symplectic orthogonal complement of $\Gamma_A$. The direct sum structure $\Gamma_{\sigma} = \Gamma_A \oplus \Gamma_{\bar A}$ allows us to define a ``symplectic-orthogonal projector'' $\Pi_A$ onto $\Gamma_A$. 
An expression for $\Pi_A$ can be obtained as follows. Lest $\Gamma^{\mathbb{C}}_A$ be the complexified version of $\Gamma_A$, and let $\gamma_A \in \Gamma^{\mathbb{C}}_A$ be a vector normalized with respect to the Klein-Gordon inner product, i.e., $\langle \gamma_A, \gamma_A \rangle = 1$, where $\langle \gamma, \gamma' \rangle := -i \bm{\omega}(\gamma^*, \gamma')$. Then the pair $(\gamma_A, \gamma_A^*)$ forms a basis for $\Gamma^{\mathbb{C}}_A$, and the projector $\Pi_A$ can be written as
\[
\Pi_A = \gamma_A \langle \gamma_A, \cdot \rangle - \gamma_A^* \langle \gamma_A^*, \cdot \rangle.
\] $\Pi_A$ is independent of the basis  $(\gamma_A,\gamma_A^*)$ chosen in this construction. See \cite{partnerformula} for further details. \label{KG}} 
Then, $\Pi_A^\perp = 1 - \Pi_A$ projects onto the symplectic orthogonal complement of $\Gamma_A$.

The partner mode $A_p$ is defined by the symplectic subspace \cite{partnerformula}
\begin{equation}\label{partner}
\Gamma_{A_p} = \Pi_A^\perp[J \Gamma_A],
\end{equation}
where $J$ is the complex structure associated with the pure Gaussian state in which the field theory is prepared.

Given basis vectors $(\gamma^{(1)}_A,\gamma^{(2)}_A)$ in $\Gamma_A$, 
 \be \frac{\Pi_A^\perp[J\gamma^{(1)}_A]}{\sqrt{|\det J_A|-1}},\quad \frac{\Pi_A^\perp[J\gamma^{(2)}_A]}{\sqrt{|\det J_A|-1}}\ee
form a (normalized) basis in $\Gamma_{A_p}$. Here, $J_A = \Pi_A J\Pi_A$ denotes the reduced complex structure in subsystem $A$. 

For the Bunch--Davies and Minkowski vacuum, we can use the expressions for the complex structures $J_{BD}$ and $J_{M}$ given in \eqref{JdS} and \eqref{JM}, respectively, to compute the partner of any given mode. (An example is shown below.)

\subsection{Asymptotic behavior of the partner of a local mode in de Bunch--Davies  vacuum\label{subsubsec:partners_general}}

Given any basis vector $\gamma_A \in \Gamma_A$, the action of $J$ is obtained by computing the integral
\begin{equation}
J_{BD}\gamma_A(\vec{x}) = \int d^3x'\, J_{BD}(\vec{x},\vec{x}')\, \gamma_A(\vec{x}')\, .
\end{equation}
A simple inspection of this expression reveals two important facts. First, because $J_{BD}(\vec{x}, \vec{x}') \neq 0$ for all $\vec{x}$ and $\vec{x}'$, the function $(J_{BD}\gamma_A)(\vec{x})$ is generically non-zero for all $\vec{x}$---even when $\gamma_A(\vec{x})$ has compact support. This implies, in particular, that the support of the partner mode overlaps with the support of the original mode $\gamma_A$. This poses no problem, as the two modes $A$ and $A_p$ are distinct and independent.\footnote{Two subsystems with classical state spaces $\Gamma_A$ and $\Gamma_B$ are said to be independent when they are symplectically orthogonal, that is,
\[
\boldsymbol{w}(\gamma_A,\gamma_B)=0,
\]
for all $\gamma_A \in \Gamma_A$ and $\gamma_B \in \Gamma_B$.
} In physical terms, this indicates that the mode $A$ is correlated and entangled with field degrees of freedom in the same region of space as $A$ itself.

Second, even if the support of $A$ is compact, the partner mode is generically not compactly supported---this is in tune with  Reeh-Schlieder's theorem \cite{reehschlieder}. Moreover, for regions far away from the support of any function in $\Gamma_A$, the behavior of the partner is entirely dictated by $J_{BD}(\vec{x}, \vec{x}')$. In other words, the asymptotic form of the partner mode of any mode $A$ is universal and determined solely by the complex structure $J_{BD}(\vec{x}, \vec{x}')$, or equivalently, by the long-distance correlations present in the quantum state.

The following proposition formalizes this statement.

\begin{prop}\label{prop:asymptotics_partner}
Let $A$ be a compactly supported single-mode subsystem, and let $A_p$ be its partner mode. When the field is prepared in the Bunch--Davies vacuum, the partner mode $A_p$ is not compactly supported and exhibits asymptotic decay at large distances $r$ (measured from any point in the support of $A$)  at least as $r^{-2\mu^2}$ as $r \to \infty$. 

More precisely, let $\gamma_{A_p}$ be any vector in $\Gamma_{A_p}^{\mathbb{C}}$---the complex version of  $\Gamma_{A_p}$---with unit Klein--Gordon norm (see footnote~\ref{KG}), so that the pair $(\gamma_{A_p}, \gamma_{A_p}^*)$ forms a symplectically-orthonormal basis of $\Gamma^{\mathbb{C}}_{A_p}$, from which we can obtain a basis in $\Gamma_{A_p}$. If we write $\gamma_{A_p}(\vec{x}) = (g_p(\vec{x}), f_p(\vec{x}))$, then $g_p(\vec{x})$ falls off  {at least} as $r^{-2\mu^2}$ and $f_p(\vec{x})$ falls off as $r^{-1-\mu^2}$,   as $r \to \infty$.
    
For comparison, in the Minkowski vacuum, $g_p(\vec{x}) \sim r^{-2}$ and $f_p(\vec{x}) \sim r^{-4}$, for $m = 0$, and  both components decay as $e^{-m r}$ for $m \neq 0$.
\end{prop}

\begin{proof}
The proof rests on standard asymptotic techniques, as follows.

Let $\gamma_{A}(\vec{x})$ be an element in $\Gamma^{\mathbb{C}}_A$ with unit Klein-Gordon norm, so that the pair $(\gamma_A, \gamma_A^*)$ forms an orthonormal basis in $\Gamma_A^{\mathbb{C}}$. { (In this proof we use a complex basis for subsystem $A$ merely for convenience---it makes the proof more compact, since the subsystem is characterized by a single complex vector $\gamma_A$ rather than two real ones. It is obvious that the arguments apply equally well to a real basis—which can be obtained by separating the complex basis vector $\gamma_A$ into its real and imaginary parts.)}

Each $\gamma_A(\vec{x})\in \Gamma^{\mathbb{C}}_{\sigma_{\rm BD}}$ consists of a pair of complex-valued functions, which we assume to be compactly supported. Let us begin by considering spherically symmetric functions, i.e., functions that depend only on the radial distance \( r := |\vec{x}| \) from the center of the support of $A$.

Let $\gamma_A\in\Gamma_A^{\mathbb{C}}$ be of the form $\gamma_A = (g_A(\vec x), f_A(\vec x))$. Using Eq.~\eqref{JdS}, we find
\begin{widetext}
    \bea\label{eq:Jgamma_dS_approx_radial}
 (J_{BD} \gamma_A)(r) &=& \frac{1}{2\pi^2 r } \int_0^{\infty} dk\, k\, \sin\left( k\,r\right) \begin{pmatrix} K\left(\frac{{ R}H}{k}\right)^{1- \mu ^2} \tilde{g}_A (k) - \left(\frac{1}{k} + K^2 (RH)^{2-2\mu^2}k^{-3+ 2\mu ^2} \right) \tilde f_A(k) \\ k \,\tilde{g}_A(k) - K \left(\frac{{ R}H}{k}\right)^{1-\mu^2} \tilde{f}_{A} (k)   \end{pmatrix}\nonumber \\  & +& { \mathcal{O}(\mu^2)} \,,
\eea
\end{widetext}
where $K = \frac{2^{1-\mu^2} \sqrt{\pi} R^{-1+\mu^2}}{\cos(\pi \mu^2)\Gamma\left(-\frac{1}{2} + \mu^2\right)}$.

To study the behavior of the first two terms in each component as $r/R \gg 1$, we express $\sin(x)$ using complex exponentials and apply the identity
\[
e^{\pm i k r} = \mp i r^{-1} \partial_k e^{\pm i kr}
\]
followed by repeated integration by parts. This procedure ultimately yields integrals of the form
\begin{equation}\label{eq:oscillatory_integrals}
I(\omega)= \int_0^b x^{-\alpha} F(x) e^{i \omega G(x)}\, dx
\end{equation}
with $F(x)$ and $G(x)$ functions with $G'(x) \neq 0$, and $\alpha \in (0,1)$. The asymptotic behavior of  integrals of this form is well-known (see, e.g.,~\cite{JCM-39-2}), and given by $I(\omega) \sim \mathcal{O}(\omega^{-(1 - \alpha)})$.

It remains to analyze the last term in the first component of Eq.~\eqref{eq:Jgamma_dS_approx_radial}, which is of the form
\begin{equation}\label{eq:problematic_dS}
\int_0^{\infty} dk\, k^{-1+2\mu^2} \tilde{F}(k)\, \mathrm{sinc}(k\,r)\,.
\end{equation}

To proceed, we perform the change of variables ${  R}k = t^{1/(2\mu^2)}$, yielding \[\frac{1}{2\mu^2} \int_0^{\infty} dt\, \mathrm{sinc}(t^{1/(2\mu^2)} \frac{r}{{  R}}) \tilde{F}(t^{1/(2\mu^2)})\,.\]

The sinc function is highly oscillatory for large arguments, making the method of the non-stationary phase appropriate to determine the asymptotic behavior of the integral in the limit \( r \gg 1 \). In this regime, the sinc function can be approximated by a step function, leading to  
\[
\begin{split}
    \frac{1}{2\mu^2}&\int_0^{\infty} \! dt\, \mathrm{sinc}\!\left(t^{1/(2\mu^2)} r/{ R}\right) \tilde{F}\!\left(t^{1/(2\mu^2)}\right)\\
    &\sim \tilde{F}(0) \int_0^{(\pi{R}/2r)^{2\mu^2}} \! dt + \mathcal{O}(r^{-1-2\mu^2}) \\
    &= \frac{\tilde{F}(0)}{2\mu^2} \left(\frac{\pi{ R}}{2r}\right)^{2\mu^2} + \mathcal{O}(r^{-1-2\mu^2})\,,
\end{split}
\]
where we have used that the upper limit of the integral, \( (\pi{ R} / 2r)^{2\mu^2} \), tends to zero for \( r/{  R} \gg 1 \), so $\tilde F$ can be approximated by the constant $\tilde F(0)$ in the interval of integration.

Putting all these results together, we find that the upper and lower components of $J\gamma_A$ behave asymptotically as:
\bea\label{eq:Jgamma_asymptotics}
(J_{BD}\gamma_A)_1 \!&=&\! - \frac{(HR)^{2-2\mu^2}}{2\pi^2} \frac{\tilde{\gamma}^{(2)}(0)}{2\mu^2}\! \left(\frac{\pi{  R}}{2r}\right)^{2\mu^2} \nonumber \\
&& + \mathcal{O}(r^{-1 - \mu^2}),\\ 
(J_{BD}\gamma_A)_2 \!&=&\! \mathcal{O}(r^{-2 - \mu^2})\,.
\eea
To obtain a basis for the partner mode, one needs to act with the symplectic projector $\Pi_A^{\perp}$ on $(J_{BD}\gamma_A)$. The action of $\Pi_A^{\perp}$ reduces to the identity outside the region of support of $A$. Since $A$ is compactly supported, the asymptotic behavior given in Eq.~\eqref{eq:Jgamma_asymptotics} remains unaffected by the projection.

A similar analysis applies to phase space elements $\gamma_A$ that are not spherically symmetric. In that case, one first expands $\gamma_A$ in spherical harmonics. Following~\cite{fokas2012fouriertransformschebyshevlegendre}, the angular integrals introduce spherical Bessel functions. Using their standard properties (e.g., oscillatory behavior, asymptotics, and location of zeros), one can verify that the leading-order decay in $r$ remains the same as in Eq.~\eqref{eq:Jgamma_asymptotics}.

Finally, this proof remains valid even if the mode $\Gamma_A$ is not compactly supported but made of functions  that decay faster than any polynomial. In particular, the Proposition holds if the functions in $\Gamma_A$ are in Schwartz space. 
\end{proof}

{ As in previous sections, our proof is restricted to modes outside the special family defined in Eq.~\eqref{zeroaverage}. The analysis of this special family of modes is presented in Appendix~\ref{app:proofs_special_functions}.}

\begin{exmp}\label{ex:singlemodeball_cont2}{\bf Single-mode subsystem supported in a ball (cont.)}

This is a continuation of Examples~\ref{ex:singlemodeball} and \ref{ex:singlemodeball_cont}.
The single mode subsystem $A$ is defined by $\Gamma_A={\rm span}[{\bm \gamma}_A^{(1)},{\bm \gamma}_A^{(2)}]$, with
\begin{equation}\label{eq:ex3vNS}
    {\bm \gamma}^{(1)} = \left(\begin{matrix} 0\\ f^{(\delta)} (\vec{x}) \end{matrix}\right)\,, \quad {\bm \gamma}^{(2)} = \left(\begin{matrix} - g^{(\delta)} (\vec{x})\\ 0 \end{matrix}\right)\,,
\end{equation}
where the real functions $g^{(\delta)} (\vec{x})$ and $f^{(\delta)} (\vec{x})$ have compact support within a ball of radius $R$ (their functional form is given in \eqref{eq:fdelta_family}. 

Alternatively, we combine ${\bm \gamma}^{(1)}$ and ${\bm \gamma}^{(2)}$ in a complex mode ${\bm \gamma_A} = \frac{1}{\sqrt{2}}\, \left({\bm \gamma}_A^{(1)}+i\, {\bm \gamma}_A^{(2)}\right)$. 

The vector 
\be {\bm \gamma}_{Ap}\equiv\frac{\Pi_A^\perp[J_{BD}\gamma_A]}{\sqrt{|\det J_A|-1}},\ee
together with its complex conjugate, form a normalized basis of 
$\Gamma^{\mathbb{C}}_{Ap}$. (Their real and imaginary part form a basis in $\Gamma_{Ap}$.)

Let us write ${\bm \gamma}_{Ap}$ as
\be {\bm \gamma}_{Ap}=\left(\begin{matrix} g_p(\vec{x})\\ f_p(\vec{x}) \end{matrix}\right)\, .\ee We compute complex functions $g_p$ and $f_p$  analytically  and plot them in Fig.~\ref{fig:partnerexampleball}. For comparison, the figure shows the shape of the partner mode both in the Bunch--Davies and in the Minkowski vacuum (the latter corresponding to a massless field). The behavior of the partner mode at large $r$ differs significantly in the two cases: in the Bunch--Davies vacuum, some components of the partner basis vectors are nearly scale-invariant, whereas in the Minkowski vacuum they decay at least as $r^{-2}$.

\begin{figure*}
    \centering
      \includegraphics[width=\textwidth]{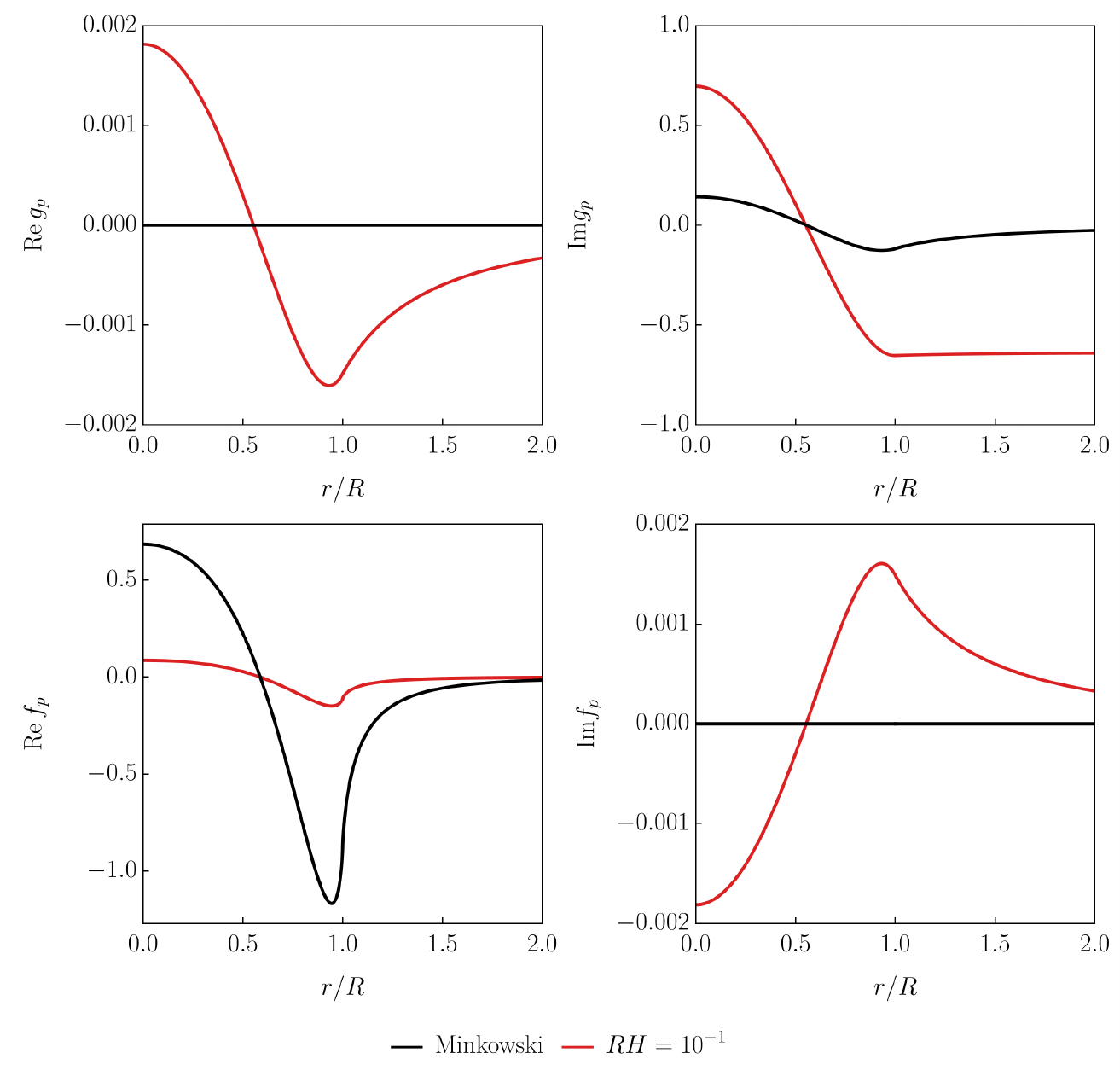}
    \caption{{  Real and imaginary parts of the two components of the complex partner mode, $\bm \gamma_{Ap} = (g_p, f_p)$, introduced in Example~\ref{ex:singlemodeball_cont2}, as a function of the radial coordinate, $r/R$. Results are presented for a massless scalar field in Minkowski spacetime (black line) and for a field in de Sitter with $\mu^2 = 10^{-4}$ supported in a spherical region of radius $RH = 10^{-1}$. In de Sitter space, the imaginary part of $g_p$ becomes nearly scale-invariant, whereas in Minkowski space it decays as $r^{-2}$. In Minkowski, the real part of $g_p$ and the imaginary part of $f_p$ are zero, but they are non-zero in de Sitter due to the non-vanishing field–momentum correlations in the Bunch-Davies vacuum}}   \label{fig:partnerexampleball}
\end{figure*}

\end{exmp}

\subsection{Where is the entanglement?\label{subsec:where_is_entanglement}}

Proposition~\ref{prop:asymptotics_partner} reveals that the distribution of correlations and entanglement in the Bunch--Davies vacuum differs substantially from that in Minkowski spacetime. In particular, correlations and entanglement extend over much larger distances in the Bunch--Davies vacuum, owing to the nearly scale-invariant support of the partner mode.

This carries an important message: a nonzero Hubble rate $H$, even if small, drastically alters the spatial distribution of entanglement in the vacuum.

The consequences of this are significant. Although a large value of $H$ makes any local mode strongly entangled with its partner mode, the latter is not accessible to local observers. The resulting long-range correlations manifest instead as large entropies for compactly supported modes. Physically, this entropy is perceived as local ``thermal noise,'' which in turn precludes local, compactly supported modes to get entangled with one another.

\section{Discussion\label{sec:discussion}}

The generation of entanglement during inflation, and more generally, entanglement in Friedmann–Lemaître–Robertson–Walker (FLRW) spacetimes, has been mainly studied by focusing on Fourier modes~{\cite{Grishchuk:1990bj,Albrecht:1992kf,Lesgourgues:1996jc,Kiefer:1998pb,Kiefer:1998qe,Kiefer:2008ku,Polarski:1995jg,Martin:2015qta}}, i.e., single subsystems with well-defined wavenumber $\vec{k}$—see~\cite{Martin:2021xml,Martin:2021qkg,K:2023oon,PhysRevD.78.044023,PhysRevD.80.124031} for notable exceptions. It has been further argued that inflation squeezes pairs $(\vec{k}, -\vec{k})$, and that this squeezing is responsible for the apparent classicality of the primordial perturbations. Several limitations of this strategy have been pointed out {\cite{Hsiang:2021kgh,Agullo:2022ttg,ack}}. In particular, in FLRW spacetime, due to the time-dependence of the spacetime metric, there is a large ambiguity in what a “single subsystem with well-defined wavenumber $\vec{k}$” means, which goes hand-in-hand with the ambiguity in the definition of particles. Namely, if $\hat{a}^{\dagger}_{\vec{k}}$ and $\hat{a}_{\vec{k}}$ are the creation and annihilation variables characterizing a $\vec{k}$ mode, then so are $\hat{a}'^{\dagger}_{\vec{k}} = \alpha \hat{a}^{\dagger}_{\vec{k}} + \beta \hat{a}_{-\vec{k}}$ and its adjoint, provided $|\alpha|^2 - |\beta|^2 = 1$. This ambiguity is large enough to arbitrarily change any conclusion regarding squeezing in the $(\vec{k}, -\vec{k})$ pairs. Moreover, the impact of any squeezing on the quantum nature—or lack thereof—of cosmological perturbations has  been  criticized in~\cite{Hsiang:2021kgh,Agullo:2022ttg}.)

Even if one devises a clever argument to select a preferred set of Fourier modes, it is important to keep in mind that such modes are supported in the entire universe.
This inherently global character implies that many aspects of these modes are not accessible to local observers. Locally, it is possible to estimate certain features of Fourier modes, such as their amplitude,  via a local Fourier transform. However, accessing the entanglement stored in pairs of such modes would require access to their full spatial support---an impossible endeavor.

This motivates the strategy followed in this article which, building upon previous works~\cite{Martin:2021qkg,K:2023oon,Martin:2021xml,ubiquitous}, focuses on compactly supported modes. In quantum field theory, \emph{all} compactly supported modes have \emph{mixed} reduced states, regardless of the field’s state \cite{Hollands:2017dov}. This is due to the unavoidable presence of correlations between any region and its complement in any physical (i.e., Hadamard) state in quantum field theory. Importantly, the mixedness of subsystems effectively acts as a decohering agent that reduces the entanglement with other  subsystems. Hence, the study of local modes is considerably richer.

In this article, we have followed a local approach to characterize the distribution of entanglement in the cosmological patch of de Sitter space and how it varies with the value of the de Sitter curvature. Our approach is formulated in slightly more geometric and invariant terms than in previous works, as most of our results are derived from the structure and properties of an inner product on the classical phase space defined by the Bunch–Davies vacuum. Indeed, one of our goals is to promote these tools, as we believe they strike a useful balance between conceptual and mathematical clarity, and the ease with which they can be applied to compute quantities of direct physical interest.

The main findings of this work can be summarized as follows. First, we highlight the fact that introducing a cosmological constant, even a tiny one, drastically changes the structure of correlations and entanglement in the vacuum state. Mathematically, this occurs because the covariance metric (or the associated complex structure) defined from the Bunch–Davies vacuum acquires a term proportional to a Sobolev product of index $-\frac{3}{2} + \mu^2$, with ${\mu}^2 \equiv \frac{3}{2} - \sqrt{\frac{9}{4} - \frac{m^2}{H^2}} \ll 1$, where $m$ denotes the mass of the field. This term introduces almost scale-invariant correlations that dominate at large separations—along with an infrared divergence in the massless limit.

It is well-known that the presence of a cosmological constant, even if tiny, produces qualitative changes in several aspects of physics, such as a change in the asymptotic structure of spacetime~\cite{Hawking:1973uf}—resulting in drastic modifications to the definition of gravitational radiation~\cite{abk1,Ashtekar:2015ooa}—making it impossible to smoothly connect with the $H = 0$ case. Analogously, $H \neq 0$ qualitatively alters the structure of correlations and entanglement in the quantum vacuum.

To study these changes, we have followed an approach in which subsystems are defined in an invariant manner, and quantifying  correlations and entanglement in a similarly invariant way. 

One of the main messages of this analysis is that increasing the de Sitter curvature increases correlations between local modes but \emph{decreases} their entanglement. This may seem counterintuitive at first, highlighting the subtle relationship between correlations and entanglement, which are commonly thought to always go hand-in-hand.

In the Bunch--Davies vacuum, correlations and entanglement with any locally supported modes are distributed over arbitrarily large distances. This contrasts with the Minkowski vacuum, where correlations fall off polynomially for a massless field and exponentially for a massive one.  Furthermore, increasing the de Sitter curvature enhances the total amount of correlations and entanglement any local mode has with its complement—the rest of the field modes. But these correlations and entanglements are spread across the entire universe in an almost scale-invariant manner. We have made this  aspects precise using the concept of \emph{partner mode}, discussed in Section~\ref{sec:partners}. 

Stronger entanglement, on the other hand, makes the von Neumann entropy of any locally supported subsystem to grow monotonically with the Hubble rate. This large entropy comes at the expense of local entanglement, which therefore decreases when the de Sitter curvature increases.

Heuristically, the reduction of local entanglement can be attributed to a kind of entanglement monogamy. Local modes are more entangled with their partners, but the partner modes are not accessible to cosmological observers. The strong entanglement with the partner mode prevents the local mode from becoming substantially entangled with other local modes. The link between entanglement monogamy and the local thermal properties of the Bunch Davies vacuum has been pointed out in \cite{Nambu_2023} based on the analysis of a rather special family of non-compactly supported field modes (see the Remark on page \pageref{remark_IR} for the definition of the family to which these modes belong to).

The results of this article have direct implications for cosmology and the long-standing debate on how much entanglement is generated by inflation. Our view is that a satisfactory answer to this question must focus on local, accessible modes—in contrast to Fourier modes—and on invariant aspects of the primordial perturbations that do not depend on  choices or conventions. When doing so, we find that the cosmological perturbations that later become accessible to observers are, at the end of inflation, \emph{less} entangled than they would have been had the field been in the Minkowski vacuum. Furthermore, any entanglement present decreases monotonically with the value of the Hubble rate during inflation. 

 Our analysis relies on the Gaussian nature of the Bunch–Davies vacuum, which is widely regarded as an excellent approximation to the quantum state of cosmological perturbations. It is therefore natural to ask how our conclusions might be affected by the presence of non-Gaussian features. In general, non-Gaussian states can exhibit entanglement properties that differ qualitatively from those of Gaussian states, and in principle could lead to an enhancement, or even more degradement, of entanglement between local modes. However, observational constraints on primordial non-Gaussianity from the cosmic microwave background and large-scale structure indicate that any such effects must be very small~\cite{Planck:2019kim}. We therefore expect that non-Gaussian corrections compatible with current observational bounds would not qualitatively modify the results presented here. A detailed quantitative analysis of genuinely non-Gaussian states, however, lies beyond the scope of the present work.

We do not have access to the freshly generated perturbations at the end of inflation—they must evolve through the complexities of the cosmos, a process that will likely decohere them and further reduce their entanglement. This complicated process has been modeled and discussed in many references{~\cite{Halliwell:1989vw,Calzetta_1995,Polarski:1995jg,Martineau:2006ki,Burgess_2008,Kiefer:2008ku,Nelson_2016,Martin_2018,Hsiang:2021kgh,Colas:2022kfu,Bhattacharyya:2024duw}}. The statements in this article are complementary to these analyses, and point that the correlations produced by inflation do \emph{not} come with a corresponding enhancement of entanglement, if one only has access to a finite portion of the universe. 

More generally, this article provides an example of how curvature and symmetries shape the entanglement content of field theory, and how this content can be characterized using tools from symplectic and Kähler geometry, along with the concept of partner modes.

\acknowledgments
The content of this paper has benefited from discussions with A. Ashtekar, E.Bianchi, A. Delhom, J. Martin, E. Martin-Martinez, S. Nadal,  W. van Suijlekom, V. Vennin and K. Yamaguchi. P.R.-M. thanks the Perimeter Institute for Theoretical Physics for their support as well as LSU for hosting her for a  visit during which part of this work was done.  P.R.-M. also thanks the Julian Schwinger Foundation for financial support at the 2025 Peyresq Spacetime Meeting, where useful discussions influenced this work. P.R.-M. is supported by the 1851 Research Fellowship.  This research has been partially supported by the Blaumann Foundation. I.A. is supported by the NSF grants PHY-2409402 and PHY-2110273, by the RCS program of Louisiana
Boards of Regents through the grant LEQSF(2023-25)-RD-A-
04,  by the Hearne Institute for Theoretical Physics and by Perimeter Institute of Theoretical Physics through the Visitor fellow program. Research at Perimeter Institute is supported in part by the Government of Canada through the Department of Innovation, Science and Industry Canada
and by the Province of Ontario through the Ministry of
Colleges and Universities.
This article was finished while B.B. was visiting the Okinawa Institute of Science and Technology (OIST) through the Theoretical Sciences Visiting Program (TSVP).

\appendix
\section{Useful definitions and propositions}
\label{app:sobolev}
This appendix provides a summary of useful definitions and results from functional analysis, matrix analysis, and Gaussian quantum information that are used throughout the article.

 \begin{defn} (\textit{Definition 1.18 in Ref.~\cite{bahouri_fourier_2011}})
    The Schwartz space $\mathcal{S}(\mathbb{R}^n)$ is the vector space of all complex-valued, infinitely differentiable functions that, together with all of their derivatives, decay faster than any inverse polynomial as $|x|\to\infty$. Formally, 
     \begin{equation}\begin{split}
         \mathcal{S}(\mathbb{R}^n) = \{&f \in C^{\infty}(\mathbb{R}^n)\,|\, \forall \alpha, \beta \in \mathbb{N}^{n}, \\&\,\mathrm{sup}_{x\in \mathbb{R}^n}\, |x^{\alpha} (D^{\beta} f) (\vec{x})| < \infty\}\,, 
     \end{split}
     \end{equation}
      where $\mathrm{sup}$ denotes the supremum, $\alpha,\beta\in \mathbb{N}^n$ are multi-indices,   $x^\alpha = x_1^{\alpha_1}\cdots x_n^{\alpha_n}$, and  $D^{\beta} = \partial_1^{\beta_1} \cdots\partial_n^{\beta_n}$. 
\end{defn}

The Schwartz space is commonly used to define the classical phase-space of fields in Minkowski spacetime because of the following properties:
\begin{enumerate}
    \item If  $f \in \mathcal{S}(\mathrm{\mathbb{R}^n})$, then its Fourier transform  $\tilde{f}$ exists and belongs to the Schwartz space: $\tilde{f}\in \mathcal{S}(\mathbb{R}^n)$. Thus, the Fourier transform is a bijection on $\mathcal{S}(\mathbb{R}^n)$. 
    \item The space of smooth functions of compact support $C^\infty_0(\mathbb{R}^n)$ is a subspace of the Schwartz space, $\mathcal{S}(\mathrm{\mathbb{R}^n}).$
\end{enumerate}

\begin{defn} (\textit{Definition 1.20 in \cite{bahouri_fourier_2011}}) A \textbf{tempered distribution} on $\mathbb{R}^n$ is a continuous linear functional on $\mathcal{S}(\mathbb{R}^n)$.\footnote{Continuity is understood with respect to the standard locally convex topology on $\mathcal{S}(\mathbb{R}^n)$, induced by the family of Schwartz seminorms (see, e.g. Ref.~\cite{treves_topological_1967}).} The space of tempered distributions is denoted by $\mathcal{S}' (\mathbb{R}^n)$.
\end{defn}

\begin{defn} The \textbf{space of locally integrable functions}, denoted as $ L^1_{\mathrm{loc}}(\mathbb{R}^n)$, consists of functions that are integrable over every compact subset of $\mathbb{R}^n$. That is, 
\begin{equation*}
    L^1_{\mathrm{loc}} (\mathbb{R}^n)\! =\! \{ f:\mathbb{R}^n\! \to\! \mathbb{C} |f\in L^1(K)~\forall~\text{compact}~K\}.
\end{equation*}

\end{defn}

\begin{defn}
     (\textit{Definition 1.31 in~\cite{bahouri_fourier_2011}}\label{def:hom_sobolev_spaces}) 
Let $s\in \mathbb{R}$. The \textbf{homogeneous Sobolev space }$\dot{H}^s(\mathbb{R}^D)$ is the space of tempered distributions $u$ over $\mathbb{R}^D$, the Fourier transform of which belongs to $L^1_{\mathrm{loc}}(\mathbb{R}^D)$ and satisfies
\begin{equation}\label{eq:sobolev_norm}
    || u||^2_{\dot{H}^s} := \int_{\mathbb{R}^D} |k|^{2s} |\tilde{u}(k)|^2 \, dk < \infty,
\end{equation}
where, as before, $\tilde{u}$ denotes the Fourier transform of $u$. In other words,
\begin{equation}
    \dot{H}^s(\mathbb{R}^D) := \left\{ u \in \mathcal{S}'\!  :\! \tilde{u} \in L^{1}_{\mathrm{loc}}, \text{ and } ||u||^2_{\dot{H}^{s}} \!<\! \infty  \right\}.
\end{equation}
\end{defn}

{Note that the Schwartz space $\mathcal{S}(\mathbb{R}^n)$ is canonically embedded in the space of tempered distributions $\mathcal{S}'(\mathbb{R}^n)$ via $f \mapsto T_f$, where $T_f(\phi) = \int dx f(x) \phi(x)$. This embedding explains why the Cauchy completion of the classical phase space $\Gamma$ (chosen as  $\mathcal{S}(\mathbb{R}^n)$) with respect to the seminorm $\sigma$ yields a homogeneous Sobolev space $\dot H_s(\mathbb{R}^n)$, whose elements are tempered distributions.  }

\begin{prop} (\textit{Prop. 1.1 in~\cite{bahouri_fourier_2011}}) \textbf{H\"older's inequality:} Let $(X,\mu)$ be a measure space and $(p,q,r)$ in $[1,\infty]^3$ be such that 
$$ \frac{1}{p} + \frac{1}{q} = \frac{1}{r}\, .$$

If $(f,g)$ belongs to $L^p(X,\mu) \times L^q(X,\mu)$, then $fg$ belongs to $L^r(X,\mu)$ and 
$$||fg||_{L^r} \leq ||f||_{L^p} ||g||_{L^q}\,.$$
\end{prop}

\begin{prop} (\textit{Proposition 1.32 in~\cite{bahouri_fourier_2011}}\label{prop:inclusions})
Let $s_0 \leq s \leq s_1$. Then, $\dot{H}^{s_0} \cap \dot{H}^{s_1}$ is included in $\dot{H}^s$,  and we have $$||u||_{\dot{H}^s} \leq ||u ||^{1-\theta}_{\dot{H}^{s_0}}||u ||^{\theta}_{\dot{H}^{s_1}} \text{ with }  s= (1-\theta)s_0 + \theta s_1\,. $$ 
Using the Fourier–Plancherel formula, one finds $L^2=\dot{H}^0$.
\end{prop} 

\begin{prop} (\textit{Proposition 1.34 in~\cite{bahouri_fourier_2011}}\label{prop:hilbert_space})
$\dot{H}^s(\mathbb{R}^D)$ is a Hilbert space if and only if $s< D/2$, with inner product

\begin{equation} \label{eq:innerprod}
    (f,g)_{s} = \int_{\mathbb{R}^D} \frac{d^Dk}{(2\pi)^D} \, |\vec k|^{2s} \tilde{f}(\vec k) \, \tilde{g}^*(\vec k)\, ,
\end{equation}
where the asterisk denotes complex conjugation.
\end{prop}

\begin{prop}(\textit{Chapter 7 in~\cite{serafini2017quantum}}) \label{prop:separability} \textbf{Gaussian separability lemma}:  Let $\sigma$ be the covariance matrix for a two-mode state of the form 
\be 
{\bm \sigma}_{AB} = \left( \begin{matrix}
    {\bm\sigma}_A & \bm C \\ 
    \bm C^T & \bm\sigma_B
 \end{matrix}\right) \,,
\ee 
where $\sigma_A$ and $\sigma_B$ are the local covariance matrices of the single-modes subsystems $A$ and $B$, and $C$ represents their correlations. Two-mode Gaussian states with $\mathrm{Det}\, C \geq 0$ are separable.   
\end{prop}

\begin{prop} (\textit{Thm. 7.8.5 in Ref.~\cite{Horn_Johnson_2012}}) \textbf{Fischer's inequality:}  Suppose that the partitioned Hermitian matrix 
\begin{equation}
    {\bm D} = \left(\begin{matrix}
        {\bm A} & {\bm C} \\ 
        {\bm C}^* & {\bm B}
    \end{matrix}\right)
\end{equation}
is a positive-definite $(p+q)\times (p+q)$ matrix, ${\bm A}$ is a $p\times p$ matrix, and ${\bm B}$ a $q\times q$ matrix. Then 
\begin{equation}
    \det {\bm D} \leq (\det {\bm A})(\det {\bm B}).
\end{equation}
    
\end{prop}

\section{Proof of Prop.~\ref{prop:BDcorrelations}\label{app:proofprop1}}
In this appendix, we show that the functionals $\mathcal{N}_{\Phi\Phi}$, $\mathcal{N}_{\Pi\Pi}$, and $\mathcal{N}_{\Phi\Pi}$ in Eqs.~(\ref{phiphi})-(\ref{phipi}) are bounded above by quantities of order $\mu^2\ll 1$. This ensures that they provide subleading corrections {to the terms in Eqs.~\eqref{phiphi}-\eqref{phipi}} in the small-$\mu^2$ regime. We also discuss their asymptotic behavior in the limits $RH \gg 1$ and $RH \ll 1$. 

The   functionals $\mathcal{N}_{\Phi\Phi}$, $\mathcal{N}_{\Pi\Pi}$, and $\mathcal{N}_{\Phi\Pi}$ in Eqs.~(\ref{phiphi})-(\ref{phipi}) are defined given by
\begin{widetext}
\begin{align}\label{eq:Nphiphi}
\mathcal{N}_{\Phi\Phi} :=\; &(RH)^{-1} \int \frac{d^3k}{(2\pi)^3}  \left[ 
  \frac{\pi}{2} \left| H^{(1)}_{\frac{3}{2} - \mu^2} \left( \frac{k}{H} \right) \right|^2 
  - \frac{2^{2 - 2\mu^2} \pi \left( \frac{k}{H} \right)^{-3 + 2\mu^2}}{ 
      \cos^2(\pi \mu^2)\, \Gamma\left( \mu^2 - \tfrac{1}{2} \right)^2 } 
  - \left( \frac{k}{H} \right)^{-1} 
\right] \nonumber \\
&\times R \,\mathrm{Re}\left( \tilde{f}_I(\vec{k})\, \bar{\tilde{f}}_J(\vec{k}) \right)\,,
\end{align}

\begin{equation}\label{eq:Npipi}
        \mathcal{N}_{\Pi\Pi} := RH \int \frac{d^3k}{(2\pi)^3}\left[\frac{\pi}{2}   \left(\frac{k}{H}\right)^2 \left|H^{(1)}_{\frac{1}{2}-\mu ^2}\left(\frac{k}{ H}\right)\right|^2-\frac{k}{H}\right] R^{-1}\mathrm{Re}(\tilde{g}_I(\vec{k})\bar{\tilde{g}}_J(\vec{k}))\,,
    \end{equation}
    and 
\begin{align} \label{eq:Nphipi}
        \mathcal{N}_{\Phi\Pi}:=\; & -\int \frac{d^3k}{(2\pi)^3}\left[\frac{\pi}{2}\frac{k }{ H}\mathrm{Re}\left(H_{\frac{3}{2}-\mu ^2}^{(1)}\left(\frac{k}{ H}\right) H_{\frac{1}{2}-\mu ^2}^{(2)}\left(\frac{k}{ H}\right)\right)+\frac{\sqrt{\pi } 2^{1-\mu ^2} \left(\frac{k}{ H}\right)^{-1+\mu ^2}}{\cos \left(\pi  \mu ^2\right) \Gamma \left(\mu ^2-\frac{1}{2}\right)}\right] \nonumber \\  & \times  \mathrm{Re}(\tilde{f}_I(\vec{k})\bar{\tilde{g}}_J(\vec{k}))\,. 
    \end{align}
\end{widetext}
Here,  $H^{(1)}_{\alpha} (\vec x)$, $H^{(2)}_{\alpha} (\vec x)$ denote  Hankel functions of order $\alpha$ of the first and second kinds, respectively. 

\subsection{Upper bounds \label{subapp:upperbounds_Nall}}
\subsubsection{Upper bound on $\left|\mathcal{N}_{\Phi\Phi}\right|$}
Using the triangle inequality and the integral representation of the Hankel's function modulus square~\cite{freitas_boundsHankel},
\begin{equation*}\begin{split}
x^{2\alpha}|H_{\alpha}^{(1)}(x)|^2 &= \frac{8}{\pi^2} \lim_{\epsilon \to 0^+} \int^x_\epsilon dy\, y^{2\alpha -1}  \\ 
& \times\int_0^{\infty} dt K_0(2y \sinh t)\cosh t\, h_{\alpha}(t),
\end{split} \end{equation*}
 valid for real positive arguments and $1/2 <\nu < 3/2$, one can bound $|H^{(1)}_{3/2-\mu^2}|$  by bounding $h_{\alpha}(t)$.  Here,  $K_0$ is the modified Bessel function of the second kind  and $h_{\alpha}(t)$  is a positive and  monotonically decreasing function.\footnote{The explicit form of the function $h_{\alpha} (t)$  can be found in { Eq.~(4.1) in}~\cite{freitas_boundsHankel}.}    Carrying out this procedure yields 
\begin{equation} \begin{split}
     \left|\mathcal{N}_{\Phi\Phi}\right| \leq & \mu^2 \Bigg(||\mathfrak{F}_{IJ}||^2_{-1/2} + \frac{2}{\pi} H ||\mathfrak{F}_{IJ}||_{-1}^2\} \Bigg),
\end{split}
\end{equation} 
where $||\cdot||_{s}^2$ indicates the norm defined from the Sobolev inner product of order $s$ and $\mathfrak{F}_{IJ}$ is defined as the inverse Fourier transform of the function $\sqrt{|\tilde{f}_I||\tilde{f}_J|}$. 

\subsubsection{Upper bound on $\left|\mathcal{N}_{\Phi\Pi}\right|$}
Using the identity~\cite{freitas_boundsHankel} 
\begin{equation}
    J_{\alpha}J_{\alpha-1} + Y_{\alpha}Y_{\alpha-1}\! =\! \frac{1}{2x^{2\alpha}} \frac{d}{dx} \left( x^{2\alpha}|H^{(1)}_{\alpha}(x)|^2\right),
\end{equation}
valid for $x>0$ and $\alpha \in \mathbb{R}$, and following a similar strategy as the one outlined above, we obtain 
\begin{equation}\begin{split}
    |\mathcal{N}_{\Phi\Pi}|& \leq \mu^2 \bigg(H||\mathfrak{K}_{IJ}||^{2}_{-1/2}+ H^2 ||\mathfrak{K}_{IJ}||^{2}_{-1} \\
    &+ H^{1/2} ||\mathfrak{K}_{IJ}||^{2}_{-1/4}\bigg)\,, 
\end{split}
\end{equation}
where $\mathfrak{K}_{IJ}$ is defined as the inverse Fourier transform of the function $\sqrt{|\tilde{f}_I| |\tilde g_J|}$. 

\subsubsection{Upper bound on $|\mathcal{N}_{\Pi\Pi}|$}
Applying the triangle inequality and the bound
$$\left|\frac{q}{RH} - \frac{\pi}{2}\left(\frac{q}{RH} \right)^2 \left| H^{(1)}_{1/2-\mu^2}\left(\frac{q}{RH} \right) \right|^2   \right| < \frac{\mu^2}{2}\,, $$
we  find
\begin{equation}
     |\mathcal{N}_{\Pi\Pi}| \leq H \frac{\mu^2}{2} ||\mathfrak{G}_{IJ}||_0^2,  
\end{equation} where $\mathfrak{G}_{IJ}$ is defined as the inverse Fourier transform of $\sqrt{|\tilde g_I||\tilde g_J|}$.

Hence, the functionals $\mathcal{N}_{\Phi\Phi}$, $\mathcal{N}_{\Phi\Pi}$, and $\mathcal{N}_{\Pi\Pi}$ are $\mathcal{O}(\mu^2)$ for smooth, square-integrable smearing functions. While the bounds above are sufficient for the statements made in the main text,  they are not sharp in $RH$. In the next subsection, the behavior of these functionals is analyzed as a function of $RH$ in the regimes $RH \gg 1$ and $RH \ll 1$. 

\subsection{Asymptotic behavior for large and small  $RH$\label{subapp:asymptoticRH}}

For $RH\gg 1$, assume that the Fourier transforms $\tilde{f}$ and $\tilde{g}$ decay sufficiently fast---e.g. exponentially---so that they can be effectively approximated by functions compactly supported in a ball in Fourier space $\mathcal{B}(\vec 0,\Lambda)$, for some finite $\Lambda> 0$. When ${R}H \gg 1$,  the integrals in  Eqs.~\eqref{eq:Nphiphi}-\eqref{eq:Nphipi} only probe the Hankel functions at small arguments $k/H \ll1$ in $\mathcal{B}$.\footnote{The argument can be extended to functions $f$ and $g$ belonging to the appropriate Sobolev spaces. In this case, one examines the integrand separately in a ball $\mathcal{B}(\vec 0, \Lambda)$ and in its complement, finding that the leading dependence on $RH$ is governed by the contribution from $\mathcal{B} $ for sufficiently large $\Lambda$. } Hence, one can use the small-argument expansion of Bessel functions to find the asymptotic behavior 
\begin{align}
\mathcal{N}_{\Phi\Phi}&\underset{RH\gg1}{\sim}\mathcal{O}\big(\mu^2\log(RH)\big),\nonumber \\ 
\mathcal{N}_{\Pi\Pi}&\underset{RH\gg1}{\sim}\mathcal{O}\big(\mu^2\log(RH)\big),\nonumber\\
\mathcal{N}_{\Phi\Pi}&\underset{RH\gg1}{\sim}\mathcal{O}\big(\mu^2(RH)^{1-\mu^2}\big).
\label{app:RHgg1scaling}
\end{align}
The logarithmic growth in $\mathcal{N}_{\Phi\Phi}$ and $\mathcal{N}_{\Pi\Pi}$ arises from the integration of the small-argument Hankel asymptotics; $\mathcal{N}_{\Phi\Pi}$ carries an extra power of $k/H$ and hence scales as $(RH)^{1-\mu^2}$. The perturbative treatment in the main text is consistent provided $\mu^2 \log(RH) \ll 1$, so the $\mathcal{O}(\mu^2)$ suppression is effective even for large $RH$. 

For $RH\ll 1$, it is convenient to split the $k$-integrals into low- and high-momentum regions. For a fixed $\lambda { H} \gg 1$ such that $\lambda H \ll 1$, 
$$ \mathcal{N}_{\Xi\Xi'} =   \int_{|\vec{k}|\leq \lambda H} [\dots]+\int_{|\vec{k}|> \lambda H}\left[\dots\right]\,, $$
with $\Xi, \Xi' \in \{\Phi, \Pi\} $. A scaling analysis shows that the second term decays more slowly with $RH$ and thus dominates the small-$RH$ behavior of  $\mathcal{N}_{\Xi,\Xi' }$. Evaluating the second contribution, where the Hankel functions can be approximated by their large-argument expansion $|H^{(1)}_{\alpha}(x)| \sim \left(\frac{2}{\pi x}\right)^{1/2} + \mathcal{O}(x^{-3/2})$, we find 
\begin{flalign*}\mathcal{N}_{\Phi\Phi} &\sim \mathcal{O}(\mu^2 (RH)^2)\,, \\
\mathcal{N}_{\Pi\Pi} &\sim \mathcal{O}(\mu^2 (RH)^2)\,, \\
 \mathcal{N}_{\Phi\Pi} &\sim \mathcal{O} (\mu^2 RH)\,,  \end{flalign*}
{in the regime $RH \ll 1$ and $\mu^2 \ll 1$. }
In particular, this implies that all three functionals vanish in the Minkowski limit $RH \to 0$.

\section{Asymptotic behavior of Sobolev products\label{app:asymptotic_sobolev}}

This appendix contains the calculation of the asymptotic behavior of Sobolev inner products used in the main text to quantify how correlations between spatially separated modes fall off with their separation. In particular, we analyze the asymptotic behavior of Sobolev inner products $(h_A,h_B)_s$ for the values of $s$ relevant for this article. The functions $h_A$ and $h_B$ are assumed to be smooth---this will be relaxed later---and compactly supported in disjoint regions $A,B \subset \mathbb{R}^n$. Additionally, we assume that a meaningful notion of distance between $A$ and $B$ can be defined---namely, that the supports of the modes in $A$ and $B$ can be separated continuously without deforming or intersecting them.  The goal is to show that, under these conditions, the decay of the Sobolev products is determined by the order $s$ and the spatial dimension $n$, and a mild dependence on the specific shape of the functions $h_A$ and $h_B$.

We are interested in the asymptotic regime when the distance between the supports of the modes $A$ and $B$ is much larger than their sizes. Then, a notion of distance between $A$ and $B$ can be defined as
\begin{equation}
|\Delta \vec x| = |\vec x_B - \vec x_A|,
\end{equation} 
for any reference points $\vec x_A \in A$ and $\vec x_B \in B$.  In this setting,  the translation property of the Fourier transform allows us to express the Sobolev inner product as
\begin{equation}\label{c2}
(h_A,h_B)_s
 = \int \frac{d^n k}{(2\pi)^n}\,
   \tilde h_A(\vec k)\, \tilde h_B^*(\vec k)\,
   |\vec k|^{2s}\, e^{-i\vec k \cdot \Delta \vec x},
\end{equation}
where $\tilde h_A$ and $\tilde h_B$ denote the Fourier transforms of $h_A$ and $h_B$.  Since both functions are smooth and compactly supported, their Fourier transforms belong to the Schwartz space (see Appendix~\ref{app:sobolev} for the definition). It follows that, for $s > -n/2$, the product $|\vec k|^{2s}\tilde h_A(\vec k)\tilde h_B^*(\vec k)$ is locally integrable in any open neighborhood containing $\vec k=0$, and the  integral \eqref{c2} is infrared-finite.

To extract the asymptotic dependence when $|\Delta \vec x|\to \infty$, it is convenient to use the identity 
$$ e^{-i \vec k \cdot \Delta \vec x} = -i /(\hat n \cdot \Delta \vec x) \,\hat n \cdot\vec \nabla_{\vec k} e^{-i \vec k \cdot \Delta \vec x}, $$ 
where $\hat n := \Delta \vec x/|\Delta \vec x|$.  Integration by parts $2s+n$ times shows that for half-integer $s$, the asymptotic behavior of $(h_A, h_B)_s$ is dominated by a boundary contribution near $\vec k=0$, yielding
$$ (h_A, h_B)_s \sim C_{n,s} \tilde{h}_A(0)\tilde{h}_B^\star(0) |\Delta \vec x|^{-(n + 2s)}, $$
{with $C_{n,s}$ a constant.} For $n=3$, this implies the asymptotic behavior $\mathrm{Re}(h_A, h_B)_{-1/2} \sim (|\Delta \vec x|)^{-2} $ and $\mathrm{Re}(h_A, h_B)_{1/2} \sim (|\Delta \vec x|)^{-4} $ used in the main text.  For general (non-half-integer) values of $s$, one can perform integration by parts $N$ times, where $N$ is the smallest integer satisfying $N> n+2s$. The resulting integral is within the family of ``oscillatory singular'' integrals \cite{JCM-39-2}, whose asymptotic behavior is well-known \cite{JCM-39-2}, leading to the same asymptotic behavior, i.e.,  $ (h_A, h_B)_s \sim (|\Delta \vec x|)^{-(n + 2s)} $.

For $s<-1$ the tools in \cite{JCM-39-2} cannot be applied. In particular, a separate analysis is required for $s = -\tfrac{3}{2} + \mu^2$ with $\mu^2 \ll 1$. In this case, an argument analogous to that in  Proposition~\ref{prop:asymptotics_partner} applies. When $\tilde f_A$ and $\tilde f_B$ are spherically symmetric, one can set $\tilde F(k)=\tilde f_A(k)\tilde f_B^*(k)$ and  $r=|\Delta \vec x|$ in Eq.~\eqref{eq:problematic_dS}, yielding 
\begin{equation}
(h_A, h_B)_{-3/2 + \mu^2}
\sim C_{\mu}\, \tilde h_A(0)\tilde h_B^*(0)\, |\Delta \vec x|^{-2\mu^2},
\end{equation}
where $C_{\mu}$ tends to a non-zero constant in the limit $\mu\ll1$. The decay thus becomes arbitrarily slow as $\mu \to 0$. For non-spherically symmetric cases, the same reasoning applies after expanding $\tilde f_A(\vec k)\tilde f_B^*(\vec k)$ in spherical harmonics and using the expansion of $e^{-i\vec k \cdot \Delta \vec x}$ in spherical waves, together with the asymptotic behavior of the spherical Bessel functions.

Finally, some remarks are in order to clarify the generality of these results. 
\begin{enumerate}
\item We have assumed that $\tilde h_A(0)\tilde h_B^*(0) \neq 0$. If instead $\tilde h_A(0)\tilde h_B^*(0)=0$, one may expand the product near $|\vec k|=0$ as $\tilde h_A(\vec k)\tilde h_B^*(\vec k)\sim |\vec k|^m$, with $m>0$ the first non-vanishing power. The leading asymptotic behavior then becomes
\begin{equation}
(h_A,h_B)_s \sim (|\Delta \vec x|)^{-(n+2s+m)}.
\end{equation}
Thus, correlation functions decay faster in this case.  This arises if $h_A$ or $h_B$ belong to the family of ``special functions'' defined in Appendix~\ref{app:proofs_special_functions}. 

\item If the shape of the two modes is such that distance between them cannot be defined without deforming them-----e.g. in the shell configuration in Example~\ref{ex:ballshellcorrelations}---no universal asymptotic behavior exists, and one needs to analyze the configuration case by case.  In the ball--shell configuration in Example~\ref{ex:ballshellcorrelations} (see Fig.~\ref{fig:geometric_ball-shell} for a graphic representation of the geometry), we find
$$\mathrm{Re}(f_A,f_B)_{-1/2}\!\sim\! (\Delta   x)^{-1},$$
$$\mathrm{Re}(f_A,f_B)_{1/2}\!\sim\! (\Delta   x)^{-2},$$
and 
$$\mathrm{Re}(f_A,f_B)_{-3/2+\mu^2}\!\sim\! (\Delta   x)^{-2\mu^2}.$$

\item The assumption of smoothness of the functions $h_A$ and $h_B$ is made only for clarity. The argument extends to compactly supported functions belonging to appropriate homogeneous Sobolev spaces. In particular, it suffices that $f_A,f_B \in \dot H_{-1/2}(\mathbb{R}^3)$ and $g_A,g_B \in \dot H_{1/2}(\mathbb{R}^3)$, ensuring sufficient decay of their Fourier transforms when $|\vec k|\to \infty $ for all required integrations by parts.
\end{enumerate}

\section{Proofs for the family of special functions\label{app:proofs_special_functions}}

This appendix provides additional details on the class of functions referred to as ``special'' in the main text,   and explains how the proofs of the main propositions extend to this family.

This appendix is organized as follows. Section~\ref{subapp:def_special_functions} reviews the definition of special functions and identifies a subclass that requires separate treatment. Section~\ref{subapp:vNE} completes the proof of Prop.~\ref{prop:vNgeneral} by showing that it remains valid when the relevant modes belong to the special family. Section \ref{subapp:entanglement} addresses the corresponding subtleties for Prop.~\ref{prop:LN_general}.
Finally, Sec.~\ref{subapp:partners} discusses the asymptotic behavior of the partner of single-mode subsystems  associated with special functions.

\subsection{Definition and classification of special functions\label{subapp:def_special_functions}}

The set of special functions was introduced in the Remark on page \pageref{remark_IR} as functions with vanishing spatial average.
They are called ``special'' because they do not probe the infrared divergence of the Bunch–Davies vacuum in the massless limit and field-field correlations  built from them do not show the distinctive almost-scale invariant behavior. 

Formally, these functions are characterized by belonging to the homogeneous Sobolev space $f \in \dot H_{-\frac{3}{2}} (\mathbb{R}^3)$, meaning that their Sobolev norm $||f||_{-3/2 + \mu^2}$ remains finite in the limit $\mu^2 \to 0$.

To see why this condition implies that $f$ has vanishing average, consider its Sobolev norm of order $-3/2$: 
\be \label{32}  ||f||_{-\frac{3}{2}}^2 = \int \frac{d^3k}{(2\pi)^3} |\vec k|^{-3} |\tilde f(\vec k)|^2  < \infty. \ee

For this integral to converge when $|\vec k|\to0$, the Fourier transform $\tilde f(\vec k)$ must vanish in the limit $|\vec k| \to 0$. 
Since $\tilde f$ is locally integrable, this condition implies that $f$ has zero spatial average.  Conversely, if $f(\vec x)$ is in the Schwartz space and has zero average, then Eq.~\eqref{32} holds.

The proofs in this article require separate consideration when the modes $\gamma = (g,f)$ contain these special functions.

Additional subtleties arise when $\gamma$ ``saturates H\"older's inequality'', that is when  $$\mathrm{Re}(f|g)_{-\frac{1-\mu^2}{2}}^2 = ||f||_{-\frac{3}{2} + \mu^2}^2 ||g||_{\frac{1}{2}}^2.$$

{\bf Remark: } If $\gamma$ saturates H\"older's inequality, the function $f$ necessarily belongs to the special family.  Indeed, saturation implies $\tilde f (\vec k)  = \pm |\vec k|^{2-\mu^2} R^{1-\mu^2} \tilde{g}(\vec k)$. Since $g \in \dot H_{1/2}$, thus it diverges at most as $|\vec k|^{-1}$ near the origin, it follows that {$\tilde f (0)=0$}, so  $f$ belongs to the special family.

\subsection{von Neumann entropy\label{subapp:vNE}}

We now complete the proof of Prop.~\ref{prop:vNgeneral} by considering the case where the mode $\gamma$ saturates H\"older’s inequality. This situation is relevant for Cases 2 and 3 of the proof of Prop.~\ref{prop:vNgeneral}, since for such modes the coefficients $\mathfrak{b}$ and $\mathfrak{c}$ vanish, and we must verify that the remaining terms still yield $\nu_A^2 - (\nu^{\mathrm{Mink}})^2$ monotonically increasing with $H$.
For brevity, we focus on Case 2, but the same reasoning applies to Case 3.  
   
Let $\gamma_A^{(1)} = (g_A^{(1)},0)$ and $\gamma_A^{(2)} = (g_A^{(2)},f_A^{(2)})$ define a single-mode subsystem, with these functions satisfying H\"older's inequality. Hence, $\mathfrak{b}=0$ and 
 \begin{widetext}
\begin{equation}\label{eq:nuI_saturation}\begin{split}
     \nu^2_A  &=(\nu^{\mathrm{Mink}}_A)^2+ 2(RH)^{1-\mu^2}\left(||g_A^{(1)}||_{\frac{1}{2}}^2 - \mathrm{Re}(g_{A}^{(1)} , g_A^{(2)})_{\frac{1}{2}}\right) \, {R^{-1+\mu^2 }} \mathrm{Re}(f_{A}^{(2)} , g_A^{(2)})_{-\frac{1}{2} + \frac{\mu^2}{2}} \\
   &  \times(1 + \mathcal{O}(\mu^2)) + \mathcal{O}(\mu^2).
 \end{split}
 \end{equation}

 \end{widetext}

Furthermore, the saturation of Holder's inequality implies $||g_A^{(1)}||^2_{\frac{1}{2}} = \mathrm{Re}(g_A^{(1)}, g_A^{(2)})_{\frac{1}{2}}$, meaning that the leading order term in the previous equation vanishes.

Moreover, for functions in the special family, the limit $\mu^2 \to 0$ exists and is finite. In this limit, the symplectic eigenvalue of such fine-tuned single-mode subsystems in de Sitter exactly coincides with its Minkowski counterpart,
$$ \lim_{\mu^2 \to 0}\nu_A^2 = (\nu_A^{\mathrm{Mink}})^2,  $$
independently of the value of $H$.

Physically, these fine-tuned modes are insensitive to spacetime curvature.  As we will show in Sec.~\ref{sec:partners}, the same conclusion extends to the partners of these modes, whose behavior in the $\mu^2 \to 0$ limit coincides with that of their corresponding modes in Minkowski spacetime.

\subsection{Entanglement\label{subapp:entanglement}}
 While the classification of cases in the proof of Prop.~\ref{prop:LN_general} was not exhaustive, the remaining cases can be handled using the same techniques as in Cases 1 and 2. However, additional subtleties arise when the modes in both regions $A$ and $B$ involve smearing functions that saturate H\"older's inequality.  The aim of this section is to complete the proof of Prop.~\ref{prop:LN_general} by examining this case.

The modes defined from functions that saturate H\"older's inequality are curvature insensitive so that for this set of modes, the comparison between the symplectic invariants in the Bunch-Davies vacuum with those of the corresponding modes in Minkowski should be done in the massless limit $\mu^2 \to 0$. Then, following Sec.~\ref{subapp:vNE}, one can show that all symplectic invariants characterizing the two-mode system  are independent of $RH$ and their values coincide exactly with those for the corresponding modes of a massless field in the vacuum of Minkowski spacetime, that is, 
 \begin{equation}
     \lim_{\mu^2\to 0}\nu_I = \nu_I^{\mathrm{Mink}}\,, \quad I=A,B\,,
\end{equation}
and 
\begin{equation}
     \lim_{\mu^2\to 0}\tilde \nu_{\pm} = \tilde \nu_{\pm}^{\mathrm{Mink}}\,.
\end{equation}
Thus, the logarithmic negativity between two compactly supported, spacelike separated special modes saturating H\"older's inequality in the Bunch–Davies vacuum coincides with that of the corresponding modes in Minkowski spacetime in the massless limit.
Combining this with the results for generic (non-special) functions, we conclude that any two compactly supported, spacelike separated modes of a light field in the Bunch–Davies vacuum are less or equally entangled than their Minkowski counterparts.

\subsection{Asymptotic behavior of the partner\label{subapp:partners}}

Finally, we analyze the behavior of the partner mode associated with a compactly supported field mode $\gamma_A = (g_A(\vec{x}),f_A(\vec x))$ when $f_A(\vec x)$ belongs to the special family. 

 It follows from Eq.~\eqref{eq:Jgamma_dS_approx_radial}, that the leading-order behavior of the partner’s second component, denote as $f_{p}$ in the main text, is identical to that discussed in the main text; we therefore focus on its first component. Since $f_A$ is special, its Fourier transform must vanish as $|\vec k| \to 0$. To facilitate the analysis,  we write the special function $\tilde f_A$ as  $\tilde f_A(\vec k) = |\vec k|^\alpha \tilde F(\vec k)$, with $\alpha \in \mathbb{Z}_+$, and $\tilde F (\vec k) \to C \neq 0$ in the limit $|\vec k| \to 0 $. Expressing the $\sin(kr)$ in Eq.~\eqref{eq:Jgamma_dS_approx_radial} as complex exponentials and integrating by parts, one obtains integrals of the type introduced in Eq.~\eqref{eq:oscillatory_integrals}. The asymptotic behavior of integrals of this type can be obtained using standard techniques of asymptotic analysis (see, e.g.,~\cite{bender1999advanced}), from which we conclude that 
 $$(J_{DB}\gamma_A)_1 \sim \mathcal{O}(r^{-\alpha - 2\mu^2}),$$ 
 for $\alpha >0$. Hence, when $f_A$ is special, the partner mode decays with an additional power $r^{-\alpha}$ with $\alpha >0$ relative to the generic case, breaking the near scale invariance found for typical modes.

\section{Validity of Prop.~\ref{prop:vNgeneral} for arbitrary $RH$ for restricted modes}
\label{app:vNspecial}

Proposition~\ref{prop:vNgeneral} applies to any single-mode of a light field whose support is large compared to the Hubble radius, \( RH \gg 1 \)---the regime of interest in cosmology. In this appendix, we show its validity can be extended for all values of $RH$ for single-mode subsystems whose classical state space \( \Gamma_A \) admits a Darboux basis of the form \( \gamma_{A}^{(1)} = (0, -f(\vec{x})) \) and \( \gamma_{A}^{(2)} = (g(\vec{x}), 0) \), so that the corresponding observables $(\hat \Phi[f], \hat \Pi[g])$ represent a canonical pair built from pure field and momentum smeared operators. For such family of subsystems, Prop.~\ref{prop:vNgeneral} can be extended to a broader range of values of  \( RH \).

\begin{prop}\label{prop:vNrestricted}
    The von Neumann entropy of a single-mode subsystem, whose operator algebra is generated by a pure-field and a pure-momentum operator, \( (\hat{\Phi}[f], \hat{\Pi}[g]) \), with \( f \) and \( g \) chosen such that they do \emph{not} saturate H\"older's inequality, grows monotonically with $RH$ for all \( RH \geq  0 \). 
\end{prop}

\begin{proof}
For such a single-mode system, Eqs.~\eqref{phiphi}-\eqref{phipi} yield
\begin{equation}
    \nu_A^2 =(\nu^{\mathrm{Mink}}_A)^2 +  (RH)^{2-2\mu^2} \mathcal{F}_A + \mathcal{O}(\mu^2)\,, 
\end{equation}
where 
\begin{equation}\begin{split}
    \mathcal{F}_A=& K^2\big(||f||_{-3/2+\mu^2}^2||g||^2_{1/2}\\
    &- \mathrm{Re}(f,g)_{-1/2+\mu^2/2}^2\big)+ \mathcal{O}(\mu^2)\,,
\end{split}
\end{equation}
with $K^2 = \frac{2^{2-2\mu^2} \pi R^{-2+2\mu^2} }{\cos^2(\pi \mu^2) \Gamma \big(-\frac{1}{2} + \mu^2\big)^2 }>0$. Since H\"older's inequality is not saturated by assumption, $\mathcal{F_A} >0$, and thus $ \nu_A^2 >(\nu^{\mathrm{Mink}}_A)^2$ for sufficiently large $RH$.  

To assess whether this inequality holds for all $RH \geq 0$, we need to analyze the possible impact of the correction terms  $\mathcal{N}_{\Phi\Phi}$, $\mathcal{N}_{\Pi\Pi}$, and $\mathcal{N}_{\Phi\Pi}$. While these are small compared to $\mathcal{F}_A$ when $RH \sim 1$, when $RH \sim \mu^2$ $(RH)^{2-2\mu^2} K \mathcal{F}_A$ the contributions $\mathcal{N}_{\Phi\Phi}$, $\mathcal{N}_{\Pi\Pi}$, and $\mathcal{N}_{\Phi\Pi}$ could in principle compete with the leading term.
However, as shown in Appendix~\ref{app:proofprop1}, $\mathcal{N}_{\Phi\Phi}$ and $\mathcal{N}_{\Pi\Pi}$ vanish at least as $\mathcal{O}((RH)^2)$ and $\mathcal{N}_{\Phi\Pi}$ scales as $\mathcal{O}(RH)$ for $RH \ll 1$. Consequently, none of these contributions can become larger than the term  with $\mathcal{F_A}$, and the inequality $\nu^2_{A} \geq (\nu_A^{\mathrm{Mink}})^2$ remains valid also when $RH \ll 1$. 

\end{proof}

\section{Extension of Prop.~\ref{prop:LN_general} to arbitrary $RH$ for restricted modes\label{app:extension_LN_allRH}}
For a restricted class of modes---specifically, pairs of identical single-mode subsystems generated by operators $(\hat \Phi[f], \hat \Pi[g])$ with $f,g$ not belonging to the special family---we show that the result of Prop.~\ref{prop:LN_general} holds for all $RH\geq 0$, thus   extending the original result of Prop. \ref{prop:LN_general},  which was restricted to $RH\gg 1$.  These subsystems are special because they are generated by a ``pure field'' operator and a ``pure momentum'' operator—by contrast with the most general subsystem, which is generated by two operators each corresponding to a linear combination of a pure field and a pure momentum operator.

\begin{prop}
Let $A$ and $B$ denote two single-mode subsystems with supports in spacelike-separated regions, generated by  $(\hat{\Phi}[f_A], \hat{\Pi}[g_A])$ and $(\hat{\Phi}[f_B], \hat{\Pi}[g_B])$, where $f_A = f_B := f$ and $g_A = g_B := g$ (up to a spatial translation), and assume that $f$ has a non-vanishing average. Their entanglement decreases monotonically with $H$. 
\end{prop}

\begin{proof}
    In Prop.~\ref{prop:LN_general}, we established that the symplectic eigenvalues $\tilde{\nu}_{-}$ (defined in \eqref{eq:symplecticeigvalsall}) can be expressed as a Minkowski contribution plus an additional functional whose leading behavior is that of a positive polynomial in $RH$ when $RH \gg 1$. We now extend this result to all $RH\geq0$ for the aforementioned family of modes.

As discussed in Sec.~\ref{entsubs}, the LN of a two-mode system depends only on $\tilde{\nu}_{-}$, which can be written as\begin{equation}\label{eq:tildenupm_LNrestrictedproof}
        \tilde{\nu}^2 _{-} = (\nu_{-}^{\mathrm{Mink}})^2 \pm (RH)^{2-2\mu^2}\tilde{\mathcal{F}}_{-}\,,
    \end{equation}
where the expression for $\tilde{\mathcal{F}}_{-}$, for the case under consideration of two identical single-mode subsystems separated from each other, follows directly from Eq.~\eqref{eq:mathcalFtm_general}.

The Gaussian separability lemma (see Appendix~\ref{app:sobolev}) ensures that the LN vanishes whenever $\det C = \det C_{\mathrm{Mink}} + \det C_{dS} \geq 0$. Hence we focus on modes satisfying $\det C <0$ for all $RH\geq0$. This requires specific sign combinations among the Sobolev inner products. For definiteness, we choose $\mathrm{Re}(g_A|g_B)_{1/2}<0$, $\mathrm{Re}(f_A|f_B)_{-1/2}>0$, and $\mathrm{Re}(f_A|f_B)_{-3/2+\mu^2}>0$, which correspond to non-negative smearing functions. This choice is made for convenience only: any alternative sign configuration satisfying $\det C<0$ would lead to the same conclusion, with minor changes in intermediate inequalities.

For any sufficiently differentiable, compactly supported function with non-vanishing average, the Sobolev norm  $||f||_{-3/2+\mu^2}^2$ satisfies $||f||_{-3/2+\mu^2}^2 \sim \mathcal{O}(\mu^{-2})$ when $\mu^2\ll 1$, ensuring that in this regime the de Sitter correction in Eq.~\eqref{eq:tildenupm_LNrestrictedproof} dominates even when $RH \sim 1$. Thus, the result of Prop.~\ref{prop:LN_general} remains valid up to $RH \sim 1$.

To study smaller values of $RH$, it is convenient to introduce $x:=RH/\mu$ and analyze three regimes: $x\gg1$, $x\sim1$, and $x\ll1$. While the first regime follows directly from Prop.~\ref{prop:LN_general}, the cases $x\sim 1$ and $x \ll 1$ require a dedicated discussion. 

 In the regime $x \ll 1$, $\tilde{\mathcal{F}}_-$ can be approximated by 
\begin{widetext}
\begin{equation}\label{eq:mathcalFtm_approx}
    (RH)^{2-2\mu^2} \tilde{\mathcal{F}}_{-}\approx - (RH)^{2-2\mu^2} \left(\frac{2\mathcal{F}_{\tilde{\Delta}} (\tilde{\nu}_-^{\mathrm{Mink}})^2 - \mathcal{F}_+ (\nu_-^{\mathrm{Mink}})^2 - \mathcal{F}_- (\nu_+^{\mathrm{Mink}})^2 }{2\mathcal{A}} \right)  + \mathcal{O}(\mu^2) \,, 
\end{equation}
where the correction terms $\mathcal{F}_{\pm}$ are given by \begin{equation*}\begin{split}
    \mathcal{F}_{\pm}  =& K^2 \Bigg[  \left(||f||_{-\frac{3}{2}+\mu^2}^2 \pm \mathrm{Re}\left[(f_A|f_B)_{-\frac{3}{2} + \mu^2}\right]\right)\left(||g||_{\frac{1}{2}}^2 \pm  \mathrm{Re}\left[(g_A|g_B)_{\frac{1}{2} }\right]\right) \\ 
    &- \left(\mathrm{Re}\left[(f|g)_{-\frac{1}{2} + \frac{\mu^2}{2}}\right]  \pm \mathrm{Re}\left[(f_A|g_B)_{-\frac{1}{2} + \frac{\mu^2}{2}}  \right] \right)^2  \Bigg], 
\end{split}
\end{equation*}
\end{widetext}  $K =\frac{2^{2-2\mu^2} \pi R^{-2+2\mu^2} }{\cos^2(\pi \mu^2) \Gamma \big(-\frac{1}{2} + \mu^2\big)^2 } $, and     \begin{equation}
        (RH)^{2-2\mu^2}\mathcal{F}_{\tilde{\Delta}}\! =\! (RH)^{2-2\mu^2}\frac{\mathcal{F}_+ + \mathcal{F}_-}{2} - 2 \det C_{dS}. 
    \end{equation}
{Finally, $\mathcal{A}$ is defined as  \begin{equation}
    \mathcal{A} = \frac{1}{2} \sqrt{\tilde{\Delta}_{\mathrm{Mink}}^2 - \det {\sigma}_{\mathrm{Mink}}}\,,
\end{equation}
where $\tilde{\Delta}_{\mathrm{Mink}}$ and ${\sigma}_{\mathrm{Mink}}$ denote the symplectic invariants for the corresponding modes of a massless scalar field in the Minkowski vacuum. Importantly, one can show that $\mathcal{A} > 0$. }  

States that are not manifestly separable (according to the Gaussian separability lemma) satisfy the  inequalities $\mathcal{F}_- \ll \mathcal{F}_+$,   $\mathcal{F}_- < \mathcal{F}_{\Delta} <\mathcal{F}_{\tilde{\Delta}} < \mathcal{F}_+$, and  $(\tilde{\nu}_-^{\mathrm{Mink}})^2 < (\nu_-^{\mathrm{Mink}})^2 \leq (\nu_{+}^{\mathrm{Mink}})^2  $, and the  {numerator in the right hand side of }Eq.~\eqref{eq:mathcalFtm_approx} can be approximated by  
\begin{equation}\begin{split}
  &  2 \mathcal{F}_{\tilde{\Delta}} (\tilde{\nu}_-^{\mathrm{Mink}})^2 - \mathcal{F}_+ (\nu_-^{\mathrm{Mink}})^2 - \mathcal{F}_- (\nu_+^{\mathrm{Mink}})^2  \\
  &=(||f||_{-1/2}^2 - \mathrm{Re}(f_A|f_B)_{-1/2})(||g||_{1/2}^2 (2\mathcal{F}_{\tilde{\Delta}} - \mathcal{F}_+) \\ 
    & + \mathrm{Re}(g_A|g_B)_{1/2} (2\mathcal{F}_{\tilde{\Delta}} + \mathcal{F}_+) - \mathcal{F}_- (\nu_+^{\mathrm{Mink}})^2).
\end{split}
\end{equation}
Using the Cauchy-Schwartz inequality, the inequalities  between $\mathcal{F}_{\pm}$,  $\mathcal{F}_{\Delta}$, and $\mathcal{F}_{\tilde \Delta}$, and the sign conventions above, we conclude that $ 2 \mathcal{F}_{\tilde{\Delta}} (\tilde{\nu}_-^{\mathrm{Mink}})^2 - \mathcal{F}_+ (\nu_-^{\mathrm{Mink}})^2 - \mathcal{F}_- (\nu_+^{\mathrm{Mink}})^2   < 0$, implying  $\tilde{\mathcal{F}}_- >0$ also when $x \ll 1$.

Since $\tilde{\mathcal{F}}_{\pm}>0$ also for $RH\gtrsim1$ (Prop.~\ref{prop:LN_general}), the dependence on $RH$ is continuous,  and using the scaling for the subleading contributions $\mathcal{N}_{\Xi\Xi'}$ for $\Xi,\Xi'\in \{\Phi,\Pi\}$ when $RH \ll 1$ derived in Appendix~\ref{app:proofprop1},  {we conclude that $\tilde{\mathcal{F}}_- > 0$ also when $RH\sim\mu^2$.}

\end{proof}

We conclude that for two spacelike-separated, identical single-mode subsystems of the form $(\hat{\Phi}[f],\hat{\Pi}[g])$, where $f$ and $g$ do not belong to the special family, defined in Appendix~\ref{app:proofs_special_functions}, $\tilde{\mathcal{F}}_{-}$ is positive for all $RH\geq0$. Consequently, the entanglement between the two modes decreases monotonically with $H$. Additionally, the Minkowski limit is recovered smoothly as $H\to0$.

\section{Simon's normal form\label{app:simon_form}}
To facilitate the calculations of  symplectic invariants, it is convenient to write the covariance matrix of the pair of degrees of freedom in the so-called Simon's normal form {(see, e.g. Sec. 7.1.1. in~\cite{serafini2017quantum})}: 
\begin{equation}
    {\bm \sigma_{AB}'} = \left(\begin{array}{cc}
        {\bm \sigma_A'} & {\bm C'}  \\
        (\bm{C}')^T &  {\bm \sigma_B'}
    \end{array}
    \right), 
\end{equation}
where  
\begin{align}
    {\bm \sigma_A'} &= \nu_A \mathbb{I}_2, \nonumber\\
   {\bm \sigma_B' }&= \nu_B \mathbb{I}_2, \nonumber  \\
   {\bm C'} &= \mathrm{diag}(c_+,c_-)\,,
\end{align}
with $\nu_I$, $I=A,B$ being the symplectic eigenvalues of the reduced single-mode covariance matrices, and $c_+\geq c_- $ the correlations between subsystems $A$ and $B$. 

Any $4\times4$ covariance matrix ${\bm\sigma_{AB}}$ with block components ${\bm\sigma_A}$, ${\bm \sigma_B}$, and ${\bm C}$ can be brought to this normal form by a change of basis within subsystems $A$ and $B$. Such change of basis is implemented by a system-local symplectic transformation,\begin{equation}
    {\bm \sigma'_{AB}} =  S \bm \sigma_{AB}  S^{T}\,,
\end{equation}
where 
\begin{equation}
    S \!=\! (O_A \oplus O_B)(U_2^{\dagger} \oplus U_2^{\dagger}) ((L^{\dagger}_A)^{-1} \!\oplus \! (L^{\dagger}_B)^{-1}).
\end{equation}
Here, $O_I$ ($I=A,B$) are orthogonal transformations, $L_I$ denotes the linear maps that diagonalize $\Omega_2\sigma_I$, and 
\begin{equation}
    U_2 = \frac{1}{\sqrt{2}} \left(\begin{matrix}
        1 & i \\ 
        1 & - i 
    \end{matrix}  \right) \,.
\end{equation}
The coefficients $c_{\pm}$ are obtained as\begin{equation} \begin{split}
    c_{\pm} &= \frac{1}{2} \Big(\sqrt{(d_{11} + d_{22})^2 + (d_{12} - d_{21})^2 } \\ 
    &\pm  \sqrt{(d_{11} - d_{22})^2 + (d_{12} + d_{21})^2 } \Big) \,,
\end{split}
\end{equation}
with 
\begin{equation} \begin{split}
     d_{11} &= \frac{1}{\sqrt{\sigma_{A}^{11}\sigma_{B}^{11}}} \frac{1}{\sqrt{\nu_A\nu_B}} \Big(\sigma_A^{11} (\sigma_B^{11} C^{22} -\sigma_B^{12}C^{12} ) \\ 
     &+ \sigma_A^{12} (\sigma_B^{12}C^{11} - \sigma_{B}^{11} C^{12})\Big) \,,
\end{split}
\end{equation}
\begin{equation}
    d_{12} = \frac{1}{\sqrt{\sigma_{A}^{11}\sigma_{B}^{11}}} \sqrt{\frac{\nu_B}{\nu_A}} \left(\sigma_A^{12}C^{11} - \sigma_A^{11}C^{12} \right) \,,
\end{equation} \begin{equation}
    d_{21} = \frac{1}{\sqrt{\sigma_{A}^{11}\sigma_{B}^{11}}} \sqrt{\frac{\nu_A}{\nu_B}} (\sigma_B^{12} C^{11} - \sigma_B^{11}C^{12})\,,
\end{equation}
and 
\begin{equation}
    d_{22} = \sqrt{\frac{\nu_A\nu_B}{\sigma_{A}^{11}\sigma_{B}^{11}}} C^{11}\,.
\end{equation}

In Simon's normal form, the global symplectic invariants and, in particular, the symplectic eigenvalues take compact analytical expressions:
\begin{equation}
    \det \sigma = (\nu_A\nu_B - c_+^2)(\nu_A\nu_B - c_-^2)\,,
\end{equation}
\begin{widetext}
    \begin{equation}\label{eq:nupm_simon}
    \nu_{\pm}^2  =\frac{1}{2} \left(\nu_A^2 + \nu_B^2 + 2 c_-c_+ \pm \sqrt{(\nu_A^2 -\nu_B^2)^2 + 4 c_+ c_- (\nu_A^2 + \nu_B^2) +4 \nu_A\nu_B (c_+^2 + c_-^2) } \right) \,,
\end{equation}
and
\begin{equation}\label{eq:nuPTpm_simon}
    \tilde{\nu}_{\pm}^2  =\frac{1}{2} \left(\nu_A^2 + \nu_B^2 - 2 c_-c_+ \pm \sqrt{(\nu_A^2 -\nu_B^2)^2 - 4 c_+ c_- (\nu_A^2 + \nu_B^2) +4 \nu_A\nu_B (c_+^2 + c_-^2) } \right) \,,
\end{equation}
\end{widetext}
where $\tilde{\nu}_{\pm}$ denote the symplectic eigenvalues of the partially transposed covariance matrix.

Simon’s normal form is particularly useful for evaluating the mutual information and logarithmic negativity. Although the explicit construction above is not strictly required for our proofs, it greatly simplifies subsequent expressions—especially when the reduced covariance matrices of subsystems $A$ and $B$ share the same symplectic eigenvalues.

\bibliography{dSlong.bib} 
\end{document}